\documentclass{lmcs}
\pdfoutput=1

\usepackage{lastpage}
\lmcsdoi{15}{3}{27}
\lmcsheading{}{\pageref{LastPage}}{}{}%
{Jun.~26,~2018}{Sep.~04,~2019}{}

\usepackage{amsfonts}
\usepackage{amssymb,amsmath}
\usepackage{ifthen}
\usepackage{bm}
\usepackage{amsthm}
\usepackage{latexsym}
\usepackage{mathdots}
\usepackage{tikz}
\usepackage{xspace}

\usepackage[shortcuts]{extdash}

\usepackage{url}
\usepackage{graphicx}
\usepackage{xcolor}

\newcommand{\toppref}{\tau_{\text{pref}}}
\newcommand{\metricpref}{d_{\text{pref}}}
\newcommand{\topdisc}{\lambda}
\newcommand{\complemento}{^\cp}
\newcommand{\preftree}{\subseteq}
\newcommand{\strictpreftree}{\subset}

\newcommand{\palt}[1]{\mathbb{P}_{#1}}

\theoremstyle{plain}
\newtheorem{claim}[thm]{Claim}

\theoremstyle{thmC}
\newtheorem{factC}[thm]{Fact}

\newcommand{\mathnm}[1]{#1}

\newcommand{\ident}[1]{\ensuremath{\texorpdfstring{\mathrm{\mathnm{#1}}}{#1}}\xspace}

\newcommand{\domain}[2]{\newcommand{#1}{\ensuremath{\mathbb {#2}}}\xspace}

\newcommand{\wadge}{<_{\mathrm{\mathnm{W}}}}
\newcommand{\wadgeq}{\leq_{\mathrm{\mathnm{W}}}}

\newcommand{\wadgee}{\equiv_{\mathrm{\mathnm{W}}}}

\newcommand{\concat}{\text{\^{}}}

\newcommand{\mysc}[1]{{\sc {\texorpdfstring{\mathnm{#1}}{#1}}}\xspace}

\newcommand{\mso}{\mysc{mso}}

\newcommand{\msou}{\mysc{mso+u}}

\newcommand{\EXPTIME}{\mysc{EXPTime}}

\domain{\A}{A}
\domain{\B}{B}
\domain{\C}{C}
\domain{\D}{D}
\domain{\E}{E}
\domain{\N}{N}
\domain{\Z}{Z}
\domain{\Q}{Q}
\domain{\R}{R}
\domain{\Pri}{P}

\newcommand{\boldclass}[3]{\ensuremath{\mathbf{#1}^{#2}_{#3}}}

\newcommand{\borel}{\ensuremath{\mathcal B}\xspace}

\newcommand{\bsigma}[1]{\boldclass{\Sigma}{0}{#1}}
\newcommand{\bpi}[1]{\boldclass{\Pi}{0}{#1}}
\newcommand{\bdelta}[1]{\boldclass{\Delta}{0}{#1}}

\newcommand{\asigma}[1]{\boldclass{\Sigma}{1}{#1}}

\newcommand{\bc}[1]{\mathcal{BC}({#1})} 

\newcommand{\eqdef}{\stackrel{\text{def}}=}

\renewcommand{\mod}[1]{\allowbreak\mkern10mu({\operator@font mod}\,\,#1)} 

\newcommand{\mathcalsym}[1]{\ensuremath{\mathcal{#1}}\xspace}

\newcommand{\w}{\omega}

\newcommand{\Aa}{\mathcalsym{A}}

\newcommand{\Cc}{\mathcalsym{C}}
\newcommand{\Dd}{\mathcalsym{D}}

\newcommand{\Gg}{\mathcalsym{G}}

\newcommand{\Ll}{\mathcalsym{L}}

\newcommand{\comment}[1]{}

\newcommand{\fun}[3]{\ensuremath{#1\colon #2 \to #3}}
\newcommand{\parfun}[3]{\ensuremath{#1\colon #2 \rightharpoonup #3}}

\newcommand{\dom}{\ident{dom}}

\usetikzlibrary{positioning}
\usetikzlibrary{decorations.pathreplacing}
\usetikzlibrary{automata,positioning}
\usetikzlibrary{decorations.pathmorphing}
\usetikzlibrary{decorations.markings}
\usetikzlibrary{decorations}
\usetikzlibrary{arrows}
\usetikzlibrary{arrows.meta}
\usetikzlibrary{patterns}
\usetikzlibrary{calc}
\usetikzlibrary{shapes}
\usetikzlibrary{fadings,	shadings}

\tikzstyle{obrace} = [draw, thick, decoration={brace, raise=0.0cm}, decorate,
    every node/.style={anchor=south, yshift= 0.1cm}]
\tikzstyle{lbrace} = [draw, thick, decoration={brace, raise=0.0cm}, decorate,
    every node/.style={anchor=east, xshift=-0.1cm}]
\tikzstyle{ubrace} = [draw, thick, decoration={brace, mirror, raise=0.0cm}, decorate,
    every node/.style={anchor=north, yshift=-0.1cm}]
\tikzstyle{rbrace} = [draw, thick, decoration={brace, mirror, raise=0.0cm}, decorate,
    every node/.style={anchor=west, xshift= 0.1cm}]

\tikzstyle{toL} = [anchor=mid east, flow]
\tikzstyle{toR} = [anchor=mid west, flow]
\tikzstyle{toC} = [align=center, anchor=mid, flow]

\tikzstyle{coord} = [draw,->,very thick]
\tikzstyle{dot} = [draw,shape=circle,fill, minimum size=1mm, inner sep=0pt,outer sep=0pt]
\tikzstyle{flow}  = [inner sep=0cm, node distance=0cm and 0cm]
\tikzstyle{edge}=[draw, thick,->]

\tikzstyle{trnode}=[draw, circle, minimum size=15pt, inner sep=0pt, outer sep=0pt]
\tikzstyle{trport}=[draw, rectangle, minimum size=12pt, inner sep=0pt, outer sep=0pt]
\tikzstyle{tredge}=[draw, -{Latex}]
\tikzstyle{trdots}=[draw, thick, loosely dotted]
\tikzstyle{trlole}=[out=-100, in=-135, loop, looseness=14, -{Latex}]
\tikzstyle{trlore}=[out=-80, in=-45, loop, looseness=14, -{Latex}]

\newcommand{\tikzEvalInt}[2]{\pgfmathparse{int(#2)}{\global\edef#1{\pgfmathresult}}} 
\newcommand{\tikzEvalFloat}[2]{\pgfmathparse{#2}{\global\edef#1{\pgfmathresult}}}

\newcommand{\dL}{\mathtt{{\mathnm{\scriptstyle L}}}} 
\newcommand{\dR}{\mathtt{{\mathnm{\scriptstyle R}}}} 

\newcommand{\hole}{\Box}

\newcommand{\cp}{\mathrm{\mathnm{c}}}

\newcommand{\lang}{\ident{L}}
\newcommand{\val}{\mathsf{val}\xspace}
\newcommand{\dist}{\mathsf{dist}\xspace}

\newcommand{\gameHinfty}{\gameH^\infty}

\newcommand{\one}{1}

\newcommand{\smat}{\ident{matrix}}

\newcommand{\restr}{{\upharpoonright}}

\newcommand{\new}[1]{#1}

\newlength{\hatchspread}
\newlength{\hatchthickness}
\newcommand{\hatchcolor}{}
\tikzset{hatchspread/.code={\setlength{\hatchspread}{#1}},
         hatchthickness/.code={\setlength{\hatchthickness}{#1}},
        hatchcolor/.code={\renewcommand{\hatchcolor}{#1}}}
\tikzset{hatchspread=3pt,
         hatchthickness=0.4pt,
         hatchcolor=black}
\pgfdeclarepatternformonly[\hatchspread,\hatchthickness]
   {custom north west lines}
   {\pgfqpoint{-2\hatchthickness}{-2\hatchthickness}}
   {\pgfqpoint{\dimexpr\hatchspread+2\hatchthickness}{\dimexpr\hatchspread+2\hatchthickness}}
   {\pgfqpoint{\hatchspread}{\hatchspread}}
   {
    \pgfsetlinewidth{\hatchthickness}
    \pgfpathmoveto{\pgfqpoint{0pt}{\hatchspread}}
    \pgfpathlineto{\pgfqpoint{\dimexpr\hatchspread+0.15pt}{-0.15pt}}
    \pgfsetstrokecolor{\hatchcolor}
    \pgfusepath{stroke}
   }

\newcommand{\trees}{\ident{Tr}}
\newcommand{\when}{\mathsf{when}}

\newcommand{\init}{\ident{\,I}}

\newcommand{\ndiff}[2]{\Dd_{#1}\big(\bsigma{{#2}}\big)}

\newcommand{\gameH}{\mathcal{H}}
\newcommand{\gameHs}{\gameH^{\in, \notin}}
\newcommand{\gameV}{\mathcal{V}}

\newcommand{\boolcomb}[1]{\mathsf{BC}(\bsigma{#1})} 
\newcommand{\monoid}{\sharp}
\newcommand{\Context}{\mathsf{Ct}_A}

\newcommand{\lunghezza}{lt}
\newcommand{\newdef}{type-labelled tree\xspace}

\newcommand{\change}{\mathsf{change}}

\newcommand{\Teq}{\overset{{\scriptscriptstyle\mathrm{Tr}}}{\approx}}
\newcommand{\Ceq}{\overset{{\scriptscriptstyle\mathrm{Ct}}}{\approx}}

\title[Regular tree languages in low levels of the Wadge hierarchy]{Regular tree languages in low levels\texorpdfstring{\\}{} of the Wadge hierarchy}
\titlecomment{{\lsuper*}This is an extended version of~\cite{bojanczyk_boolean} and contains a~correction of a~proof from~\cite{michalewski_delta}.}
\thanks{The first, second, and fourth authors were supported by Polish National Science Centre grant no. 2016/21/D/ST6/00491.}

\author[M.~Boja{\'n}czyk]{Miko{\l}aj Boja{\'n}czyk\rsuper{a}}
\address{\lsuper{a}Institute of Informatics, University of Warsaw}
\email{\{bojan,mskrzypczak\}@mimuw.edu.pl}
\author[F.~Cavallari]{Filippo Cavallari\rsuper{b}}
\address{\lsuper{b}University of Turin and University of Lausanne}
\email{filippo.cavallari@unito.it}

\author[T.~Place]{\texorpdfstring{\\}{}Thomas Place\rsuper{c}}
\address{\lsuper{c}LaBRI, Bordeaux University}
\email{thomas.place@labri.fr}

\author[M.~Skrzypczak]{Micha{\l} Skrzypczak\rsuper{a}}


\begin{document}

\begin{abstract}
In this article we provide effective characterisations of regular languages of infinite trees that belong to the low levels of the Wadge hierarchy. More precisely we prove decidability for each of the finite levels of the hierarchy; for the class of the Boolean combinations of open sets $\boolcomb{1}$ (i.e.~the union of the first $\omega$ levels); and for the Borel class~$\bdelta{2}$ (i.e.~for the union of the first $\omega_1$ levels).
\end{abstract}

\maketitle

\section{Introduction}%
\label{sec:introduction}

The space of all infinite trees over a~finite alphabet is homeomorphic to the Cantor space. Therefore, it makes sense to ask if a~language of infinite trees --- in our setting, we are interested in regular ones --- is for instance open, Borel, or of specific descriptive set theoretical complexity. As witnessed by a~number of conjectures and results~\cite{skurczynski_borel_infinite,murlak_wadge,murlak_final_game,walukiewicz_buchi,cavallari_gdelta}, topologically defined classes often have natural automata counterparts. For instance, in the case of $\omega$\=/words, the structure of parity \emph{deterministic} automata (defined in terms of Wagner hierarchy) is strictly connected to the Wadge hierarchy, see~\cite{wagner_hierarchy}. In the case of regular tree languages that are Borel, there is a~strong connection between the \emph{Borel rank} and priorities used by \emph{weak alternating automata} (see~\cite{duparc_murlak_operations} and~\cite{cavallari_gdelta}).

Algorithms that determine if a~regular language belongs to a~subclass~$\Ll$ of regular languages are known as \emph{effective characterisations}. Typically, an~effective characterisation comes with a~structural description of automata (or algebras) that recognise languages from~$\Ll$. The seminal example is Sch\"utzenberger's Theorem~\cite{schutzenberger_fo}, which says that a~regular language is star\=/free if and only if it is recognised by an~aperiodic semigroup. For other examples about finite words, see the survey~\cite{Place:2015ke}, which discusses effective characterisations for the low levels of the quantifier alternation hierarchy. For examples on $\omega$\=/words, including topologically motivated ones, see~\cite{perrin_pin_words}. For examples on finite trees, see e.g.~\cite{bojanczyk_forest} or a~survey~\cite{bojanczyk_tree_algebra}.

Most of the known effective characterisations speak about languages of words (finite or infinite) or finite trees. The case of regular languages of infinite trees seems to be much more difficult, mainly because of the inherent non\=/determinism needed to recognise these languages~\cite{blumensath_recognisability,bilkowski_unambiguity}. Thus, the known examples of effective characterisations are usually limited either to simple classes of sets (e.g.~open sets~\cite{KusWil02,igwwollic}) or to restricted classes of languages given as the input (e.g.~recognised by deterministic automata~\cite{murlak_phd}).

In this work we focus on the very low levels of the Borel hierarchy: the~class~$\boolcomb{1}$ of Boolean combinations of open sets; and the self\=/dual class~$\bdelta{2}$ at the second level of the hierarchy. We use algebraic methods for infinite trees, i.e.~our characterisations are defined by equations which must be satisfied by the syntactic algebra of a~language.


This paper continues a~line of work aimed at understanding the algebraic theory of regular languages of infinite trees~\cite{blumensath_recognisability,bilkowski_unambiguity,bojanczyk_ef_trees}. The obtained results show that even simple algebras (i.e.~not strong enough to distinguish all regular languages or not \emph{complete}, see Subsection~\ref{ssec:limitations}) can be adequate for characterising classes of languages that are sufficiently simple. This opens the possibility that a~bit more complex algebraic structure (but a~priori not complete) might be enough for the successive levels of the Borel hierarchy, like~$\bdelta{3}$.

\paragraph*{\bf Contribution of this article}

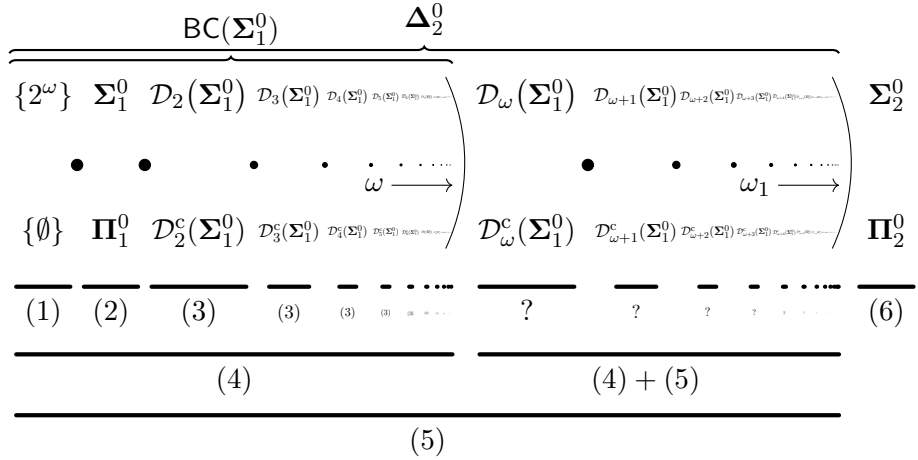
\begin{figure}
\centering
\begin{tikzpicture}[scale=0.9]
\tikzstyle{decres}=[draw, ultra thick, cap=round]
\tikzstyle{declab}=[toC]

\tikzEvalFloat{\decshi}{0.8}

\tikzEvalFloat{\yS}{1.0}
\tikzEvalFloat{\yM}{0.0}
\tikzEvalFloat{\ss}{0.5}

\tikzEvalFloat{\x}{0}
\node[toC] at (\x,+\yS) {$\{2^\w\}$};
\node[toC] at (\x,-\yS) {$\{\emptyset\}$};

\draw[decres] (\x-0.4, -\yS-\decshi) -- ++(0.8, 0.0);
\node[declab] at (\x, -\yS-\decshi-0.4) {$(1)$};

\tikzEvalFloat{\x}{\x+\ss}
\node[dot, scale=1.5] at (\x,\yM) {};
\tikzEvalFloat{\x}{\x+\ss}
\node[toC] at (\x,+\yS) {$\bsigma{1}$};
\node[toC] at (\x,-\yS) {$\bpi{1}$};

\draw[decres] (\x-0.4, -\yS-\decshi) -- ++(0.8, 0.0);
\node[declab] at (\x, -\yS-\decshi-0.4) {$(2)$};

\tikzEvalFloat{\x}{\x+\ss}

\tikzEvalFloat{\scc}{1}

\foreach \i in {0,...,9} {
	\node[dot, scale=\scc*1.5] at (\x,\yM) {};
	\tikzEvalFloat{\x}{\x+\scc*\ss*1.6}
	\tikzEvalInt{\ind}{\i+2}
	\node[toC, scale=\scc] at (\x,+\yS) {$\ndiff{\ind}{1}$};
	\node[toC, scale=\scc] at (\x,-\yS) {${\mathcal{D}_{\ind}\complemento}(\bsigma{1})$};

	\draw[decres] (\x-0.7*\scc*\scc, -\yS-\decshi) -- ++(1.4*\scc*\scc, 0.0);
	\node[declab, scale=\scc] at (\x, -\yS-\decshi-0.4) {$(3)$};

	\tikzEvalFloat{\scc}{\scc*0.65}
	\tikzEvalFloat{\x}{\x+\scc*\ss*2.5}
}

\tikzEvalFloat{\afterBC}{\x}

\tikzEvalFloat{\scc}{1}

\tikzEvalFloat{\x}{\x+\ss*2.2}
\node[toC, scale=\scc] at (\x,+\yS) {$\ndiff{\omega}{1}$};
\node[toC, scale=\scc] at (\x,-\yS) {${\mathcal{D}_{\omega}\complemento}(\bsigma{1})$};

\draw[decres] (\x-0.7*\scc*\scc, -\yS-\decshi) -- ++(1.4*\scc*\scc, 0.0);
\node[declab, scale=\scc] at (\x, -\yS-\decshi-0.4) {$?$};

\tikzEvalFloat{\x}{\x+\ss*1.8}

\tikzEvalFloat{\scc}{\scc*0.65}

\foreach \i in {0,...,11} {
	\node[dot, scale=\scc*1.5/0.65] at (\x,\yM) {};
	\tikzEvalFloat{\x}{\x+\scc*\ss*2.2}
	\tikzEvalInt{\ind}{\i+1}
	\node[toC, scale=\scc] at (\x,+\yS) {$\ndiff{\omega+\ind}{1}$};
	\node[toC, scale=\scc] at (\x,-\yS) {${\mathcal{D}_{\omega+\ind}\complemento}(\bsigma{1})$};

	\draw[decres] (\x-0.7*\scc*\scc, -\yS-\decshi) -- ++(1.4*\scc*\scc, 0.0);
	\node[declab, scale=\scc] at (\x, -\yS-\decshi-0.4) {$?$};
	\tikzEvalFloat{\scc}{\scc*0.65}
	\tikzEvalFloat{\x}{\x+\scc*\ss*2.8}
}

\tikzEvalFloat{\afterDL}{\x}

\tikzEvalFloat{\x}{\x+\ss+0.2}
\node[toC] at (\x,+\yS) {$\bsigma{2}$};
\node[toC] at (\x,-\yS) {$\bpi{2}$};

\draw[decres] (\x-0.4, -\yS-\decshi) -- ++(0.8, 0.0);
\node[declab] at (\x, -\yS-\decshi-0.4) {$(6)$};


\draw[decres] (-0.4, -\yS-\decshi-1.0) -- (\afterBC, -\yS-\decshi-1.0);
\node[declab] at (-0.2+\afterBC*0.5, -\yS-\decshi-1.4) {$(4)$};

\draw[decres] (\afterBC+0.4, -\yS-\decshi-1.0) -- (\afterDL, -\yS-\decshi-1.0);
\node[declab] at (\afterBC*0.5+\afterDL*0.5, -\yS-\decshi-1.4) {$(4)+(5)$};

\draw[decres] (-0.4, -\yS-\decshi-1.9) -- (\afterDL, -\yS-\decshi-1.9);
\node[declab] at (-0.2+\afterDL*0.5, -\yS-\decshi-2.3) {$(5)$};

\draw (-0.5, \yS+0.5) edge[obrace] node {$\boolcomb{1}$} (\afterBC, \yS+0.5);
\draw (-0.5, \yS+0.65) edge[obrace] node {$\bdelta{2}$} (\afterDL, \yS+0.65);

\draw (\afterBC-0.1, +\yS+0.3) arc (25:-25:\yS*3);
\draw (\afterDL-0.1, +\yS+0.3) arc (25:-25:\yS*3);

\node[toL] at (\afterBC-1.0, -0.3) {$\omega$};
\draw[->] (\afterBC-0.9, -0.3) -- ++(0.9, 0.0);

\node[toL] at (\afterDL-1.0, -0.3) {$\omega_1$};
\draw[->] (\afterDL-0.9, -0.3) -- ++(0.9, 0.0);

\end{tikzpicture}
\caption{The first $\omega_1+1$ levels of the Wadge hierarchy for infinite trees; together with their decidability results. The self\=/dual classes are depicted by $\bullet$, they are formed by the intersection of the two consecutive non self\=/dual classes.}%
\label{fig:wadge-precise}
\end{figure}

Consider a~class $\Gamma$ of languages (e.g.~the class of open sets~$\bsigma{1}$). Then, an~\emph{effective characterisation} of $\Gamma$ is an algorithm for the following decision problem:

\begin{prob}[Effective Characterisation of $\Gamma$]%
\label{decidabilityproblem}
Given a~representation of a~regular language~$L$, decide if $L \in \Gamma$.
\end{prob}

As explained above, there are multiple results providing effective characterisations for various classes of languages. In this article we focus on the classes of the Wadge hierarchy inside~$\bdelta{2}$, see Figure~\ref{fig:wadge-precise} (the relevant technical notions are introduced in Section~\ref{sec:prelim}). Notice that, since regular languages are effectively closed under complement, an effective characterisation of $\Gamma$ provides at the same time an effective characterisation of $\Gamma\complemento$ and $\Gamma\cap\Gamma\complemento$. Thus, we will focus on non self\=/dual classes on one \emph{side} of the hierarchy. Since $\trees_A$ is homeomorphic to $2^\omega$ (see Fact~\ref{ft:trees-are-Cantor}) all the results from Subsections~\ref{ssec:difference} and~\ref{ssec:wadge} apply. Therefore, the Wadge hierarchy over $\trees_A$ introduces the following classes of sets:
\begin{enumerate}
\item The class $\{\trees_A\}$. A~language $L$ belongs to $\{\trees_A\}$ if and only if $L=\trees_A$, thus solving the effective characterisation for that class boils down to checking universality of $L$, which reduces to non\=/emptiness of the complement of $L$~\cite{rabin_s2s}.
\item The class of open sets $\bsigma{1}$. That characterisation follows from~\cite{KusWil02,igwwollic}.
\item The classes of the difference hierarchy $\ndiff{n}{1}$ for $2\leq n<\omega$. These classes are characterised in this paper, see Theorem~\ref{thm:decidforsinglefinitediff}.
\item The class $\boolcomb{1}=\bigcup_{n\in\omega}\ndiff{n}{1}$ of Boolean combinations of open sets. This is the main contribution of this paper, see Theorem~\ref{thm:mainforBC1}.
\item The self\=/dual class $\bdelta{2}=\bigcup_{\xi<\omega_1}\ndiff{\xi}{1}$ of the Borel hierarchy. This result was claimed in~\cite{michalewski_delta}, however the arguments there contain a~flow, see discussion in Section~\ref{sec:delta-case}. In this paper we provide a~complete argument, see Theorem~\ref{thm:delta-char}.
\item The class $\bsigma{2}$ from the second level of the Borel hierarchy. This class seems to be out of reach of the algebras considered in this paper, see Subsection~\ref{ssec:limitations}. However, an effective characterisation for that class exists, see~\cite{cavallari_gdelta}.
\end{enumerate}

\noindent
What remains open is how to characterise the specific classes $\ndiff{\xi}{1}$ for $\omega\leq \xi<\omega_1$. Notice that there are only countably many regular languages and therefore there must exist $\xi_0<\omega_1$ such that no regular language belongs to $\ndiff{\xi}{1}$ for $\xi\geq \xi_0$. However, the value of $\xi_0$ is not known. Duparc and Murlak~\cite{duparc_murlak_operations} have proved that there exist regular languages in any Wadge degree with Wadge rank less than $\omega^\omega$ (i.e.~$\xi_0\geq\omega^\omega$). We do not know if $\xi_0=\omega^\omega$, even if this is a~quite reasonable conjecture.

\paragraph*{\bf Related work}

First, a~series of works~\cite{niwinski_deterministic,niwinski_gap,murlak_wadge,murlak_final_game} provide effective characterisations for almost all natural classes when the input is restricted to deterministic automata or their dualised variant --- the so\=/called game automata. These results are based on the \emph{pattern method} saying that the language recognised by a~deterministic automaton is complex if and only if the automaton itself contains a~\emph{complex pattern}. Unfortunately, there is no known method allowing to extend these techniques to languages involving non\=/determinism.

Recently, certain new techniques have been developed that show how to deal with the non\=/determinism of regular languages of infinite trees. The first result of this kind is the reduction of the general Rabin\=/Mostowski index problem to a~certain boundedness problem for cost automata~\cite{loding_index_to_bounds}. Unfortunately, the latter problem is not known to be decidable. However, the game approach used in the above reduction turned out to work for the lowest indices~\cite{colcombet_weak}. By adopting these techniques, the authors of~\cite{walukiewicz_buchi} provided a~characterisation of Borel sets among languages recognisable by B\"uchi automata. A~similar approach used in~\cite{cavallari_gdelta} provided an~effective characterisation of the Borel class $\bpi{2}$ among all regular tree languages. An effective (but not algebraic) characterisation of the class $\bdelta{2}$ follows directly from that result, however it does not solve the more difficult case of~$\boolcomb{1}$.

The paper is based on the conference papers~\cite{bojanczyk_boolean} and~\cite{michalewski_delta}, see Conclusions for a~discussion on relations between the new paper and the original ones.

\paragraph*{\bf Structure} The paper is structured as follows. In Section~\ref{sec:prelim} we recall some basic notions about words and trees and we set the notation used throughout the article, \new{with a~special emphasis on the topological hierarchies involved}. In Section~\ref{sec:finite-game} we describe a~topological game that will be used to obtain effective characterisations of the levels of the Wadge hierarchy up to $\omega$. \new{Subsection~\ref{ssec:finite-levels} provides an~easy application of the game to prove decidability of each of the first $\omega$ levels of the Wadge hierarchy (i.e.~all the finite levels). Subsection~\ref{ssec:inf-game} extends the game to an~infinite variant used later to characterise $\bdelta{2}$.} In Section~\ref{sec:algebra} we present the algebraic structure used in this paper to represent regular languages of infinite trees. In Section~\ref{sec:main-thm} we state Theorem~\ref{thm:mainforBC1} characterising the class~$\boolcomb{1}$ in terms of equations defined in the syntactic algebra. The proof of this theorem is spread across Subsections~\ref{ssec:three-to-four},~\ref{ssec:four-to-three},~\ref{ssec:limit-bounded} and~\ref{ssec:limit-unbounded}. Finally, in Section~\ref{sec:delta-case} we state and prove Theorem~\ref{thm:delta-char} that uses the algebraic tools from Theorem~\ref{thm:mainforBC1} to characterise the Borel class~$\bdelta{2}$.

\section{Basic notions}%
\label{sec:prelim}

If $f$ is a~function, by $\dom(f)$ we denote the domain of $f$. We denote by $\omega$ the first infinite ordinal and by $\omega_1$ the first uncountable ordinal.

\subsection{Words and trees}

Consider a~non\=/empty set $A$. We call $A$ an~\emph{alphabet} if $A$ is finite. Let $A^n$ be the space of the functions of the form $\fun{s}{\{0,\ldots,n{-}1\}}{A}$. Such a~function can be represented as a~word
$s=(s(0),\ldots, s(n{-}1))=s_0\cdots s_{n-1}$ over $A$. If $s=s_0s_1\cdots s_{n-1}$ then we say that $n$ is the \emph{length} of $s$, and we denote it by $\lunghezza(s)$. The empty word is denoted by $\epsilon$, i.e.~$\lunghezza(\epsilon)=0$ and $A^0=\{\epsilon\}$. By $A^{\leq n}$ we denote the set of words over $A$ of length at most $n$, i.e.~$A^{\leq n} \eqdef A^0 \cup A^1 \cup \cdots \cup A^n$. We denote by $A^\ast$ the set of all the finite words over $A$: $A^\ast \eqdef \bigcup_{n \in \omega} A^n$. By $A^\omega$ we denote the set of infinite words over the alphabet $A$, formally the elements of this space are functions of the form $\fun{\alpha}{\omega}{A}$. Such a~function can be represented as an~infinite sequence
$(\alpha(0),\alpha(1),\alpha(2),\ldots)  = \alpha_0\alpha_1\alpha_2 \cdots$ Finally, we set $A^{\leq\w} \eqdef A^\ast \cup A^\omega$.

If $\alpha \in A^{\leq\omega}$ and $n \in \omega$, we define $\alpha\restr n\eqdef \alpha_0 \alpha_1\cdots \alpha_{n-1} \in A^n$ (if $\alpha$ is finite this definition makes sense only if $n \leq \lunghezza(\alpha)$). We say that $s \in A^\ast$ is a~\emph{prefix} of $\alpha \in A^{\leq\omega}$ if $s = \alpha\restr n$ for some $n$; in symbols $s \preceq \alpha$. We write $s\prec \alpha$ if $s \preceq \alpha$ but $s \neq \alpha$. The \emph{concatenation} of $s,t \in A^\ast$, where $s=s_0 \cdots s_{n-1}$ and $t=t_0 \cdots t_{m-1}$, is the word $s\concat t= st= s_0 \cdots s_{n-1}t_0 \cdots t_{m-1}$. We can also consider the concatenation $s\concat \alpha$ of a~finite word $s$ and an~infinite word $\alpha$ defined in the obvious way: $s\concat \alpha = s_0 s_1 s_2 \cdots s_{\lunghezza(s)-1}\alpha_0\alpha_1\cdots$

Now let us generalise these notions to trees. In this article, we focus on trees with binary branching, where the two \emph{directions} are \emph{left} $\dL$ and \emph{right} $\dR$. A~\emph{partial tree} over an~alphabet $A$ is a~partial function $\parfun{t}{{\{\dL,\dR\}}^\ast}{A}$ with a~non\=/empty prefix\=/closed domain $\dom(t)$ (i.e.~if $s \in \dom(t)$ and $s' \preceq s$ then $s' \in \dom(t)$). A~node $u\in\dom(t)$ is either an~\emph{internal node} (i.e.~both $u\concat\dL$ and $u\concat\dR$ belong to $\dom(t)$), a~\emph{unary node} (i.e.~exactly one of $u\concat\dL$ and $u\concat\dR$ belongs to $\dom(t)$), or a~\emph{leaf} (i.e.~none of $u\concat \dL$ and $u \concat \dR$ belongs to $\dom(t)$). For the sake of readability, we write $u\in t$ to denote that $u\in\dom(t)$ is a~node of $t$. The empty sequence~$\epsilon$ belongs to every partial tree and it is called the \emph{root} of a~tree. A~\emph{branch} of a~partial tree $t$ is a~word $\pi$ such that $\pi\restr n \in t $ for any $n \leq \lunghezza(\pi)$ if $\pi$ is finite (resp.\ for any $n \in \omega$ if $\pi$ is infinite). An~infinite branch of a~partial tree $t$ is called a~\emph{path} of $t$. A~node $u$ is \emph{on a~branch} (finite or infinite) $\pi$ if it is a~prefix of $\pi$, i.e.~if $u \preceq \pi$. A~partial tree $t$ is \emph{finite} if its domain $\dom(t)$ is finite. A~\emph{tree} is a~partial tree $t$ with $\dom(t) = {\{\dL, \dR\}}^\ast$ (i.e.~a~\emph{complete tree}). The set of all trees over an~alphabet $A$ is denoted $\trees_A$, that is $\trees_A \eqdef\{ t \mid \fun{t}{{\{ \dL, \dR\}}^\ast}{A}\}$.

If $p$ and $t$ are partial trees, we say that $p$ is a~\emph{prefix} of $t$, and we denote it by $p \preftree t$, if $\dom(p) \subseteq \dom(t)$ and $p(u)=t(u)$ for any $u \in \dom(p)$. We write $p \strictpreftree t$ if $p \preftree t$ but $p \neq t$. Finally, if $t$ is a~partial tree and $u$ is a~node of $t$, by $t.u$ we indicate the partial tree $t$ \emph{truncated} in $u$: for any $w$ such that $uw \in \dom(t)$, we have $t.u(w)=t(uw)$.

\new{
For $d>0$ a \emph{$d$\=/prefix} of a~tree $t$ is the prefix $p\eqdef t\restr{\{\dL,\dR\}}^{< d}$, i.e.~the prefix of $t$ containing all the nodes at depths smaller than $d$. For instance, the $1$\=/prefix of $t$ consists of the root of $t$ only.}

A~subset $L \subseteq \trees_A$ is a~\emph{tree language}. \emph{Regular} tree languages are the ones recognised by parity non\=/deterministic automata or, equivalently, definable in Monadic Second\=/order Logic (for this equivalence see for example~\cite{gradel_automata}).

\subsection{Polish spaces}%
\label{ssec:polish}
The subsequent subsections recall the topological notions coming from Descriptive Set Theory that we will use throughout the article. We do not aim for completeness, for more details we refer the reader to~\cite{kechris_descriptive}. In the first subsection we define Polish spaces, that are the main objects studied in Descriptive Set Theory. In the remaining section we introduce the main hierarchies usually considered for Polish spaces, i.e.~the Borel, difference, and Wadge hierarchies.

We denote a~topological space by $(X,\tau)$, where $X$ is a~non\=/empty set and $\tau$ is a~family of subsets of $X$ called \emph{open sets}. If $\tau$ is understood from the context we write just $X$ and suppress $\tau$ from the notation. We say that $X$ is a~\emph{Polish space} if $\tau$ is completely metrizable (i.e.~there exists a~complete metric on $X$ that generates $\tau$) and separable (i.e.~there exists a~countable dense subset of $X$).

If a~space $X$ is known from the context and $A\subseteq X$ then by $A\complemento\eqdef X\setminus A$ we denote the complement of $A$ in $X$. Similarly, if $\Gamma$ is a~family of subsets of $X$ then $\Gamma\complemento\eqdef\{A\complemento\mid A\in\Gamma\}$.

Consider a~non\=/empty at most countable set $B$. \new{We will now introduce the so\=/called \emph{prefix topologies} on the spaces $B^\w$ and $\trees_B$ by providing explicitly their bases. However, it is worth noticing that these topologies coincide with the \emph{Tychonoff topologies} when $B$ is considered as a~discrete topological space, see~\cite[Section~2.3]{engelking_topology}}.

\new{The \emph{prefix topology} on the space of infinite words $B^\omega$ over $B$ is} generated by the basic open sets of the form:
\[ N_s = \{ \alpha \in B^\w \mid s \prec \alpha\},\]
with $s \in B^\ast$. When $B= \{0,1\}$, we obtain the \emph{Cantor space}, denoted by $2^\w$. When $B= \omega$, we obtain the \emph{Baire space}, denoted by $\w^\w$. Every space of the form $B^\w$ with the prefix topology is Polish. It is easy to check that the prefix topology is completely metrizable: the metric  $d(\alpha, \beta) = 2^{-n}$, where $n$ is the minimum index such that $\alpha(n) \neq \beta(n)$ \new{(and $d(\alpha,\beta)=0$ if $\alpha=\beta$)}, is complete and it generates the prefix topology. Moreover, if we fix a~symbol $c \in B$, the set
\[ D= \{ s\concat ccc\cdots \mid s \in B^\ast\}\]
is countable (since $B^\ast$ is countable) and dense, so $B^\w$ is separable. In particular, the Cantor space and the Baire space are Polish spaces.

The prefix topology can easily be generalised to trees $\trees_A$, \new{with} basic open sets of the form:
\[ N_p = \{t \in \trees_A \mid p \strictpreftree t \},\]
where $p$ is a~finite partial tree. Every $N_p$ is actually a~clopen (i.e.~both open and closed). We denote this topology by $\toppref$. The topology~$\toppref$ is generated by several metrics. The usual metric considered to generate $\toppref$ is $\metricpref(t_1,t_2) = 2^{-n}$ where $n$ is the minimum length of a~node $u$ such that $t_1(u) \neq t_2(u)$ \new{(and $\metricpref(t_1,t_2)=0$ for $t_1=t_2$)}. Each open ball of $\metricpref$ is a~basic clopen set $N_p$ for a~certain finite partial tree~$p$.

Notice that the metric $\metricpref$ satisfies the following strengthening of the triangle inequality:
\[\metricpref(x,z)\leq \max \big(\metricpref(x,y),\metricpref(y,z)\big).\]
Such a~metric is called an~\emph{ultrametric}, see~\cite[Exercise~2.2] {kechris_descriptive}. This property makes $\metricpref$ too rigid for our way of choosing optimal witnesses (see Definition~\ref{def_optimalstrategies} of optimal strategy trees). Therefore, we will also consider a~different metric, denoted by $\topdisc$ and called the \emph{discounted distance}. This metric also generates $\toppref$ but has a~less intuitive family of open balls. Fix some enumeration $u_0, u_1, \ldots$ of all the nodes in ${\{\dL,\dR\}}^\ast$. Given two trees $t_1$ and $t_2$, for any node $u$, define $\dist(t_1(u),t_2(u))= 0$ if $t_1(u) = t_2(u)$, $1$ in the other case. Let
\[\topdisc(t_1,t_2) \eqdef \sum_{n \geq 0} \frac{1}{2^n}\cdot \dist\big(t_1(u_n),t_2(u_n)\big).\]

\begin{fact}
Regardless of the enumeration $(u_0, u_1,\ldots)$, the prefix and discounted distances yield the same topology.
\end{fact}

\begin{proof}
It is enough to observe that $\toppref$ is exactly the product topology obtained by starting from the discrete topology and the discounted distance is exactly the product metric.
\end{proof}

\begin{fact}%
\label{ft:trees-are-Cantor}
$\trees_A$ with the topology $\toppref$ is a~Polish space homeomorphic to the Cantor space~$2^\omega$.
\end{fact}

\begin{proof}
The proof is a~standard encoding of one compact product space into another. One can also use a~characterisation of the Cantor space, see~\cite[Theorem~7.4, page~35]{kechris_descriptive}.
\end{proof}

\subsection{The Borel hierarchy}

Let $(X,\tau)$ be a~topological space. Recall that $\omega_1$ is the first uncountable ordinal. We define, by a~transfinite recursion on $1 \leq \xi < \omega_1$, the following classes:
\begin{align*}
\bsigma{1}(X) &\eqdef \{A\subseteq X \mid A \text{ is open} \};\\
\bpi{\xi}(X) &\eqdef \{A\complemento\subseteq X \mid A\in \bsigma{\xi} (X) \}= \big(\bsigma{\xi}(X)\big)\complemento;\\
\bsigma{\xi}(X) &\eqdef \big\{ \bigcup_n A_n \mid A_n \in \bpi{\xi_n} (X), \ 1\leq  \xi_n < \xi, \ n \in \omega  \big\}.
\end{align*}
Moreover, for $1 \leq \xi < \omega_1$, we define the intersection of the two classes $\bdelta{\xi}(X)\eqdef \bsigma{\xi}(X) \cap \bpi{\xi}(X)$.

\begin{fact}
For each $1\leq\xi<\omega_1$, the classes $\bsigma{\xi}(X)$ and $\bpi{\xi}(X)$ are closed under finite unions and finite intersections. The class $\bdelta{\xi}(X)$ is also closed under complement and therefore forms a~Boolean algebra.
\end{fact}

The smallest Boolean algebra containing all the sets from $\bsigma{\xi}(X)$ is denoted $\boolcomb{\xi}(X)$. The above fact implies that $\boolcomb{\xi}(X)\subseteq \bdelta{\xi+1}(X)$. For uncountable Polish spaces the inclusion is strict, this fact follows from~\cite[Exercise~22.26(iii)]{kechris_descriptive}.

The \emph{Borel sets} of $X$ are:
\[ \borel(X) = \bigcup_{\xi \in \omega_1} \bsigma{\xi}(X) = \bigcup_{\xi \in \omega_1} \bpi{\xi}(X)= \bigcup_{\xi \in \omega_1} \bdelta{\xi}(X).\]
When the space $X$ is clear from the context, we omit it and write just $\bsigma{\xi}$, $\bpi{\xi}$, etc\ldots

\new{Notice that if $\xi\leq \xi'<\omega_1$ then directly from the definition we know that $\bsigma{\xi}\subseteq\bsigma{\xi'}$. The following fact shows that the opposite containment does not hold in general.}

\begin{factC}[{\cite[Theorem~22.4]{kechris_descriptive}}]
Let $(X,\tau)$ be an~uncountable Polish space. Then the Borel hierarchy of $X$ does not collapse i.e.~every class $\bsigma{\xi}$ is properly contained in~$\bsigma{\xi +1}$.
\end{factC}

For the rest of the article we will focus on the first two levels of the Borel hierarchy, as depicted in Figure~\ref{fig:borel-hier}.

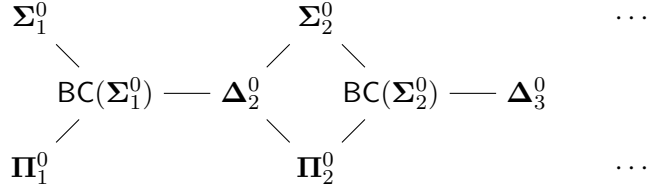
\begin{figure}
\begin{center}
\begin{tikzpicture}
  \newcommand{\wid}{3.8}
  \foreach \x in {1, 2} {
    \node (s-\x) at (\wid*\x-1,1) {$\bsigma{\x}$};
    \node (b-\x) at (\wid*\x,0) {$\boolcomb{\x}$};
    \node (p-\x) at (\wid*\x-1,-1) {$\bpi{\x}$};
  }

  \foreach \x in {2, 3} {
    \node (d-\x) at (\wid*\x-2,0) {$\bdelta{\x}$};
  }

  \node at (\wid*2.85,+1) {$\cdots$};
  \node at (\wid*2.85,-1) {$\cdots$};

  \foreach \from/\to in {1/2, 2/3} {
    \draw[-] (p-\from) -- (b-\from);
    \draw[-] (s-\from) -- (b-\from);

    \draw[-] (b-\from.east) -- (d-\to.west);
  }

  \draw[-] (d-2) -- (s-2);
  \draw[-] (d-2) -- (p-2);
\end{tikzpicture}
\end{center}
\caption{The first levels of the Borel hierarchy.}%
\label{fig:borel-hier}
\end{figure}

\subsection{The difference hierarchy}%
\label{ssec:difference}
The Borel hierarchy is refined by the so\=/called the \emph{difference hierarchy}, see~\cite[Section~22.E]{kechris_descriptive}. First notice that every ordinal $\theta$ can be uniquely written as $\lambda+n$, where $\lambda$ is $0$ or a~limit ordinal and $n \in \omega$. We say that the \emph{parity} of $\theta$ is \emph{even} (resp. \emph{odd}) if $n$ is even (resp.\ odd).

\begin{defi}%
\label{def:gerarchiadiff}
Let $X$ be a~topological space, $\Gamma$ a~family of subsets of $X$, and $\theta<\omega_1$ a~countable ordinal. A~set $A\subseteq X$ is called a~\emph{$\theta$\=/difference of $\Gamma$ sets} if and only if there exists a~$\theta$\=/indexed sequence of sets ${(A_\eta)}_{\eta<\theta}\subseteq\Gamma$ that is non\=/decreasing (i.e.~$A_{\eta} \subseteq A_{\eta'}$ if $\eta\leq \eta'$) and:
\begin{align*}
x \in A \Longleftrightarrow & \text{ the minimum $\eta<\theta$ such that $x \in A_\eta$,}\\
&\text{ has parity opposite to that of $\theta$.}
\end{align*}
The family of all $\theta$\=/differences of $\Gamma$ sets is denoted ${\mathcal{D}_{\theta}}(\Gamma)$. In particular, for each $\xi<\omega_1$ the class $\ndiff{\theta}{\xi}(X)$ is the family of all $\theta$\=/differences of sets from the Borel class $\bsigma{\xi}(X)$.
\end{defi}

\new{
The above definition requires the sequence ${(A_\eta)}_{\eta<\theta}$ to be non\=/decreasing and we will focus on $\Gamma=\bsigma{\xi}(X)$. These are important assumptions, because of the following remark.
}

\begin{rem}
\new{For each $\xi<\omega_1$ we have $\mathcal{D}_\omega\big(\bpi{\xi})=\bsigma{\xi+1}$.}
\end{rem}

\begin{proof}
\new{
The inclusion $\mathcal{D}_\omega\big(\bpi{\xi})\subseteq\bsigma{\xi+1}$ follows directly from the definition. For the opposite direction, consider a~set $A\in\bsigma{\xi+1}$. It is easy to check that $A$ can be written as $\bigcup_{n\in\w} A_n$ with ${(A_n)}_{n\in\omega}$ non\=/decreasing and each $A_n$ in $\bpi{\xi}$. Then take $A'_n\eqdef A_{\lfloor n\rfloor}$ and notice that this sequence is also non\=/decreasing and contains only $\bpi{\xi}$ sets. However, $x\in A$ iff $\exists n.\ x\in A_n$ iff the minimum $n$ such that $x\in A_n'$ is even. Therefore, $A$ is an~$\omega$\=/difference of $\bpi{\xi}$ sets and $A\in \mathcal{D}_\omega\big(\bpi{\xi})$.
}
\end{proof}

Notice that for a~natural number $n$ we have $A \in  \ndiff{n}{\xi}$ if and only if it can be written as follows (see Figure~\ref{fig:differences}):
\begin{align}
A &= A_0 \cup (A_2 \setminus A_1) \cup \cdots \cup (A_{n-1} \setminus A_{n-2})&\text{if $n$ is odd,}
\label{eq:form-diff-n-odd}\\
A &=  (A_1 \setminus A_0) \cup \cdots \cup (A_{n-1} \setminus A_{n-2})&\text{if $n$ is even,}
\label{eq:form-diff-n-even}
\end{align}
with $A_0, \ldots, A_{n-1}$ belonging to $\bsigma{\xi}$.

\new{
As one can easily check from the definition, the classes $\ndiff{\eta}{\xi}$ are monotone both in $\eta$ and in $\xi$.
}

\begin{fact}%
\label{ft:diff-monotone}
\new{For each $\eta\leq \eta'$ and $\xi\leq \xi'$ we have
\[\ndiff{\eta}{\xi}\subseteq \ndiff{\eta'}{\xi'}.\]}
\end{fact}

The following theorem shows that the difference hierarchy over $\bsigma{\xi}$ saturates the successive class $\bdelta{\xi+1}$.

\begin{figure}
\centering
\begin{tikzpicture}[scale=0.8]
\tikzstyle{lll}=[scale=0.8, anchor=south west, inner sep=0]
\newcommand{\dotss}{0.45}
\newcommand{\dotsx}{0.45*0.4}
\newcommand{\dotsy}{0.45*0.9}

\coordinate (z) at (0,0);

\fill[black] (z) circle (\dotss*11);
\fill[white] (z) circle (\dotss*9);
\fill[black] (z) circle (\dotss*5);
\fill[white] (z) circle (\dotss*3);
\fill[black] (z) circle (\dotss*1);

\node[toC, scale=2.2] at (0, \dotss*6+0.15) {$\vdots$};

\node[lll, black] at (\dotsx*1.1,\dotsy*1.1) {$A_0$};
\node[lll, white] at (\dotsx*3.2,\dotsy*3.2) {$A_1$};
\node[lll, black] at (\dotsx*5.2,\dotsy*5.2) {$A_2$};
\node[lll, white] at (\dotsx*9.2,\dotsy*9.2) {$A_{n-2}$};
\node[lll, black] at (\dotsx*11.2 ,\dotsy*11.2 ) {$A_{n-1}$};
\end{tikzpicture}
\caption{A~set $A=A_0 \cup (A_2 \setminus A_1) \cup \cdots \cup(A_{n-1} \setminus A_{n-1})$ ($A$ is the union of the black parts) with $A_0 \subseteq A_1 \subseteq A_2 \subseteq \cdots \subseteq A_{n-2} \subseteq A_{n-1}$ and $A_i \in \bsigma{\xi}$ for every $i$.}%
\label{fig:differences}
\end{figure}
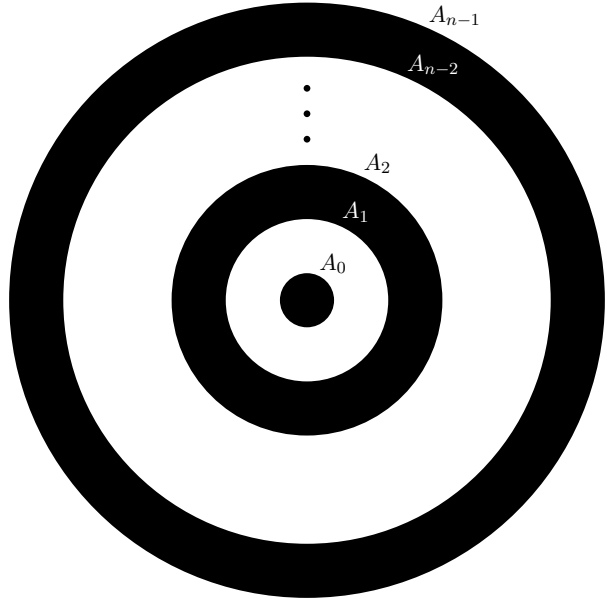

\begin{thm}[{Hausdorff, Kuratowski, see~\cite[Theorem~22.27, page~176]{kechris_descriptive}}]%
\label{thm:HausdorffKuratowski}
In Polish spaces and for any $1 \leq \xi < \omega_1$ we have that
\[\bdelta{\xi+1}= \bigcup_{1 \leq \theta < \omega_1} \ndiff{\theta}{\xi}.\]
\end{thm}

Similarly, the first $\omega$ levels of the hierarchy coincide with the class~$\boolcomb{\xi}$:

\begin{thmC}[{\cite[Exercise~22.30, page~177]{kechris_descriptive}}]
In Polish spaces and for any $1 \leq \xi < \omega_1$ we have that
\[\boolcomb{\xi}= \bigcup_{1 \leq \theta < \omega} \ndiff{\theta}{\xi}.\]
\end{thmC}

\subsection{The Wadge hierarchy}%
\label{ssec:wadge}

\new{We are now in the position} to define the Wadge hierarchy of a~general topological space. Later in the article we will focus on the specific case of the Wadge hierarchy of the Cantor space $2^\w$.

\begin{defi}
Let $X$, $Y$ be two topological spaces. We say that a~set $A \subseteq X$ \emph{continuously reduces} to $B\subseteq Y$ if there is a~continuous function $\fun{f}{X}{Y}$ such that the~pre\=/image $f^{-1}(B)=\{x\in X \mid f(x) \in B\}$ is equal to $A$ (i.e.~$x \in A \Leftrightarrow f(x) \in B$ for every $x \in X$).
\end{defi}

The following proposition shows that continuous reductions preserve the topological classes defined above.

\begin{prop}
Let $\Gamma$ be a~level of the Borel hierarchy or of the difference hierarchy. Then $\Gamma$ is closed under continuous preimages: if $B$ is a~subset of a~topological space $Y$ such that $B \in \Gamma(Y)$ and $\fun{f}{X}{Y}$ is a~continuous function from a~topological space $X$ to $Y$, then $f^{-1}(B) \in \Gamma(X)$.
\end{prop}

Now we are in place to define the Wadge order.

\begin{defi}
Let $X$ and $Y$ be two topological spaces and let $A \subseteq X$ and $B \subseteq Y$. We say that $A$ is \emph{Wadge reducible} to $B$, and we denote it by $A \wadgeq B$, if there exists a~continuous reduction of $A$ to $B$. We say that $A$ is \emph{Wadge equivalent} to $B$, in symbols $A \wadgee B$, if $A \wadgeq B$ and $B \wadgeq A$. Finally, we write $A \wadge B$ if $A \wadgeq B$ and $B \wadgeq A$ does not hold.
\end{defi}

\begin{fact}
${\wadgeq}$ is an~equivalence relation.
\end{fact}

The relation ${\wadgeq}$ induces a~partial order between the ${\wadgee}$\=/classes, called \emph{Wadge degrees}, of subsets of topological spaces. If we fix a~space $X$ and we restrict the ordering induced by $\wadgeq$ to the sets of $X$, we obtain the \emph{Wadge hierarchy} of $X$. If $A$ is a~subset of $X$, then by ${[A]}_{\mathrm{W}}$ we denote its Wadge degree:
\[ {[A]}_{\mathrm{W}} = \{B \subseteq X \mid B \wadgee A\}.\]

Even tough the Wadge hierarchy can be defined for any topological space, its shape for a~generic space can be very complicated (for example the Wadge hierarchy of many non zero\=/dimensional topological spaces, including the space of real numbers, is very complicated: see for instance~\cite{rosschlicht} and~\cite{rosselivanov}).
Also, the good properties of the hierarchy (like its width or well\=/foundedness) depend on the determinacy of related games. Therefore, in our work we will restrict our attention to the order ${\wadgeq}$ restricted to the first levels of the Borel hierarchy of the Cantor space $2^\omega$. We refer the reader to~\cite{andretta_camerlo_wadge} for a~description of the structure of the Wadge hierarchy for~$2^\omega$.


\begin{thm}[{Wadge's lemma, see~\cite[Theorem~21.14, page~156]{kechris_descriptive}}]%
\label{thm:WadgeLemma}
For any $A, B \in \borel(2^\w)$ it holds that
\[ A \wadgeq B \quad \text{or} \quad B\complemento \wadgeq A.\]
\end{thm}

\begin{thm}[{Wadge, Martin, Monk, see~\cite[Theorem~21.15, page~158]{kechris_descriptive}}]
The ordering ${\wadgeq}$ among the Borel sets of $2^\w$ is well\=/founded.
\end{thm}

\begin{defi}
A~set which is Wadge reducible to its complement is called \emph{self\=/dual}, otherwise it is called \emph{non self\=/dual}.
\end{defi}

\begin{exa}
It is easy to check that every non\=/trivial \new{(i.e.~different than $\emptyset$ and $2^\w$)} clopen subset of $2^\w$ is self\=/dual.
\end{exa}

Since the notion of self\=/duality is invariant under $\wadgee$ we can speak of self\=/dual and non self\=/dual Wadge degrees. If ${[A]}_{\mathrm{W}}$ is a~non self\=/dual Wadge degree then we say that the pair $\{{[A]}_{\mathrm{W}},{[A\complemento]}_{\mathrm{W}}\}$ is a~\emph{non self\=/dual pair}.

\begin{cor}
The anti\=/chains in the Wadge degrees have length at most $2$ and are of the form
\[ \big\{{[A]}_{\mathrm{W}}, {[A\complemento]}_{\mathrm{W}}\big\},\]
with $A$ non self\=/dual.
\end{cor}

Notice that technically every Wadge degree does not contain the elements contained in the previous degrees of the hierarchy. For example, the Wadge degree ${[C]}_{\mathrm{W}}$, where $C$ is a~clopen set different from $2^\w$ and $\emptyset$, contains all the clopen sets except the whole space $2^\w$ and the empty set $\emptyset$. This is obvious, since any Wadge degree is an~equivalence class of the relation ${\wadgee}$. We define the \emph{Wadge class} of a~set as the union of its Wadge degree with all its predecessors in the hierarchy. For example the Wadge class $\bdelta{1}$ is obtained by taking the union of the Wadge degree $\bdelta{1} \setminus \{ \emptyset, 2^\w\}$ with its predecessors $\{\emptyset\}$ and $\{ 2^\w\}$. It is clear that the two hierarchies, the one of Wadge degrees and the one of Wadge classes, are isomorphic as orders, so we can treat both the hierarchies in the same way. In the pictures of this section we show the hierarchy of the Wadge classes, because they are more intuitive and easier to describe.

\begin{thm}[See~\cite{andretta_camerlo_wadge}]
In the Cantor space, ${[2^\w]}_{\mathrm{W}}=\{2^\w\}$ and ${[\emptyset]}_{\mathrm{W}}=\{ \emptyset\}$ are the two minimal Wadge degrees and they clearly form a~non self\=/dual pair. Then we have the Wadge degree formed by any clopen set different from $2^\w$ and $\emptyset$  and this is a~self\=/dual Wadge degree. The hierarchy continues with a~constant alternation of a~non self\=/dual pair and one self\=/dual Wadge degree. All the limit levels of the Wadge hierarchy consist of a~non self\=/dual pair. Certain specific Wadge classes coincide with the levels of the Borel hierarchy.
\end{thm}

Hence, the Wadge hierarchy of the Cantor space has the shape as depicted in Figure~\ref{fig:wadge-precise}.

Now we can assign an~ordinal to any level of the hierarchy. This ordinal is the \emph{Wadge rank} of a~Wadge degree (or of the corresponding Wadge class). The two bottom Wadge degrees $\{\emptyset\}$ and $ \{2^\w\}$ have Wadge rank \new{$1$}, the Wadge degree $\bdelta{1} \setminus \{\emptyset, 2^\w\}$ has Wadge rank \new{$2$} (so it has the Wadge class $\bdelta{1}$), and so on.

Among the non self\=/dual Wadge classes we find the classes $\bsigma{n}$ and $\bpi{n}$. The classes $\bsigma{1}$ and $\bpi{1}$ are immediately after the Wadge class $\bdelta{1}$, so their Wadge rank is \new{$3$}. Then, between the~non self\=/dual pair
\[\{\bsigma{n}, \bpi{n}\}\]
and the successive
\[\{\bsigma{n+1}, \bpi{n+1}\}\]
there are $\omega_1^{\omega_1^{\iddots_{n \text{ times} }^{\omega_1}}}$ Wadge classes. In particular, there are $\omega_1$ Wadge classes between the pair $\{\bsigma{1}, \bpi{1}\}$ and $\{\bsigma{2} , \bpi{2} \}$. Hence, the Wadge rank of the Wadge classes $\bsigma{2}$ and $\bpi{2}$ is~$\omega_1$, and the Borel class $\bdelta{2}$ contains $\omega_1$ different levels of the Wadge hierarchy.

\begin{figure}
\centering
\begin{tikzpicture}[scale=0.9]
\node[dot, scale=2.0] at (-10,+1) {};
\node[toC, inner sep=0, anchor=north] at (-10, +1.8) {$\{2^\w\}$};
\node[dot, scale=2.0] at (-10,-1) {};
\node[toC, inner sep=0, anchor=south] at (-10, -1.8) {$\{\emptyset\}$};
\node[dot, scale=2.0] at (-9,0) {};
\node[toC, inner sep=0, anchor=north] at (-9, +0.6) {$\bdelta{1}$};
\node[dot, scale=2.0] at (-8,+1) {};
\node[toC, inner sep=0, anchor=north] at (-8, +1.8) {$\bsigma{1}$};
\node[dot, scale=2.0] at (-8,-1) {};
\node[toC, inner sep=0, anchor=south] at (-8, -1.8) {$\bpi{1}$};
\node[scale=3.0] at (-6,0) {\ldots};
\node[toC, inner sep=0, anchor=north] at (-6, +0.6) {$\omega_1$ alternations};
\node[dot, scale=2.0] at (-4,+1) {};
\node[toC, inner sep=0, anchor=north] at (-4, +1.8) {$\bsigma{2}$};
\node[dot, scale=2.0] at (-4,-1) {};
\node[toC, inner sep=0, anchor=south] at (-4, -1.8) {$\bpi{2}$};
\node[dot, scale=2.0] at (0,+1) {};
\node[scale=3.0] at (-2,0) {\ldots};
\node[toC, inner sep=0, anchor=north] at (-2.1, +0.6) {$\omega_1^{\omega_1}$ alternations};
\node[toC, inner sep=0, anchor=north] at (0, +1.8) {$\bsigma{3}$};
\node[dot, scale=2.0] at (0,-1) {};
\node[toC, inner sep=0, anchor=south] at (0, -1.8) {$\bpi{3}$};
\node[scale=3] at (1,0) {\ldots};
\end{tikzpicture}
\caption{An~initial fragment of the Wadge hierarchy inside the Borel hierarchy of $2^\w$.}%
\label{fig:first-omega-wadge}
\end{figure}
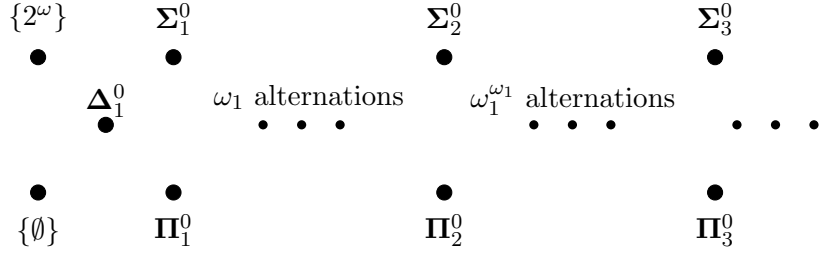

Now we focus on the segment that we will study in this article, that is the initial segment from the beginning of the hierarchy until the Wadge classes $\bsigma{2}$ and $\bpi{2}$. Using this notion we can express a~characterisation of the first $\omega_1$ levels of the Wadge hierarchy in $2^\omega$ in terms of the difference hierarchy.

\begin{thm}[See~\cite{andretta_camerlo_wadge}]
For every $m\geq 1$ every class $\ndiff{\theta}{m}$ corresponds to a~Wadge class of a~non self\=/dual pair. For $m=1$ these are essentially all Wadge classes: there is no non self\=/dual Wadge class between $\ndiff{\theta}{1}$ and $\ndiff{\theta+1}{1}$. Therefore, Figure~\ref{fig:wadge-precise} depicts the first $\omega_1$ levels of the Wadge hierarchy.
\end{thm}

Hence, in the Cantor space the difference hierarchy is an~important tool to understand the Wadge hierarchy, especially in the initial part up to $\bsigma{2}$. Beyond that level, the difference hierarchy becomes much coarser (i.e.~it skips a~lot of Wadge degrees).

\section{The game for Wadge ranks}%
\label{sec:finite-game}

In this section we define a~game that we will use in this article to obtain results of decidability of Wadge degrees with Wadge ranks up to $\omega$ (i.e.~the classes $\ndiff{n}{1}$ for $n\in\w$). This game is played by two players, named Alternator and Constrainer and it is a~finite duration game. In this article we work on the space $\trees_A$, but \emph{a~priori} this game can be defined in any topological space and the characterization that it gives holds in general. Yet, in the case of regular languages of trees, it is possible to decide which player wins the game. This fact will be crucial to state the results about decidability.

Let us describe the game. Let $X$ be a~topological space, $U_0\subseteq X$ open and non\=/empty, and let $X_1, \ldots,X_n$ be arbitrary subsets of $X$. We define the game
\[ \gameH_{U_0}(X_1, \ldots, X_n)\]
played by Constrainer (choosing open subsets of $X$) and Alternator (choosing points of $X$). The game will last for $n$ rounds, a~round $i$ for $1,\ldots, n$ of the game is played as follows:
\begin{enumerate}
\item Alternator chooses a~point $x_i \in U_{i{-}1}\cap X_i$.
If there is no such point $x_i$, the game is interrupted and Constrainer wins immediately.
\item Constrainer chooses an~open set $U_i\subseteq U_{i-1}$ that contains $x_i$ and the next round is played.
\end{enumerate}
If Alternator manages to survive $n$ rounds then he wins, otherwise Constrainer wins.

A~special variant of the game, when $U_0=X$ is the whole space, is denoted simply as $\gameH(X_1,\ldots,X_n)$.
Now we prove some properties of this game.

Let $(X, \tau)$ be a~topological space, $U\subseteq X$ open non\=/empty, and let $X_1, \ldots, X_n$ be subsets of $X$. Consider the game $\gameH_U(X_1, \ldots, X_n)$. In this framework we can represent a~position of a~play trough a~tuple
\[ \langle U_0, x_1, U_1, x_2, U_2, \ldots, x_i, U_i\rangle,\]
where $x_1, \ldots, x_i \in X$ and $U_0, U_1, \ldots, U_i \in \tau$ with $U_0=U$.
A~\emph{strategy for Constrainer} in the game $\gameH_U(X_1, \ldots, X_n)$ is a~function
\[ \fun{\sigma}{X^{\leq n}}{\tau}.\]
If $s$ is a~word belonging to $X^{\leq n}$ compatible with the game $\gameH$ and with $\sigma$, then $\sigma(s)$ is the open set played by Constrainer in the position
\[\langle s_0, \sigma(s_0), s_1, \sigma(s_0s_1),\ldots, s_{\lunghezza(s)-1}, \sigma(s)\rangle.\]
If $s$ does not represent a~position compatible with the game $\gameH$ and with $\sigma$ (for example because the second letter of $s$ is not an~element of $X_2$ or it is not an~element of $\sigma(s_0)$), then we put $\sigma(s)= \emptyset$ by convention. As usual, a~strategy $\sigma$ for Constrainer is \emph{winning} if Constrainer wins every play where he follows $\sigma$. In a~specular way we could define strategies for Alternator, but we will not use them in this article.

\new{
Since the duration of the game is finite, it is determined --- one of the players has a~winning strategy. The following remark is not used in this article but provides a~simplified form for the strategies of the players.
}

\begin{rem}
\new{The game $\gameH_U(X_1,\ldots,X_n)$ is positionally determined: the winner of the game does not change if we insist that the players' strategies depend only on the number of the round and the last move of the opponent. This means that we can freely assume that Alternators's point $x_i$ depends only on $i$ and the current open set $U_{i-1}$; while Constrainer's open set $U_i$ depends only on $i$, the current point $x_i$, and the previous open set $U_{i-1}$.
}
\end{rem}

The first property we prove is Refinement Lemma, that states that if the sets $X_1, \ldots, X_n$ are split into finitely many parts each and Alternator wins $\gameH(X_1, \ldots ,X_n)$ then he can win for some choice of parts of $X_1, \ldots, X_n$.

\begin{lem}[Refinement Lemma]%
\label{lem:refinement_lemma}
Let $X_1, \ldots, X_n$ be subsets of a~topological space $X$. For $i \in \{1,\ldots, n\}$, let $\mathcal{Y}_i$ a~finite family
of sets partitioning $X_i$. For any non\=/empty open $U\subseteq X$, if Alternator wins
\[ \gameH_U(X_1, \ldots ,X_n)\]
then there exist $Y_1 \in  \mathcal{Y}_1, \ldots, Y_n \in \mathcal{Y}_n$ such that Alternator wins
\[\gameH_U(Y_1, \ldots, Y_n).\]
\end{lem}

\begin{proof}
We prove the theorem by induction on $n$.
The induction base is immediate, because Alternator always wins when $n = 0$ and he wins when $n=1$ if and only if $X_1\neq\emptyset$. Now prove the induction step. Consider the first move by Alternator, and assume that he chooses the same point $x \in U$ \new{as chosen by his winning strategy in $\gameH_U(X_1, \ldots ,X_n)$}. This point necessarily belongs to some $Y_1 \in \mathcal{Y}_1$. For $i \in \omega$, let $U_i$ be the open ball around $x$ of radius $\frac{1}{i}$. By the definition of the game, we know that for every $i<\w$ Alternator wins
\[\gameH_{U_i} (X_1, \ldots , X_n).\]
Thus, he also wins $\gameH_{U_i}(X_2,\ldots,X_n)$ and by the inductive assumption, we know that for every $i$ there exist $Y^{(i)}_2 \in
\mathcal{Y}_2$, \ldots, $Y^{(i)}_n \in \mathcal{Y}_n$ such that Alternator wins
\[ \gameH_{U_i}(Y^{(i)}_2, \ldots, Y^{(i)}_n).\]
By Pigeon\=/hole Principle, 
there must be some $Y_2$, \ldots, $Y_n$ such that
\[ (Y_2, \ldots, Y_n) = (Y^{(i)}_2, \ldots, Y^{(i)}_n)\]
holds for infinitely many $i$. Since the game $\gameH_V(Y_2, \ldots, Y_n)$ grows more difficult for Alternator as the open set $V$ becomes smaller, and since every open set $V$ that contains $x$ contains some $U_i$, we conclude that Alternator wins
\[\gameH_V (Y_2, \ldots ,Y_n)\]
for every $V$ that contains $x$. By viewing $V$ as a~response of Constrainer to Alternator's move
$x \in Y_1$, we conclude that Alternator wins the game
\[\gameH_U(Y_1, \ldots ,Y_n). \qedhere\]
\end{proof}

Now let $Y$ be a~subset of our topological space $X$ and consider a~particular case of the game where the sets $X_1, \ldots, X_n$ alternate between $Y$ and its complement, i.e.~we consider $\gameH(X_1, \ldots, X_n)$ where $X_i$ is $Y$ if $i$ is odd, $Y\complemento$ otherwise. We denote by $\gameHs(Y,n)$ that game and by $\gameHs_U(Y,n)$ the variant relativised to a~non\=/empty open set $U\subseteq X$.

\begin{exa}
Consider the game where the topological space $X$ is the space of real numbers $\R$ and $Y=\Q$, i.e.~the rational numbers. Then for every $n$, Alternator wins the game $\gameHs(Y,n)$.
\end{exa}

\begin{rem}
Notice that if Alternator wants to survive in $\gameHs(Y,n)$ as long as possible, he has to avoid to play points in the interior\footnote{Recall that the interior of a~set $Y$ is the union of all open sets contained in $Y$.} of $Y$ and the interior of $Y\complemento$. For example, if at the first round Alternator plays $x$ belonging to the interior of $Y$ then Constrainer can play (a~subset of) the interior of $Y$ and Alternator loses because he cannot go outside $Y$ any more.
\end{rem}

\begin{exa}
In the real numbers $\R$, let $Y$ be the complement of $\{ \frac{1}{n}\in\R \mid n \in \omega\}$.
Alternator wins $\gameHs(Y,3)$. Indeed, in the first round Alternator can play $0 \in Y$. In the second round, Alternator plays $\frac{1}{n}\notin Y$ for some large $n$ depending on Constrainer's move. In the third round, Alternator plays $\frac{1}{n} + \epsilon \in Y$, for some small $\epsilon$ depending on Constrainer's move.

\new{
We now argue that Constrainer wins $\gameHs(Y,n)$ for $n \geq 4$. Notice that since $x_2$ must belong to $Y^\cp=\{ \frac{1}{n}\in\R \mid n \in \omega\}$, Constrainer can choose $U_2\subseteq U_1$ in such a~way to guarantee that $x_2\in U_2$ but $0\notin U_2$. Then Alternator chooses a~point $x_3\in Y\cap U_2$ that must be distinct from $0$. Thus, there exists an~open set $U_3$ such that $x_3\in U_3\subseteq U_2$ but $U_3\subseteq Y$. This guarantees that Alternator is not able to choose $x_4\in Y^\cp\cap U_3$, making him lose the game $\gameHs(Y,4)$.
}
\end{exa}

\begin{rem}%
\label{rem:strategies}
Let $\sigma_1$ and $\sigma_2$ be two strategies for Constrainer such that $\sigma_1 (s) \subseteq \sigma_2(s)$ for any finite word $s\in X^{\leq n}$. If $\sigma_2$ is winning for Constrainer then $\sigma_1$ is winning for Constrainer too.
\end{rem}

\begin{lem}%
\label{lem:wecanplaybasic}
Choose some basis $B$ for the topology of the topological space $X$. If Constrainer has a~winning strategy in $\gameHs(Y,n)$ then he has a~winning strategy which uses only basic open sets from $B$.
\end{lem}

\begin{proof}
\new{Consider} any~function $f$ that to every pair $(U,x)$, where $U$ is an~open set and $x \in U$, assigns a~basic open set $V\in B$ such that $x \in V$ and $V \subseteq U$. Let $\sigma$ be a~winning strategy for Constrainer. We can define another winning strategy $\bar{\sigma}$ that \new{always} takes sets from~$B$: given a~finite word $s$ of length $i$, we define $\bar{\sigma}(s)= f\big(\sigma(s),s_i\big)$. Since $\sigma$ was winning for Constrainer, by Remark~\ref{rem:strategies} also $\bar{\sigma}$ is winning.
\end{proof}

Now we are ready to give a~characterization of $\ndiff{n}{1}$ sets in terms of the game $\gameH$.

\begin{thm}%
\label{thm:decidforsinglefinitediff}
Let $X$ be a~topological space and let $Y \subseteq X$. The following conditions are equivalent:
\begin{enumerate}
\item $Y$ belongs to $\ndiff{n}{1}$.
\item Constrainer wins the game $\gameHs(Y,n+1)$.
\end{enumerate}
\end{thm}

\begin{proof}

We have to prove both directions separately.

\paragraph*{\bf Implication $1 \Rightarrow 2$.} Suppose that $n$ is odd and let
\begin{equation}
\label{eq:form-of-Y}
Y = A_0 \cup (A_2 \setminus A_1) \cup  \cdots \cup (A_{n{-}1} \setminus A_{n{-}2}),
\end{equation}
where $A_0 \subseteq A_1 \subseteq A_2 \subseteq \cdots \subseteq A_{n-2} \subseteq A_{n-1}$ and every $A_i$ is open for $0 \leq i \leq n-1$. Then, a~winning strategy for Constrainer in $\gameHs(Y,n)$ is to play $A_{n-1}$ as the first move, $A_{n-2}$ as the second move, and so on. Equation~\eqref{eq:form-of-Y} implies that this is a~valid strategy of Constrainer. The $n$th move will be $A_0$, and since $A_0 \subseteq Y$, at that point Alternator loses at the $(n{+}1)$th round. The case of $n$ even is completely dual, in that case $A_0 \subseteq Y\complemento$.

\paragraph*{\bf Implication $2 \Rightarrow 1$.}

Suppose that $n$ is odd, the opposite case can be solved by an entirely dual argument. Our aim is to present $Y$ as in~\eqref{eq:form-of-Y}. For $i=0,1,2,\ldots$ define sets
\begin{align*}
A_i &\eqdef \bigcup\big\{ U\mid \text{Constrainer wins $\gameHs_U(Y,i+1)$}\big\}&\text{for odd $i$}\\
A_i &\eqdef \bigcup\big\{ U\mid \text{Constrainer wins $\gameHs_U(Y\complemento,i+1)$}\big\}&\text{for even $i$}
\end{align*}
where the unions range over open sets $U\subseteq X$. Notice that if Constrainer wins $\gameHs_U(Y,i)$ then he also wins $\gameHs_U(Y\complemento,i+1)$ by the same strategy, just playing $U$ in the first round. Therefore, the family $(A_i)$ is increasing. By the definition, all the sets $A_i$ are open. Clearly, the assumption that Constrainer has a~winning strategy in $\gameHs(Y,n+1)$ implies that $A_{n}$ is the whole space. Thus, it is enough to inductively prove the following claim.

\begin{claim}%
\label{cl:Y-invariant}
For $i=0,1,\ldots$ the following holds
\begin{align*}
Y\cap A_i &= A_0\cup (A_2\setminus A_1) \cup\cdots\cup(A_{i-1}\setminus A_{i-2}) &\text{for odd $i$}\\
Y\complemento\cap A_i &= (A_1\setminus A_0)\cup (A_3\setminus A_2) \cup\cdots\cup(A_{i-1}\setminus A_{i-2}) &\text{for even $i$}
\end{align*}
\end{claim}

Notice that by the definition $A_0\subseteq Y$ --- whenever there exists $x\in U\cap Y\complemento$ then Alternator wins $\gameHs_U(Y\complemento,1)$ by playing $x$. Therefore, the above claim holds for $i=0$ as $Y\complemento\cap A_0=\emptyset$.

Assume that Claim~\ref{cl:Y-invariant} holds for $i{-}1$ and consider the two cases for $i$.

\paragraph*{\bf The case of odd $i$.}
In that case we need to prove that $Y\cap A_i$ is of the form from~\eqref{eq:form-of-Y}. Consider a~point $x\in A_i$. First consider the case that $x\in A_{i-1}$. Then by the inductive assumption $x\in Y$ if and only if
\[x\notin (A_1\setminus A_0)\cup(A_3\setminus A_2)\cup\cdots\cup(A_{i-2}\setminus A_{i-3})\]
what is equivalent to the disjunction of $x\in (A_{i-1}\setminus A_{i-2})$ or
\[x\in A_0\cup (A_2\setminus A_1)\cup\cdots\cup(A_{i-3}\setminus A_{i-4}).\]
Thus, $x\in Y$ if and only if
\[x\in A_0\cup (A_2\setminus A_1)\cup\cdots\cup(A_{i-1}\setminus A_{i-2}).\]
Thus, the statement of Claim~\ref{cl:Y-invariant} holds in that case.

Now assume that $x\notin A_{i-1}$. We will prove that in that case $x\notin Y$. Assume contrarily that $x\in Y$. Since $x\in A_i$, there exists an open set $U$ such that $x\in U \subseteq A_i$ and Constrainer wins $\gameHs_U(Y,i+1)$. Consider the first round of $\gameHs_U(Y,i+1)$ in which Alternator plays $x$ and a~winning strategy of Constrainer replies with $V\ni x$. This means that Constrainer has a~winning strategy in $\gameHs_V(Y\complemento,i)$ and therefore by the definition $x\in V\subseteq A_{i-1}$, a~contradiction.

\paragraph*{\bf The case of even $i$.} This case is entirely dual: we take $x\in A_i$ and consider the case that $x\in A_{i-1}$. Then the following conditions are equivalent:
\begin{align*}
x&\notin Y\\
x&\notin A_0\cup (A_2\setminus A_1)\cup\cdots\cup(A_{i-2}\setminus A_{i-3})\\
x&\in (A_1\setminus A_0)\cup(A_3\setminus A_2)\cup\cdots\cup(A_{i-3}\setminus A_{i-4})\text{ or }x\in (A_{i-1}\setminus A_{i-2})\\
x&\in (A_1\setminus A_0)\cup(A_3\setminus A_2)\cup\cdots\cup(A_{i-1}\setminus A_{i-2})
\end{align*}
and therefore Claim~\ref{cl:Y-invariant} holds in that case. If $x\notin A_{i-1}$ then we need to prove that $x\in Y$. Since $x\in A_i$, we know that $x\in U\subseteq A_i$ with Constrainer winning $\gameHs_U(Y\complemento, i+1)$ for some open $U$. Assume contrarily that $x\notin Y$ and as before we see a~contradiction, as $x$ is a~valid move of Alternator in $\gameHs_U(Y\complemento,i+1)$ and therefore $x\in A_{i-1}$.
\end{proof}

\begin{cor}%
\label{cor:char-by-winning-the-game}
The following conditions are equivalent for a~set $Y$:
\begin{enumerate}
\item $Y$ belongs to $\bigcup_{n \in \omega} \ndiff{n}{1}=\boolcomb{1}$.
\item Constrainer wins the game $\gameHs(Y,n)$ for all but finitely many $n$.
\end{enumerate}
\end{cor}

\begin{proof}
It follows \new{directly from Theorem~\ref{thm:decidforsinglefinitediff} and Fact~\ref{ft:diff-monotone}.}
\end{proof}

Since the family of sets defined by prefixes $N_p$ for all finite prefixes $p$ is a~basis of the topology on $\trees_A$, we obtain the following corollary for the case $X=\trees_A$.

\begin{cor}%
\label{cor:playing-prefixes}
Assume that $X=\trees_A$ is the space of all trees and $L\subseteq \trees_A$. Then, when considering strategies of Constrainer in $\gameHs(L,n)$ we can assume that each open set $U_i$ played by him is a~basic open set, i.e.~$U_{i+1}=N_p$ for a~finite prefix $p$ of the currently played tree $t_i$. Then the condition that $U_{i+2}\subseteq U_{i+1}$ boils down to the assumption that $p_{i+2}\supseteq p_{i+1}$.
\end{cor}

\subsection{Decidability of finite levels of the Wadge hierarchy}%
\label{ssec:finite-levels}

In this short subsection we use the game $\gameH(L_1,\ldots , L_n)$ to show that, given a~natural number $n$, it is decidable if a~regular language~$L$ is an~$n$\=/difference of open sets (i.e.~belongs to $\ndiff{n}{1}$). Recall that, without loss of generality (see Corollary~\ref{cor:playing-prefixes}) we can assume that Constrainer plays finite prefixes of the trees played by Alternator, and Alternator has to extend the finite prefixes played by Constrainer.

\begin{exa}
Consider the language
\[ L= \{ t\in\trees_A \mid \text{infinitely many letters $a$ appear in $t$}\}.\]
$L$ is regular and it is easy to check that Alternator wins the game $\gameHs(L,n)$ for every $n \in \omega$. This is because every finite prefix can be extended to a~tree with finitely many $a$ or to a~tree with infinitely many $a$. So $L$ is not a~Boolean combination of open sets (it is easy to see, indeed, that $L$ is a $\bpi{2}$ set but it is not a $\bsigma{2}$ set).
\end{exa}

\begin{lem}%
\label{lem:decidendifference}
Given regular tree languages $L_1, \ldots, L_n$, one can decide who wins the game $\gameH(L_1,\ldots , L_n)$. In particular, given $L$ and $n$, one can decide who wins $\gameHs(L,n)$.
\end{lem}

\begin{proof}
We prove the statement for two regular languages $L_1, L_2$. It is easy to generalise it to $n$ regular languages. The sentences ``Alternator wins the game $\gameH(L_1, L_2)$'' and ``Constrainer wins the game $\gameH(L_1,L_2)$'' can be effectively formalized in Monadic Second\=/order Logic on the complete binary tree. For instance, the sentence for
\[\lq \lq \text{Alternator wins the game } \gameH(L_1, L_2)\text{''} \]
is:
\[\text{there exists a~tree $t_1 \in L_1$ such that for any finite prefix $p$ of $t_1$} \]
\[\text{there exists a~tree $t_2 \in L_2$ that extends $p$}.\]
In a~similar way we can write the sentence that says that Constrainer wins for $n>2$.
\end{proof}

So we obtain:

\begin{cor}%
\label{cor:n-diff-dec}
It is decidable, given a~regular tree language $L$ and $n\in\omega$, whether $L$ is an $n$\=/difference of open sets.
\end{cor}

\begin{proof}
It directly follows from Theorem~\ref{thm:decidforsinglefinitediff} and Lemma~\ref{lem:decidendifference}.
\end{proof}

Hence, since Wadge degrees with Wadge ranks below $\omega$ are formed by Boolean combinations of the levels of the difference hierarchy, we easily obtain:

\begin{thm}
Given a~regular language $L$ and a~Wadge degree ${[A]}_W$ with Wadge rank less than $\omega$, it is decidable if $L$ belongs to ${[A]}_W$.
\end{thm}

\subsection{The infinite variant of the game}%
\label{ssec:inf-game}

We denote by $\gameHinfty(Y)$ the infinite variant of $\gameHs(Y,n)$: $\gameHinfty(Y)$ is the infinite duration game played the same way as $\gameHs(Y,n)$ but the winning condition for Alternator is that he has to survive for infinitely many turns. By $\gameHinfty_U(Y)$ we denote the relativised game.

\begin{rem}
The game $\gameHinfty(Y)$ is determined: every play where Constrainer wins is finite. Therefore, the winning condition for Constrainer is an~open condition, while the winning condition for Alternator is a~closed condition, both in the space\footnote{Since Alternator chooses points in $X$ and Constrainer chooses open sets in $\tau$, each round of the game can be represented as an~element of $X\times\tau$. Thus, the winning condition of the game is a~subset of ${\big(X\times \tau\big)}^\w$.} ${\big(X\times \tau\big)}^\w$, where the sets $X$ and $\tau$ are taken with the discrete topology. Hence, the game is determined by Gale--Stewart Theorem~\cite{gale_games} (see also~\cite{martin_borel_determinacy} for the more general result for Borel sets). 
\end{rem}

\begin{rem}%
\label{rem:can-play-basic-infinite}
In the same fashion as in the case of the finite game, without loss of generality we can assume that Constrainer in his strategies uses only basic open sets, see Corollary~\ref{cor:playing-prefixes}.
\end{rem}

The following fact follows directly from the definition of the two variants of the game.

\begin{fact}%
\label{ft:inf-to-fin-win}
If Alternator wins $\gameHinfty(Y)$ then he wins $\gameHs(Y,n)$ for any $n$.
\end{fact}

\begin{prop}%
\label{pro:fin-to-inf-no}
The converse to Fact~\ref{ft:inf-to-fin-win} is not true, even for regular tree languages $L$.
\end{prop}

\new{
This proposition follows a~posteriori from Theorem~\ref{thm:mainforBC1} and Proposition~\ref{pro:inf-game-char}, using the fact that the involved games are determined and there exist sets in $\bdelta{2}\setminus \bc{\bsigma{1}}$. However, the proof provided here is supposed to provide an~informative illustration of the considered games.
}

\begin{proof}
We have to exhibit a~counterexample. To do that it is convenient to work with an~alphabet with three different symbols, so let $A$ be the alphabet $\{a,b,c\}$.
For the sake of this example, assume that if $t_1, \ldots, t_n$ are trees, by $[t_1, \ldots, t_n]$ we denote the tree

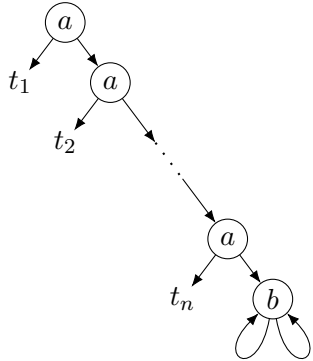
\begin{center}
\begin{tikzpicture}

\node[trnode] (r0) at (-3, 0) {$a$};
\node[toC] (l0) at ($(r0) + (-0.6, -0.8)$) {$t_1$};

\node[trnode] (r1) at ($(r0) + (+0.6, -0.8)$) {$a$};
\node[toC] (l1) at ($(r1) + (-0.6, -0.8)$) {$t_2$};

\coordinate (r2) at ($(r1) + (+0.6, -0.8)$);

\coordinate (r3) at ($(r2) + (+0.6*0.6, -0.8*0.6)$);

\node[trnode] (r4) at ($(r3) + (+0.6, -0.8)$) {$a$};
\node[toC] (l4) at ($(r4) + (-0.6, -0.8)$) {$t_n$};

\node[trnode] (r5) at ($(r4) + (+0.6, -0.8)$) {$b$};

\draw[tredge] (r0) -- (l0);
\draw[tredge] (r1) -- (l1);
\draw[tredge] (r4) -- (l4);

\draw[tredge] (r0) -- (r1);
\draw[tredge] (r1) -- (r2);
\draw[trdots] (r2) -- (r3);
\draw[tredge] (r3) -- (r4);
\draw[tredge] (r4) -- (r5);

\draw[tredge] (r5) edge[trlole] (r5);
\draw[tredge] (r5) edge[trlore] (r5);

\end{tikzpicture}
\end{center}
Now let $L$ be the language of all the trees of the form $[t_1, \ldots, t_n]$ where for every $i \in \{1,\ldots,n\}$ the tree $t_i$ is
\begin{enumerate}
\item either a~tree with every node labelled by $a$,
\item or a~tree that contains only finitely many letters different than $c$, i.e.~a~tree for which there exists a~finite prefix $p$ of $t_i$ such that $t_i(u)=c$ for any node $u \notin p$.
\end{enumerate}
If $t_i$ respects the first condition we say that $t_i$ is a \emph{first case tree}, if it respects the second condition we say it is a \emph{second case tree}.  We first prove that Alternator wins $\gameHs(L,n)$ for any $n$. Fix a~natural number $n$, we provide a~winning strategy for Alternator for the game $\gameHs(L,2n)$. We define the following sets of trees:
\begin{itemize}
\item For $i \in \{1,\ldots, n\}$ let $L_i$ be the set of trees $[t_1, \ldots, t_n]$ such that the trees $t_1, \ldots, t_{i-1}$ are second case trees and the trees $t_i, \ldots, t_n$ are first case trees. It is clear that $L_i \subseteq L$ for any $i$.
\item We define $L'_i$ as $L_i$ except that the tree $t_i$ contains both $a$ and $b$ nodes, but no $c$ nodes. Obviously $L'_i$ is disjoint from $L$.
\end{itemize}
Every prefix of a~tree in $L_i$ can be completed into a~tree in $L'_i$ and every prefix of a~tree in $L'_i$ can be completed into a~tree in $L_{i+1}$. It follows that Alternator wins the game
\[ \gameH(L_1,L'_1, L_2, L'_2, \ldots, L_n, L'_n)\]
and therefore also Alternator wins $\gameHs(L,2n)$.

Now we move to a~proof that Alternator loses $\gameHinfty(L)$. Consider the tree played by
Alternator in the first round. Since this tree belongs to $L$, it must be of the
form $[t_1, \ldots, t_n]$ with $t_1, \ldots, t_n$ either first case trees or second case trees. Let $p_1$ be a~finite prefix of this tree which contains the node $\dR^n$. Constrainer uses a~strategy, which preserves the following properties:

\begin{enumerate}
\item All prefixes played by Constrainer extend the prefix $p_1$. Consequently, all the trees played by Alternator are of the form $[s_1, \ldots, s_n]$. Indeed, the prefix $p_1$ guarantees that the played trees $t$ satisfy $t(\dR^n)=b$ and $t(\dR^k)=a$ for $k<n$. Hence, all the modifications done by Alternator are relative to the trees $t_1, \ldots, t_n$ and  therefore for $k \in \{1, \ldots, n\}$ it is meaningful to talk about the $k$th coordinate of the tree played by Alternator in a~round which refers to the tree $s_k$.

\item Suppose that Alternator plays a~tree $[s_1, \ldots, s_n]$ in some round $i$. Let $p_i$ be a~finite prefix of this tree such that for every coordinate $k \in \{1, \ldots, n\}$ we have:
\begin{itemize}
\item If $s_k$ is a~second case tree then $p_i$ contains a~prefix of $s_k$ such that under that prefix every node is labelled by $c$.
\item If $s_k$ contains some $b$ then $p_i$ contains some $b$ in the subtree $s_k$.
\end{itemize}
In the next round Constrainer chooses $p_i$. Consequently, if $i,j$ are rounds with $i < j$ and $k \in \{1, \ldots,n\}$ then
\begin{itemize}
\item If the $k$th coordinate of Alternator's tree in round $i$ is a~second case tree then also the $k$th coordinate of Alternator's tree in round $j$ has to be a~second case tree.
\item If the $k$th coordinate of Alternator's tree in round $i$ contains a $b$ then also the $k$th coordinate of Alternator's tree in round $j$ contains a $b$.
\end{itemize}
So, in an odd\=/numbered round, Alternator's tree belongs to the language and therefore all the coordinates with a $b$ are second case trees. In an even\=/numbered round, Alternator's tree is outside the language. Therefore, when going from an odd\=/numbered round to the next even\=/numbered round, Alternator must change some coordinate from a~first case tree without $b$ to a~tree with $b$. It follows that the number coordinates with $b$ increases in each even\=/numbered round. Since this can happen at most $n$ times, Alternator must lose after at most $2n$ rounds.
\end{enumerate}
The proof is complete.
\end{proof}

Now we can give a~non\=/effective characterisation of the class $\bdelta{2}$ for topological spaces that are completely metrizable (for the levels $\ndiff{n}{1}$ we gave a~characterisation that holds in general for any topological space, but here we are forced to require complete metrizability).

\begin{prop}%
\label{pro:inf-game-char}
Let $X$ be a~completely metrizable topological space \new{with a~countable basis of the topology (i.e.~$X$ is \emph{second\=/countable})}. The following conditions are equivalent for a~subset $Y$ of $X$:
\begin{enumerate}
\item Constrainer wins $\gameHinfty(Y)$.
\item $Y \in \bdelta{2}(X)$.
\end{enumerate}
\end{prop}

\begin{proof}
The proof is very similar to the analysis of other games of this kind, for instance Banach--Mazur game, see~\cite[Section~8.H]{kechris_descriptive}. 

\paragraph*{\bf Implication $1 \Rightarrow 2$.}

Assume that Constrainer has a~winning strategy $\sigma$ in $\gameHinfty(Y)$. By Remark~\ref{rem:can-play-basic-infinite} we can assume that $\sigma$ plays only basic open sets.

Notice that $\sigma$ (seen as a~tree) is well\=/founded because the strategy is winning and therefore it admits no infinite play. We will prove by induction on the structure of $\sigma$ that if $\langle U_0,x_1,U_1,\ldots,U_{i-1}\rangle $ is a~position compatible with $\sigma$ then $Y\cap U_{i-1}\in\bdelta{2}$. Consider such a~position $P=\langle U_0,x_1,U_1,\ldots,U_{i-1}\rangle$ and assume that the thesis holds for all the positions extending that one. If the position $P$ is instantly winning for Constrainer (i.e.~Alternator cannot play a~single round from $P$) then, depending on parity of $i$, either $U_{i-1}\subseteq Y\complemento$ or $U_{i-1}\subseteq Y$. In both cases the inductive thesis holds. Now assume that $P$ is not instantly winning for Constrainer. By the symmetry lets assume that $i$ is odd, i.e.~Alternator is forced to play $x_i\in U_{i-1}\cap Y$. Let ${(B_x)}_{x\in U_{i-1}\cap Y}$ be the indexed family of basic open sets $B_x$ played by $\sigma$ as a~response to Alternator playing $x$. By the inductive assumption we know that for each $x\in U_{i-1}\cap Y$ we have $Y\cap B_x\in\bdelta{2}$.

By the definition of the family $B_x$ we know that
\[Y\cap U_{i-1}=\bigcup_{x\in U_{i-1}\cap Y} \big(B_x\cap Y\big),\]
where the union is in fact countable since there is only countably many basic open sets in $X$. As every set taken in the union is $\bsigma{2}$ (in fact $\bdelta{2}$) we know that $Y\cap U_{i-1}$ is $\bsigma{2}$. Dually
\[Y\complemento\cap U_{i-1}=\left(U_{i-1}\setminus \bigcup_{x\in U_{i-1}\cap Y} B_x\right)\ \cup\ \bigcup_{x\in U_{i-1}\cap Y} \big(B_x\cap Y\complemento\big),\]
which again is a $\bsigma{2}$ presentation of $Y\complemento\cap U_{i-1}$.

Thus, the above induction implies that $Y\cap U_0=Y\cap X = Y$ is $\bdelta{2}$.

\paragraph*{\bf Implication $2 \Rightarrow 1$.}

We need to prove that if $Y \in \bdelta{2}$ then Constrainer wins $\gameHinfty(Y)$. Indeed, if $Y$ is in $\bdelta{2}$ we can write $Y$ and its complement as
\[ Y= \bigcap_{j \in \omega} A_j\quad \text{ and }\quad Y\complemento = \bigcap_{j \in \omega} B_j,\]
where all the sets $A_i$ and $B_i$ are open. Now we can describe a~winning strategy for Constrainer in $\gameHinfty(Y)$. Suppose $i=1,2,\ldots$ is the round we are playing and $i$ is odd (resp. $i$ is even). Let $j=\lfloor \frac{i-1}{2} \rfloor$. Assume that $U_{i-1}$ is the open set that was played last ($U_0=X$) and let $x_i$ be the point played by Alternator in the current round. By the definition of the game, if $i$ is odd then $x_i\in Y$ and otherwise $x_i\in Y\complemento$. Let Constrainer play $U_i$ such that $\overline{U_i}\subseteq U_{i-1}$; $U_i\subseteq A_j$ (resp. $U_i\subseteq B_j$); and the diameter of $U_i$ is smaller than $2^{-i}$. Such a~set exists because $x_i\in U_{i-1}\cap A_j$ (resp. $x_i\in U_{i-1}\cap B_j$).

Clearly it is a~valid strategy of Constrainer. Consider an infinite play consistent with this strategy. Since $X$ is Polish and the sets $U_i$ are of decreasing diameter with $\overline{U_i}\subseteq U_{i-1}$, there must exists $x\in\bigcap_{n\in\w} U_i$. But by the construction of $U_i$, such $x$ must belong to both $\bigcap_{j\in\w}A_j=Y$ and $\bigcap_{j\in\w}B_j=Y\complemento$, a~contradiction. Thus, each play consistent with the above strategy is finite and therefore winning for Constrainer.
\end{proof}

\begin{cor}
Let $X$ be a~completely metrizable topological space and $Y$ be a~subset of~$X$. Then $Y \in \bdelta{2} \setminus \boolcomb{1}$ if and only if Alternator wins $\gameHs(Y,n)$ for all $n$ but he loses $\gameHinfty(Y)$.
\end{cor}

\section{The algebra on trees}%
\label{sec:algebra}

Our next goal is decidability of the class $\bigcup_{n \in \omega} \ndiff{n}{1}=\boolcomb{1}$. Notice that Corollary~\ref{cor:n-diff-dec} gives us a~semi\=/algorithm for deciding if a~regular language is in $\bigcup_{n \in \omega} \ndiff{n}{1}$. Indeed, for $n = 1, 2, \ldots$ we can use Corollary~\ref{cor:n-diff-dec} to decide if $L\in\ndiff{n}{1}$. If for some $n$ this is the case then $L\in\bigcup_{n \in \omega} \ndiff{n}{1}$ and the algorithm terminates. Otherwise, the algorithm does not terminate. Section~\ref{sec:main-thm} provides an alternative algebraic algorithm that always terminates and solves the above decision problem.

The algebraic approach we define in this section is based on the so\=/called \emph{Myhill--Nerode equivalence} that allows to distinguish trees based on their behaviour when put into certain \emph{contexts}. 

\begin{defi}
A \emph{multicontext} over an alphabet $A$ is a~partial tree $t$ over $A\sqcup \{\hole\}$ where $\hole\notin A$ such that: $t$ does not contain any unary nodes and a~node of $t$ is a~leaf if and only if it is labelled $\hole$. A \emph{port} of a~multicontext $C$ is any node of $t$ that is labelled $\hole$ (i.e.~any leaf of $t$).
\end{defi}

The number of ports is called the \emph{arity} of the multicontext. \emph{A~priori} a~multicontext may have infinitely many ports and in that case the arity is $\infty$. A~multicontext with exactly one port is called a~\emph{context}. Given a~multicontext $C$ and a~valuation $\eta$ which maps ports of $C$ to trees in $\trees_A$, we write $C[\eta]$ for the tree obtained by replacing each port $u$ by the tree $\eta(u)$. The tree $C[\eta]$ is said to \emph{extend} the multicontext $C$. If $L$ is a~set of trees and $C$ is a~multicontext then by $C[L]$ we denote the set of trees obtained by plugging trees of $L$ in the ports of $C$ in all the possible ways. The set of all trees extending a~multicontext $C$ is denoted by $C[\ast]$. If $C$ is a~multicontext, possibly with infinitely many ports, and $t$ is a~tree, we denote by $C[t]$ the tree obtained by putting $t$ in every port of $C$.

\begin{exa}
The multicontext $C_0$ consists of only one node --- the root. It is called the \emph{trivial context} and denoted $\hole$. $C_1$ is a~tree, it has no ports, and $C_1[\ast]$ is $\{C_1\}$. The multicontext $C_2$ is a~context and if we complete $C_2$ with a~tree we obtain a~tree where the root label is $a$ and the left subtree of the root is labelled with only letters $b$. Finally, $C_3$ is a~finite multicontext and $C_3[\ast] = \{ t \mid t(\epsilon)=a\}$.
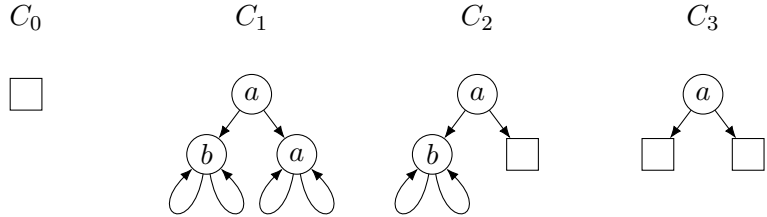
\begin{center}
\begin{tikzpicture}
\node[trport] (z0) at (-6, 0) {};
\node[toC] at ($(z0)+(0, 1.0)$) {$C_0$};

\node[trnode] (z1) at (-3, 0) {$a$};
\node[toC] at ($(z1)+(0, 1.0)$) {$C_1$};
\node[trnode] (l1) at ($(z1) + (-0.6, -0.8)$) {$b$};
\node[trnode] (r1) at ($(z1) + (+0.6, -0.8)$) {$a$};

\draw[tredge] (z1) -- (l1);
\draw[tredge] (z1) -- (r1);
\draw[tredge] (l1) edge[trlole] (l1);
\draw[tredge] (l1) edge[trlore] (l1);
\draw[tredge] (r1) edge[trlole] (r1);
\draw[tredge] (r1) edge[trlore] (r1);

\node[trnode] (z2) at (-0, 0) {$a$};
\node[toC] at ($(z2)+(0, 1.0)$) {$C_2$};
\node[trnode] (l2) at ($(z2) + (-0.6, -0.8)$) {$b$};
\node[trport] (r2) at ($(z2) + (+0.6, -0.8)$) {};

\draw[tredge] (z2) -- (l2);
\draw[tredge] (z2) -- (r2);
\draw[tredge] (l2) edge[trlole] (l2);
\draw[tredge] (l2) edge[trlore] (l2);

\node[trnode] (z3) at (+3, 0) {$a$};
\node[toC] at ($(z3)+(0, 1.0)$) {$C_3$};
\node[trport] (l3) at ($(z3) + (-0.6, -0.8)$) {};
\node[trport] (r3) at ($(z3) + (+0.6, -0.8)$) {};

\draw[tredge] (z3) -- (l3);
\draw[tredge] (z3) -- (r3);
\end{tikzpicture}
\end{center}
\end{exa}

Now we focus on contexts, i.e.~multicontexts with exactly one port. We write $\Context$ for the set of all non\=/trivial contexts over $A$. Given two contexts $C,D$ we write $C\cdot D$ for the context obtained by replacing the port of $C$ with $D$. Moreover, if $C\neq\hole$ (i.e.~$C$ is non\=/trivial) then we write $C^\infty$ for the infinite tree
\[ C \cdot C \cdot C \cdot C \cdots\]
\new{Notice that the above definition does not construct a~tree for $C=\hole$ (the trivial context). That is why we restrict the set of contexts $\Context$ to only non\=/trivial ones.}

\begin{rem}
It is easy to verify that $\cdot$ is \emph{associative}, therefore $\big(\Context\sqcup\{\hole\},\cdot\big)$ is an infinite monoid, with ${\cdot}$ associative and the trivial context $\hole$ as the neutral element.
\end{rem}

Now we define the two Myhill--Nerode equivalence relations: one for trees and one for contexts. These equivalence relations depend on a~fixed language $L\subseteq \trees_A$. 

\begin{defi}%
\label{def:myhill-nerode}
In the Myhill--Nerode equivalence for trees, we say that two trees $t$ and $t'$ are $L$\=/\emph{equivalent} \new{(denoted $t\Teq_L t'$)} if 
\[ C[t] \in L \Longleftrightarrow C[t'] \in L\quad \text{ for every multicontext $C$.}\]
\end{defi}

\begin{rem}%
\label{rem:L-equi}
\new{
Notice that because of the duality between $L$ and the complement $L^\cp$ in the above definition, we know that the equivalence relations $\Teq_L$ and $\Teq_{L^\cp}$ coincide. Moreover, since $\hole$ (i.e.~the trivial context) is a~multicontext we know that if $t\in L$ and $t'\in L^\cp$ then $t\not\Teq_L t'$.
}
\end{rem}

\new{
The above remark says that the equivalence relation $\Teq_L$ is always able to distinguish trees from $L$ from those not in $L$. In the following example these are the only two classes of $\Teq_L$, i.e.~$\Teq_L$ is not able to distinguish anything else.}

\begin{exa}
Consider the language $L = \{ t \in \trees_A \mid t(\epsilon) = a\}$. In this case we have just two equivalence classes that are $L$ and the complement $L\complemento$. The multicontext that establishes if a~tree belongs to $L$ or $L\complemento$ is the trivial context $\hole$.
\end{exa}

To give a~similar definition for contexts, we use a~variant of multicontexts where the ports can be substituted by contexts and not trees.

\begin{defi}
A \emph{context environment} over an alphabet $A$ is a~partial tree labelled by $A\sqcup\{\hole\}$ such that: $t$ has no leaves and a~node of $t$ is unary if and only if it is labelled $\hole$. A~\emph{port} of a~context environment is any node of $t$ that is labelled $\hole$ (i.e.~any unary node of $t$).
\end{defi}

Given a~context environment $E$ and a~non\=/trivial context $C$, we write $E[C]$ for the tree obtained by substituting $C$ for every port of $E$ in the following way. \new{For each port $v$ of $E$ we first plug a~copy of $C$ into that port and then plug the subtree of $E$ below that port into the port of $C$. Notice that there may be more than one occurrence of $\hole$ on a~branch of $E$ and the above plugging needs to be performed in each of them, see the left\=/most branch of $E[C]$ in Example~\ref{ex:plugging}. This guarantees that $E[C]$ is a~tree --- it does not contain any port of~$E$ nor any of the copies of the port of~$C$.}

\new{The assumption that $C$ is non\=/trivial is important in the above construction, because if $E$ is a~unary tree labelled everywhere $\hole$ (i.e.~a~$\hole$\=/labelled infinite path) and $C=\hole$ is the trivial context, then $E[C]$ (if defined at all) does not contain any node.}

\begin{exa}%
\label{ex:plugging}
In the following figure, $C$ is a~context, $E$ is a~context environment, and $E[C]$ is a~tree.
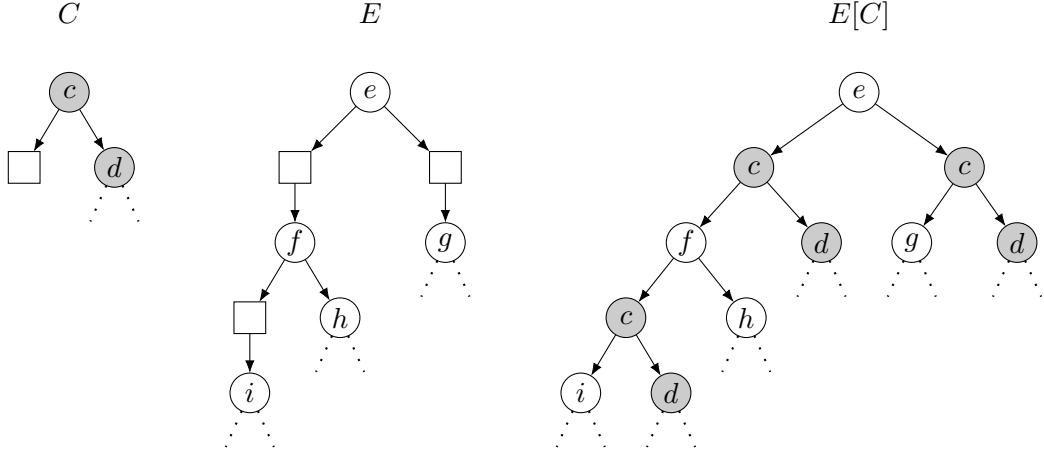
\begin{center}
\begin{tikzpicture}

\newcommand{\dotize}[1]{
\coordinate (dotzl) at ($(#1) + (+0.4*0.8, -0.8*0.9)$);
\coordinate (dotzr) at ($(#1) + (-0.4*0.8, -0.8*0.9)$);

\draw[trdots] (#1) -- (dotzl);
\draw[trdots] (#1) -- (dotzr);
}

\tikzstyle{szary}=[fill=gray!40]

\newcommand{\vstr}{1.0}

\node[trnode, szary] (z1) at (-10, 0) {$c$};
\node[toC] at ($(z1)+(0, 1.0)$) {$C$};
\node[trnode, szary] (l1) at ($(z1) + (+0.6, -\vstr)$) {$d$};
\node[trport] (r1) at ($(z1) + (-0.6, -\vstr)$) {};

\dotize{l1}

\draw[tredge] (z1) -- (l1);
\draw[tredge] (z1) -- (r1);

\node[trnode] (z2) at (-6, 0) {$e$};
\node[toC] at ($(z2)+(0, 1.0)$) {$E$};
\node[trport] (l2) at ($(z2) + (-1.0, -\vstr)$) {};
\node[trport] (r2) at ($(z2) + (+1.0, -\vstr)$) {};

\draw[tredge] (z2) -- (l2);
\draw[tredge] (z2) -- (r2);

\node[trnode] (ll2) at ($(l2) + (-0.0, -\vstr)$) {$f$};
\node[trnode] (rr2) at ($(r2) + (+0.0, -\vstr)$) {$g$};

\draw[tredge] (l2) -- (ll2);
\draw[tredge] (r2) -- (rr2);

\node[trport] (lll3) at ($(ll2) + (-0.6, -\vstr)$) {};
\node[trnode] (llr3) at ($(ll2) + (+0.6, -\vstr)$) {$h$};

\node[trnode] (llr4) at ($(lll3) + (+0.0, -\vstr)$) {$i$};

\draw[tredge] (ll2) -- (lll3);
\draw[tredge] (ll2) -- (llr3);
\draw[tredge] (lll3) -- (llr4);

\dotize{llr3}
\dotize{rr2}
\dotize{llr4}

\node[trnode] (z3) at (0.5, 0) {$e$};
\node[toC] at ($(z3)+(0, 1.0)$) {$E[C]$};
\node[trnode, szary] (l3) at ($(z3) + (-1.4, -\vstr)$) {$c$};
\node[trnode, szary] (r3) at ($(z3) + (+1.4, -\vstr)$) {$c$};
\node[trnode, szary] (lr3) at ($(l3) + (+0.9, -\vstr)$) {$d$};
\node[trnode, szary] (rr3) at ($(r3) + (+0.7, -\vstr)$) {$d$};
\node[trnode] (ll3) at ($(l3) + (-0.9, -\vstr)$) {$f$};
\node[trnode] (rl3) at ($(r3) + (-0.7, -\vstr)$) {$g$};

\node[trnode] (llr4) at ($(ll3) + (+0.8, -\vstr)$) {$h$};
\node[trnode, szary] (lll4) at ($(ll3) + (-0.8, -\vstr)$) {$c$};

\node[trnode, szary] (lllr5) at ($(lll4) + (+0.6, -\vstr)$) {$d$};
\node[trnode] (llll5) at ($(lll4) + (-0.6, -\vstr)$) {$i$};

\dotize{lr3}
\dotize{rr3}
\dotize{rl3}
\dotize{llr4}
\dotize{lllr5}
\dotize{llll5}

\draw[tredge] (z3) -- (l3);
\draw[tredge] (z3) -- (r3);
\draw[tredge] (l3) -- (ll3);
\draw[tredge] (l3) -- (lr3);
\draw[tredge] (r3) -- (rl3);
\draw[tredge] (r3) -- (rr3);
\draw[tredge] (ll3) -- (lll4);
\draw[tredge] (ll3) -- (llr4);
\draw[tredge] (lll4) -- (llll5);
\draw[tredge] (lll4) -- (lllr5);

\end{tikzpicture}
\end{center}
\end{exa}

\begin{defi}%
\label{def:cont-equiv}
We define two non\=/trivial contexts $C$ and $C'$ to be $L$\=/\emph{equivalent} \new{(denoted $C\Ceq_L C'$)} if
\[ E[C] \in L \Longleftrightarrow E[C'] \in L\quad \text{ for every context environment $E$}.\]
\end{defi}

\subsection{The syntactic tree algebra \texorpdfstring{$(H_L,V_L)$}{(HL,VL)}}

We denote by $H_L$ the set of the equivalence classes of trees with respect to $L$, the elements of $H_L$ are called \emph{tree types}. Similarly we denote by $V_L$ the set of the equivalence classes of non\=/trivial contexts with respect to $L$, the elements of $V_L$ are called \emph{context types}. \new{By $\one_L$ we denote an~additional type corresponding to the trivial context $\hole$. Since $V_L$ denotes types of the non\=/trivial contexts, in general $\one_L\notin V_L$}.

We will now prove a~number of results, showing that the sets $H_L$ and $V_L$ bear certain algebraic structure. We start with a~simple fact following directly from the definition of $L$\=/equivalence.

\begin{fact}%
\label{ft:equi-compose}
For every multicontext $D$ and context environment $E$, the following two operations preserve the $L$\=/equivalence:
\begin{enumerate}
\item the operation $t \rightarrow D[t]$, i.e.~if $t\Teq_L t'$ then $D[t]\Teq_L D[t']$,
\item the operation $C \rightarrow E[C]$, i.e.~if $C\Ceq_L C'$ then $E[C]\Teq_L E[C']$.
\end{enumerate}
\end{fact}

\noindent
\new{The following corollary follows directly from Fact~\ref{ft:equi-compose}.}

\begin{cor}%
\label{cor:lift-types}
\new{
Given a~multicontext $D$, a~context environment $E$, and elements $h\in H_L$, $v\in V_L$, the sets
\begin{align*}
D[h] &\eqdef \{D[t]\mid t\in h\}\subseteq \trees_A,\\
E[v] &\eqdef \{E[C]\mid C\in v\}\subseteq \trees_A.
\end{align*}
are $L$\=/equivalence classes, i.e.~elements of $H_L$.
}
\end{cor}

\new{Thus, each multicontext $D$ (resp.~context environment $E$) induces an~operation $H_L\to H_L$ (resp.~$V_L\to H_L$).}

As expressed by the next lemma, the following natural operations on contexts and trees respect the $L$\=/equivalence.

\begin{lem}%
\label{lem:preserve}
The following operations respect $L$\=/equivalence.
\begin{enumerate}
\item The composition of contexts $(C_1,C_2) \rightarrow C_1 \cdot C_2$.
\item Substituting a~tree in the port of a~context $(C,t) \rightarrow C[t]$.
\item Infinite iteration of a~non\=/trivial context $C \rightarrow C^\infty$.
\item For every symbol $c \in A$, the operations $t \mapsto c(t,\hole)$ and $t \mapsto c(\hole, t)$, that produce new contexts with roots labelled $c$, and $t$ plugged as the left or the right subtree under the root, respectively.
\end{enumerate}
\end{lem}

\begin{proof}
We prove all the items.
\paragraph*{\bf Item~1} Let $C_1$, $C_2$, $D_1$, $D_2$ be contexts such that $C_1$ is $L$\=/equivalent with $D_1$ and $C_2$ is $L$\=/equivalent with $D_2$. Let $E$ be some context environment, we prove that
\[E[C_1\cdot C_2] \in L \Leftrightarrow E[D_1\cdot D_2]  \in L.\]
We construct from $E$, $C_1$ and $E$, $D_2$ respectively two new context environments
$E_C$ and $E_D$ such that
\[E_C[C_2] = E[C_1\cdot C_2] \text{ and } E_D[D_1] =E[D_1\cdot D_2].\]
\new{
By the symmetry, lets focus on $E_C$. This context environment is obtained from $E$ by inserting into each occurrence of $\hole$ a~whole copy of $C_1$ (including its own $\hole$). Thus, the overall number of holes of $E_C$ and $E$ is the same and the above equality follows.
}

Using the $L$\=/equivalence, we have:
\[
\begin{array}{lll}
E_D[C_1] \in L & \Leftrightarrow & E_D[D_1] \in L \\
E_C[C_2] \in L & \Leftrightarrow & E_C[D_2] \in L
\end{array}
\]
\noindent Note that by the definition of $E_C$ and $E_D$ we have $E_C[D_2]=E_D[C_1]$. It follows that
\[E[C_1\cdot C_2] \in L \Leftrightarrow E_C[C_2] \in L \Leftrightarrow E_C[D_2] \in L \Leftrightarrow E_D[C_1] \in L \Leftrightarrow E_D[D_1] \in L.\]
Therefore, $C_1\cdot C_2$ and $D_1\cdot D_2$ are $L$\=/equivalent.

\paragraph*{\bf Item~2} Let $C$, $C'$ be $L$\=/equivalent contexts and $t$, $t'$ be $L$\=/equivalent trees. We want to prove that $C[t]$ and $C'[t']$ are two $L$\=/equivalent trees. Let $D$ be a~generic multicontext. We prove that
\[D[C[t]] \in L \Leftrightarrow D[C'[t']]  \in L.\]
Let $D'$ be a~multicontext such that $D'[t]=D[C[t]]$ and $E$ a~context environment such that
$E[C'] = D[C'[t']]$. Using the $L$\=/equivalence we have:

\[
\begin{array}{lll}
D'[t] \in L & \Leftrightarrow & D'[t'] \in L, \\
E [C] \in L & \Leftrightarrow & E[C'] \in L.
\end{array}
\]

\noindent By the definition of $E$ and $D'$ we have $D'[t']=E[C]$. It follows that
\[D[C[t]] \in L \Leftrightarrow
D'[t] \in L \Leftrightarrow
D'[t'] \in L \Leftrightarrow
E[C] \in L \Leftrightarrow
E[C'] \in L \Leftrightarrow
D[C'[t']] \in L.\]
Therefore, $C[t]$ and $C'[t']$ are $L$\=/equivalent.

\paragraph*{\bf Item~3} Let $C$, $C'$ be $L$\=/equivalent non\=/trivial contexts. We want to prove that $C^\infty$ and $C'^\infty$ are two $L$\=/equivalent trees. Let $D$ be some
multicontext, we prove that
\[ D[C^\infty] \in L \Leftrightarrow D[C'^\infty] \in L.\]
Consider a~context environment $E$ constructed from $D$ by replacing each port of $D$ with an infinite chain of ports. Using $L$\=/equivalence of $C$ and $C'$ we get:

\[
E[C] \in L \Leftrightarrow E[C'] \in L
\]

\noindent
By definition of $E$ we have $E[C]=D[C^\infty]$ and $E[C']=D[C'^\infty]$. It follows that
\[D[C^\infty] \in L \Leftrightarrow D[C'^\infty] \in L.\]
Therefore, $C^\infty$ and $C'^\infty$ are $L$\=/equivalent.

\paragraph*{\bf Item~4} We only do the proof for $t \mapsto c(t,\hole)$, the other operations is handled symmetrically. Let $t,t'$ be $L$\=/equivalent trees. Let $E$ be some context environment, we prove that
\[E[c(t,\hole)] \in L \Leftrightarrow E[c(t',\hole)] \in L.\]
By inserting $c$ into the ports of $E$ we construct a~multicontext $C$ such that for all trees $s$,
$C[s]=E[c(s,\hole)]$. Using $L$\=/equivalence of $t$ and $t'$ we get:
\[
C[t] \in L \Leftrightarrow C[t'] \in L
\]
Therefore, $c(t,\hole)$ and $c(t',\hole)$ are $L$\=/equivalent.
\end{proof}

\begin{cor}%
\label{cor:algebraic-struct}
The above operations on $\trees_A$ and $\Context$ induce the following algebraic structure on $(H_L, V_L)$:
\begin{itemize}
\item composition of context types $V_L\ni u,v\longmapsto u\cdot v\in V_L$,
\item action of $V_L$ on $H_L$ i.e.~$V_L\times H_L\ni (u,h)\longmapsto u\cdot h\in H_L$,
\item infinite composition $V_L\ni u \longmapsto u^\infty \in H_L$,
\item creation of contexts $H_L\ni h \longmapsto c(\hole, h), c(h,\hole)\in V_L$ for $c\in A$.
\end{itemize}
Moreover, $(V_L,{\cdot})$ is a~semigroup acting over $H_L$ via $\cdot$, $\big(V_L\sqcup\{\one_L\},\cdot\big)$ is a~monoid, and $(H_L,V_L)$ satisfy the axioms of a~Wilke algebra~\cite{wilke_algebraic}.
\end{cor}

\new{
A~pair of sets $(H,V)$, equipped with the operations as in Corollaries~\ref{cor:lift-types} and~\ref{cor:algebraic-struct}, is called a~\emph{tree algebra}. Once restricted to the operations from Corollary~\ref{cor:algebraic-struct}, it becomes a~\emph{thin algebra}\footnote{The name comes from the theory of \emph{thin trees}, sometimes called \emph{scattered trees}, see~\cite{idziaszek_thin_journal}.}. An~advantage of a~thin algebra is that (given that $H$ and $V$ are finite) it can be represented effectively on a~computer by giving the sets and the multiplication tables for the five involved operations. Because of the simplicity of the considered classes of topological complexity, the operations involved in our characterisations involve only operations of thin algebras, see Equations~\eqref{eq:cond-bool} and~\eqref{eq:cond-limit}.
}

Notice that the pair $(\trees_A,\Context)$ with the actual operations of composition of trees and contexts is in fact a~tree algebra, we call it the \emph{free} tree algebra.

\new{
The \emph{syntactic morphism} of a~language $L$, denoted by $\alpha_L$, is the two\=/sorted function
\[ \fun{\alpha_L}{(\trees_A,\Context)}{(H_L,V_L)},\]
that maps a~tree $t\in\trees_A$ into its \new{${\Teq_L}$\=/equivalence} class $\alpha_L(t)\in H_L$ \new{and a~context $C\in\Context$ into its ${\Ceq_L}$\=/equivalence class $\alpha_L(C)\in V_L$}.
}

The next fact is considered folklore, see for instance~\cite[Proposition~1.11 in Annex~A on page~442]{perrin_pin_words}. The usual notation for the number $\monoid$ is $\w$, however \new{it is better not to use} that symbol in this paper to avoid confusion with the infinite repetition.

\begin{fact}%
\label{fact_monoid}
Given any finite semigroup $V$, there is a~number $\monoid_V$ (denoted just $\monoid$ if $V$ is known from the context) such that for each element $v$ of $V$ the element $v^\monoid$ is an idempotent, i.e.~$v^\monoid = v^\monoid \cdot v^\monoid$.
\end{fact}

\subsection{Congruences}

\new{
To be able to compute the syntactic tree algebra $(H_L,V_L)$ we will need to start with its approximation. This notion is formalised as follows.
}

\begin{defi}
Consider a~two\=/sorted function $\fun{\alpha}{(\trees_A,\Context)}{(H,V)}$ that is surjective onto a~pair of sets $(H,V)$. We say that it is a~\emph{congruence} if the following two conditions hold:
\begin{itemize}
\item for every multicontext $D$ and every pair of trees $t,t'\in\trees_A$ such that $\alpha(t)=\alpha(t')$ we have
\[\alpha\big(D[t]\big)=\alpha\big(D[t']\big),\]
\item for every context environment $E$ and every pair of contexts $C,C'\in\Context$ such that $\alpha(C)=\alpha(C')$ we have
\[\alpha\big(E[C]\big)=\alpha\big(E[C']\big).\]
\end{itemize}
We say that $\alpha$ \emph{recognises} a~set of trees $L\subseteq\trees_A$ if $L=\alpha^{-1}(S)$ for some subset $S\subseteq H$.
\end{defi}

\begin{rem}%
\label{rem:cong-to-finer}
If $\alpha$ is a~congruence that recognises $L$ then its kernel is finer than the $L$\=/equivalence in the sense that for each pair of trees $t$, $t'$ and pair of contexts $C$, $C'$ we have:
\[\alpha(t)=\alpha(t') \Rightarrow t\Teq_L t'\quad\text{ and }\quad \alpha(C)=\alpha(C')\Rightarrow C\Ceq_L C'.\]
\end{rem}

\begin{proof}
\new{
Assume that $S\subseteq H$ is such that $\alpha^{-1}(S)=L$. Consider the first claim and take two trees $t,t'\in\trees_A$ such that $\alpha(t)=\alpha(t')$. We need to prove that $t\Teq_L t'$. Take a~multicontext~$D$ and assume by the symmetry that $D[t]\in L$, i.e.~$\alpha\big(D[t]\big)\in S$. However, the assumption that $\alpha$ is a~congruence implies that $\alpha\big(D[t']\big)=\alpha\big(D[t]\big)$ and therefore $D[t']\in \alpha^{-1}(S)=L$.
}
\end{proof}

\new{
Corollary~\ref{cor:lift-types} implies that the syntactic morphism $\fun{\alpha_L}{(\trees_A,\Context)}{(H_L,V_L)}$ is a~congruence. Lemma~\ref{lem:preserve} implies that $\alpha_L$ is a~homomorphism of thin algebras, i.e.~it commutes with the operations of these algebras --- for instance for the operation $v\cdot v'$ one needs to observe that for every two contexts $C,C'\in\Context$ we have
\[\alpha_L(C)\cdot \alpha_L(C') = \alpha_L\big(C\cdot C'\big).\]
Additionally, Remark~\ref{rem:L-equi} implies that $\alpha_L$ recognises~$L$.
}

\subsection{Computing \texorpdfstring{$(H_L,V_L)$}{(HL,VL)}}

We will now prove that the syntactic algebras can be constructed effectively, as stated by the following proposition.

\begin{prop}%
\label{pro:alg-finite-for-reg}
If $L$ is a~regular language then both $H_L$ and $V_L$ are finite. Moreover, given a~parity non\=/deterministic automaton $\Aa$ recognising $L$, one can compute in \EXPTIME the syntactic algebra $(H_L,V_L)$ for $L$, together with the structure of thin algebra on $(H_L,V_L)$. Additionally, both $H_L$ and $V_L$ have at most exponentially many elements in the number of states of $\Aa$.
\end{prop}

The following remark shows that finiteness of $H_L$ and $V_L$ is not sufficient for regularity (the remark is motivated by Running Example~2 in~\cite{bojanczyk_monads_dlt}).

\begin{rem}
Both $H_L$ and $V_L$ are finite for any language defined in \msou (see~\cite{bojanczyk_bounding}).
\end{rem}

\new{
The rest of this subsection is devoted to a~proof of Proposition~\ref{pro:alg-finite-for-reg}. We follow implicitly the approach from~\cite{bojanczyk_monads_dlt}, see~\cite[Theorem~3.1]{bojanczyk_monads_arxiv} in the full version of the paper. A~very similar construction is also given in~\cite[Appendices~D.1 and~D.2]{idziaszek_ef} or~\cite{idziaszek_phd}.
}

\new{
The proof is divided into two stages. The first stage shows how to construct some congruence $\fun{\alpha}{(\trees_A,\Context)}{(H,V)}$ recognising $L$, such that the sets $H$ and $V$ are at most exponential in the number of states of $\Aa$ --- a~non\=/deterministic automaton recognising $L$. The second stage merges certain elements of that congruence to obtain the syntactic algebra $(H_L,V_L)$.
}

\paragraph*{\bf The first stage --- a~congruence $\alpha$} We assume the notions of parity automata over infinite trees as in~\cite{thomas_languages}: such an~automaton is a~tuple $\Aa=\langle A, Q, q_\init, \delta,\Omega\rangle$, where $A$ is an~alphabet, $Q$ is a~set of states, $q_\init\in Q$ is an~initial state, $\delta\subseteq Q\times A \times Q^2$ is a~transition relation and $\fun{\Omega}{Q}{\{i,\ldots,j\}}$ is a~priority assignment. We use the standard concepts of an~accepting run of a~given automaton over a~given tree.

Define a~two\=/sorted function $\fun{\alpha}{(\trees_A,\Context)}{\Big(2^Q,2^{Q\times\{i,\ldots,j\}\times Q}\Big)}$ as follows:
\begin{align*}
\alpha(t) &\eqdef \big\{q\in Q\mid \text{$\Aa$ has an accepting run over $t$ from $q$}\big\},\\
\alpha(C) &\eqdef \big\{(q,\ell,q')\in Q\times\{i,\ldots,j\}\times Q\mid\\
&\qquad\qquad\text{$\Aa$ has an accepting run over $C$ that:}\\
&\qquad\qquad\qquad \text{starts from $q$ in the root of $C$,}\\
&\qquad\qquad\qquad \text{reaches $q'$ in the port of $C$,}\\
&\qquad\qquad\qquad \text{and the maximal priority seen in that run}\\
&\qquad\qquad\qquad\qquad \text{on the path to from the root to the port equals $\ell$}\big\}.\\
\end{align*}

\new{
Let $(H,V)$ be the range of $\alpha$. Intuitively, the value $\alpha(t)$ says from which states the automaton $\Aa$ can accept a given tree. In particular, $t\in L$ if and only if $q_\init\in\alpha(t)$ (thus, $\alpha$ recognises $L$). Similarly, the value $\alpha(C)$ for a context $C$ contains information about the relation between the states $\Aa$ can have in the root and in the port of $C$. However, as our algebra needs to deal with the operation $C^\infty$, we additionally need to know what is the maximal priority on the considered path, to make sure that it satisfies the parity condition once repeated infinitely many times.
}

\begin{claim}
$\alpha$ is a~congruence.
\end{claim}

\begin{proof}
\new{
We start with the first bullet of the definition (the second bullet speaking about contexts is analogous). Consider two trees $t,t'\in \trees_A$ such that $\alpha(t)=\alpha(t')=h\in H$ and let $D$ be a~multicontext. Our aim is to prove that $\alpha\big(D[t]\big)=\alpha\big(D[t']\big)$. Assume by the symmetry that $q\in \alpha\big(D[t]\big)$, i.e.~the tree $D[t]$ can be accepted from a~state $q\in Q$. We need to prove that $q\in \alpha\big(D[t']\big)$ (the symmetric case is analogous). There exists a~run of $\Aa$ over $D[t]$ that is accepting and starts from the state $q$ in the root of $D[t]$. By the assumption that $\alpha(t)=\alpha(t')$ we can construct a new run of $\Aa$ over $D[t']$ that is also accepting and starts from $q$, it is enough to replace the runs over $t$ by the respective runs over $t'$, below each of the ports of $D$.
}
\end{proof}

We will now define explicitly the operations of thin algebra on $(H,V)$, for $h\in H_{\Aa}$, $v,v'\in V_{\Aa}$, and $c\in A$:
\begin{align*}
v\cdot v' &\eqdef \{(q,\max(\ell,\ell'), q'')\mid (q,\ell,q')\in v, (q',\ell',q'')\in v'\},\\
v\cdot h &\eqdef \{q\in Q\mid (q,\ell,q')\in v, q' \in h\},\\
v^\infty &\eqdef {(v^\monoid)}^\infty &\text{see Fact~\ref{fact_monoid},}\\
e^\infty &\eqdef \{q\in Q\mid (q,\ell,q')\in e, (q', 2\ell',q')\in e\}&\text{for $e\in V$ idempotent,}\\
c(\hole,h) &\eqdef \{(q,\Omega(q), q_{\dL})\mid (q,c,q_\dL,q_\dR)\in \delta, q_\dR\in h\},\\
c(h,\hole) &\eqdef \{(q,\Omega(q), q_{\dR})\mid (q,c,q_\dL,q_\dR)\in \delta, q_\dL\in h\}.
\end{align*}

It is relatively easy to check that the above operations are compatible with $\alpha$, i.e.~$\alpha$ is a~homomorphism of thin algebras. Thus, we have constructed a~congruence $\fun{\alpha}{(\trees_A,\Context)}{(H,V)}$ that recognises $L$ and additionally we have given explicitly a~structure of thin algebra on $(H,V)$.

\paragraph*{\bf The second stage --- a~quotient of $(H,V)$.}
We now provide a~sketch of the second stage of the construction: computing the equivalence relation $\approx$ on $(H,V)$ defined as:
\begin{equation}
t\Teq_L t' \Leftrightarrow \alpha(t) \approx \alpha(t')\quad\text{ and }\quad C\Ceq_L C' \Leftrightarrow \alpha(C) \approx \alpha(C').
\label{eq:def-of-approx}
\end{equation}
First notice that this is a~correct definition that does not depend on the choice of witnesses because of Remark~\ref{rem:cong-to-finer}.

\new{
For the sake of simplicity we will focus on the case of trees, i.e.~we need to compute in \EXPTIME if $h\approx h'$. However, Equation~\eqref{eq:def-of-approx} says that $h\not\approx h'$ if and only if there exists a~multicontext $D$ such that $D[h] \in L \Leftrightarrow D[h']\notin L$ (the question whether $D[h]\in L$ is well\=/posed because $\alpha$ is a~congruence recognising $L$).
}

\new{
This goal is achieved by constructing a~non\=/deterministic parity tree automaton $\Cc_{h}$ of size polynomial in $\Aa$, that recognises the language of those multicontexts $D$ such that $D[h] \in L$. The automaton $\Cc_h$ non\=/deterministically guesses an~accepting run of $\Aa$ over $D[t]$ for some hypothetical tree $t$. When it reaches a~port of $D$, it verifies if the current state $q$ of $\Aa$ belongs to $h$. Once these automata are constructed, the question whether $h\approx h'$ boils down to checking if $\lang(\Cc_h)=\lang(\Cc_{h'})$, which can be done in \EXPTIME\@.
}

\new{
Now, that the equivalence relation $\approx$ is computed, it is enough to divide $(H,V)$ by it. Because of~\eqref{eq:def-of-approx}, there is a~bijection between the quotient $(H,V)/\approx$ and the syntactic algebra $(H_L,V_L)$. Moreover, since both $\alpha$ and $\alpha_L$ are homomorphisms of thin algebras, the thin algebra operations in $(H,V)$ must preserve $\approx$. This means that we can obtain the structure of a~thin algebra on $(H,V)/\approx$ that corresponds to the thin algebra structure of $(H_L, V_L)$. This concludes the proof of Proposition~\ref{pro:alg-finite-for-reg}.
}

\subsection{Quotients}%
\label{ssec:quotients}

Similarly as in the case of finite words, the syntactic algebra induces a~natural notion of a~quotient of a~language: given a~multicontext $D$ with $n$ holes and a~language $K$ of trees, by $D^{-1}(K)$ we denote the set of tuples $(t_1,\ldots,t_n)$ such that the valuation $\eta$ mapping the $i$th port of $D$ into $t_i$ satisfies
\[D[\eta]\in K.\]
Notice that if $D$ is a~context then in fact $D^{-1}(K)$ is a~set of $1$\=/tuples of trees, \new{which we identify with trees, i.e.~$D^{-1}(K)\subseteq \trees_A$.}

If $L$ is a~language recognised by a~morphism $\alpha_L$ then the fact whether $D[t]\in L$ depends only on $\alpha_L(D)\in V_L\sqcup\{\one_L\}$ and $\alpha_L(t)\in H_L$. Therefore, it makes sense to write $v^{-1}(L)$ for $v\in V_L\sqcup\{\one_L\}$. Also in that case $v^{-1}(L)$ is a~language recognised by the morphism $\alpha_L$. Directly from the definition we get that
\begin{equation}%
\label{eq:quot-comp}
{(vu)}^{-1}(L)=u^{-1}\big(v^{-1}(L)\big).
\end{equation}

\subsection{The game on types}%
\label{ssec:game-on-types}

Now we want to extend the definition of the game $\gameH$ to sets of contexts. Recall that contexts are defined as a~special case of partial trees, with an additional port label that appears in exactly one leaf. \new{Notice that if $p$ is a~finite partial context (i.e.~$p$ has exactly one leaf labelled $\hole$) then every context in the basic open set $N_p$ must have the same port as $p$ --- in such a~case we say that $p$ \emph{fixes the port}. Also, if $U$ is an~open set of contexts and $C\in U$ then there exists a~sufficiently big prefix $p\subset C$ such that $C\in N_p\subseteq U$ and $p$ fixes the port. Thus, without loss of generality we can assume that we consider only those basic open sets $N_p$ where $p$ does fix the port. This means that every two contexts $C,C'\in N_p$ must have the same position of the port.}

\new{
This yields the definition of a~game $\gameV(K_1,\ldots,K_n)$ for a~sequence $K_1,\ldots ,K_n$ of context languages (i.e.~subsets of $\Context\sqcup\{\hole\}$), which is played by Alternator and Constrainer. The game is played in $n$ rounds. Round $i=1$ is special: Alternator chooses a~context $C_1 \in K_1$. Let $u$ be the port of the context $C_1$. This port will stay fixed for the rest of the game; all contexts produced by Alternator
will have their port in the node $u$. Next, Constrainer chooses a~finite prefix $D_1$ of $C_1$, which has one of its leaves in the node $u$ (i.e.~fixes the port).
}

A~subsequent round $i \in \{2,\ldots,n\}$ is played as follows. Let $D_{i-1}$ be the finite partial tree over $A\sqcup\{\hole\}$ chosen by Constrainer in the previous round with a~leaf in the node $u$.
\begin{itemize}
\item Alternator provides a~context $C_i$, which extends $D_{i-1}$, belongs to $K_i$, and has its port in the node $u$. If there is no such context, the game is interrupted and Constrainer wins immediately.
\item Constrainer chooses a~finite prefix $D_i$ of $C_i$ and which has a~leaf in the node $u$.
\end{itemize}

\noindent
If Alternator manages to survive $n$ rounds then he wins. Recall that by the definition of the syntactic morphism, a~tree type $h \in H_L$ is actually equal to the set of trees $\alpha_L^{-1}(h)$, similarly for a~context type $v\in V_L$. Therefore, it makes sense to talk about the games $\gameH(h_1, \ldots, h_n)$, and $\gameV(v_1,\ldots,v_n)$ for sequences of types. Using these games we define the following sets of sequences of types:
\begin{defi}
We define two sets:
\begin{align*}
\gameH_L&=\{ (h_1, \ldots, h_n) \in {(H_L)}^\ast \mid \text{ Alternator wins $\gameH(h_1, \ldots, h_n)\}$},\\
\gameV_L&= \{ (v_1, \ldots, v_n) \in {(V_L)}^\ast \mid \text{ Alternator wins $\gameV(v_1, \ldots, v_n)\}$}.
\end{align*}
\end{defi}
A~comment on notation is in order here. The sets $\gameH_L$ and $\gameV_L$ contain words, over alphabets $H_L$ and $V_L$, respectively. Usually when dealing with words, one omits the brackets and commas and writes $abc$ instead of $(a, b, c)$. When the alphabet is $V_L$ this leads to ambiguity, since the expression $vwu$ can be interpreted as a~word with a~single letter obtained by multiplying the three context types $v$, $w$, and $u$, or a~three\=/letter word over the alphabet $V_L$. These two interpretations should not be confused, so we write $(v_1, \ldots , v_n)$ for $n$\=/letter words over the alphabet $V_L$. For the sake of uniformity, we also write $(h_1, \ldots, h_n)$ for $n$\=/letter words over the alphabet $H_L$, although there is no risk of ambiguity here.

It turns out that the sets of words $\gameH_L$ and $\gameV_L$ have specific structure. We say that a~word $u$ is a~\emph{subword} of $w$ if $u$ can be obtained from $w$ by removing some letters.

\begin{fact}%
\label{ft:removing-letters}
\new{Both sets $\gameH_L$ and $\gameV_L$ are closed under removing letters, i.e.~if $u$ is a~subword of $w$ and $w\in\gameH_L$ then also $u\in \gameH_L$ (similarly for $\gameV_L$).}
\end{fact}

\begin{proof}
\new{It is clear from the definitions of the games $\gameH$ and $\gameV$.}
\end{proof}

\new{The order induced by the subword relation (known also as the Higman's order) has the following finiteness property.}

\begin{lem}[Higman's Lemma, see~\cite{higman_lemma}]
The set of finite words $A^\ast$ over a~finite alphabet $A$ with the~subword ordering is a~well\=/quasi order: there is no infinite antichain nor an~infinite descending chain.
\end{lem}

\begin{fact}%
\label{ft:removing-regular}
If $L\subseteq A^\ast$ is closed under removing letters then $L$ is regular.
\end{fact}

\begin{proof}
Consider $L\subseteq A^\ast$ that is closed under removing letters. Then $L$ forms a~downward\=/closed set with respect to the subword relation. Thus, Higman's Lemma implies that there exists a~finite set of words $w_1,\ldots,w_N$ such that $w$ does not belong to $L$ if and only if one of $w_1,\ldots,w_N$ is a~subword of $w$. Such a~condition is a~regular condition.
\end{proof}

\begin{cor}
Both $\gameH_L$ and  $\gameV_L$ are regular languages of finite words.
\end{cor}

\begin{proof}
Both languages are closed under removing letters and therefore Fact~\ref{ft:removing-regular} applies.
\end{proof}

The above corollary is amusing, but useless \new{for our needs}, because it does not say how to compute automata for $\gameH_L$ and $\gameV_L$ as a~function of a~representation of the language $L$.

\begin{lem}%
\label{lem:compose-winh-with-multicontext}
The following properties hold \new{for each context $C$ and context environment~$E$}:
\begin{enumerate}
\item $(h_1, \ldots, h_n) \in \gameH_L$ implies $(C[h_1], \ldots, C[h_n])\in \gameH_L$.
\item\label{it:composes-E-and-v} $(v_1, \ldots, v_n) \in \gameV_L$ implies $(E[v_1], \ldots, E[v_n]) \in \gameH_L$.
\item\label{it:compose-v-and-w}
 $(v_1,\ldots,v_n),(w_1,\ldots,w_n)  \in \gameV_L$ implies $(v_1w_1,\ldots,v_{n}w_{n}) \in \gameV_L$.
\item\label{it:compose-v-and-h} $(v_1,\ldots,v_n)\in \gameV_L ,(h_1,\ldots,h_n)  \in \gameH_L$ implies $(v_1h_1,\ldots,v_{n}h_{n}) \in \gameH_L$.
\item\label{it:compose-v-infty} $(v_1,\ldots,v_n) \in \gameV_L$  implies $(v_1^\infty,\ldots,v_n^\infty) \in \gameH_L$.
\item $(h_1,\ldots,h_n)  \in \gameH_L$ implies $(c[\hole,h_1],\ldots,c[\hole,h_n]) \in \gameV_L$.
\item $(h_1,\ldots,h_n)  \in \gameH_L$ implies $(c[h_1,\hole],\ldots,c[h_n,\hole]) \in \gameV_L$.
\end{enumerate}
\end{lem}

\begin{proof}
All properties are proved by composing strategies, we prove the first one. All other properties are proved similarly. Assume that $(h_1,\ldots,h_n)  \in \gameH_L$, and consider some multicontext $C$ (possibly with infinitely many ports). For all $i\leq n$ let $L_i$ be the set of trees that are Alternator's first move in some winning strategy for $\gameH(h_i,\dots,h_n)$. Note that since $(h_1,\dots,h_n)\in \gameH_L$, $L_i$ is non\=/empty for all $i$.

\new{Lemma~\ref{lem:compose-winh-with-multicontext} follows directly from the following claim.}

\begin{claim}
For all $i \leq n$, for all trees $t$ obtained by plugging trees of $L_i$ in the ports of $C$, Alternator has a~winning strategy in $\gameH(C[h_i],\dots, C[h_n])$ such that the tree chosen in round $1$ is $t$.
\end{claim}

\begin{proof}
We proceed by induction on $i$. For $i=n$ this is obvious. Assume the result holds for $i$ and we prove it for $i-1$. Let $p$ be a~prefix of $t$, for all subtree $s$ plugged into a~port of $C$, $p$ yields some (possibly empty) prefix $p_s$ of $s$. Since $s \in L_{i-1}$, $p_s$ can be completed into a~tree $s'\in L_i$. It follows that $p$ can be completed into $t'$ obtained by plugging trees of $L_i$ in the ports of $C$. By induction hypothesis, Alternator has winning strategy in $\gameH(C[h_i],\dots, C[h_n])$ such the tree chosen in round $1$ is $t'$. Finally we conclude that Alternator has winning strategy
in $\gameH(C[h_{i-1}],\dots,C[h_n])$ such the tree chosen in round~$1$ is~$t$.
\end{proof}
This completes the proof of Lemma~\ref{lem:compose-winh-with-multicontext}.
\end{proof}

\begin{defi}%
\label{def:alternation}
We define the \emph{alternation} of a~finite word to \new{be the length of the word obtained by iteratively eliminating letters that are identical to their predecessors}. We say that a~set of words has \emph{unbounded alternation} if it contains words with arbitrarily large alternation. In the other case we say that a~set has \emph{bounded alternation}. A~word is \emph{alternating} if every two consecutive letters are distinct.
\end{defi}

For example the alternation of \new{$abccbbb$ is $4$ and $abcb$ is the alternating word witnessing that.} The notion of \emph{alternation} gives us another characterization of the class $\boolcomb{1}$.

\begin{lem}%
\label{lem:char-by-alternation}
For a~regular language $L$ of infinite trees, the following conditions are equivalent:
\begin{enumerate}
\item Alternator wins the game $\gameHs(L,n)$ for \new{infinitely many} $n$.
\item The set $\gameH_L$ has unbounded alternation.
\end{enumerate}
\end{lem}

\noindent
Clearly, if $n\leq n'$ and Constrainer wins the game $\gameHs(L,n)$ then he also wins the game $\gameHs(L,n')$. This means that in fact there are only two possibilities: either Alternator wins $\gameHs(L,n)$ for all $n$; or Constrainer wins $\gameHs(L,n)$ for all except finitely many $n$. Thus, the conditions of Lemma~\ref{lem:char-by-alternation} are in fact equivalent to saying that Alternator wins the game $\gameHs(L,n)$ for all $n$.

\begin{proof}
We have to prove both implications.
\begin{enumerate}[leftmargin=13mm]
\item[$1 \Rightarrow 2$.]
We show that for $n\in\w$ if Alternator wins the game $\gameHs(L,n)$, then $\gameH_L$ contains an~alternating word of length $n$. Suppose that Alternator wins $\gameHs(L,n)$. Both $L$ and $L\complemento$ can be partitioned into tree types, \new{see Remark~\ref{rem:L-equi}}. By Lemma~\ref{lem:refinement_lemma}, Alternator wins $\gameH(h_1,\ldots,h_n)$ for some sequence of types, such that $h_i$ is included in $L$ or its complement, depending on the parity of $i$. In particular, the consecutive types are different.

\item[$2 \Rightarrow 1$.] Suppose that $\gameH_L$ has unbounded alternation. \new{By Fact~\ref{ft:removing-letters} we know} that $\gameH_L$ is closed under removing letters. We will now argue that there must be some $g,h \in H_L$ with $g\neq h$ such that $\gameH_L$ contains all the words
\[(g,h), (g,h,g,h), (g,h,g,h,g,h), \ldots\]
\new{Assume contrarily and use finiteness of $H_L$: there must exist a~global bound $B$ on the number of times any pair of distinct types $g\neq h$ can appear in a~word in $\gameH_L$ in the alternating way as above. Consider an~alternating word of length $|H_L|^2\cdot (B+1)+1$ in~$\gameH_L$. This word contains $|H_L|^2\cdot (B+1)$ pairs of consecutive distinct letters and therefore some pair of those must appear there at least $B+1$ times --- a~contradiction.}

Since $g$ and $h$ are different elements of the syntactic algebra, it follows that there must be some multicontext $C$ such that the tree type $C[g]$ is contained in $L$, while the tree type $C[h]$ is disjoint with $L$. By the first item of Lemma~\ref{lem:compose-winh-with-multicontext}, we can conclude that $\gameH_L$ contains  all the words
\[\big(C[g],C[h]\big), \big(C[g],C[h],C[g],C[h]\big), \big(C[g],C[h],C[g],C[h],C[g],C[h]\big), \ldots\]
It follows that Alternator can alternate arbitrarily long between the language $L$ and its complement.
\qedhere
\end{enumerate}
\end{proof}

\begin{lem}%
\label{lem:VLunbHLunb}
If $\gameV_L$  has unbounded alternation then so does $\gameH_L$.
\end{lem}

\begin{proof}
Assume that $\gameV_L$ has unbounded alternation. Take $n>0$, we will find a~sequence in~$\gameH_L$ of alternation at least $n$. Let $N\eqdef 2\cdot n\cdot |V_L|^2$ and let $(v_1, \ldots, v_{N})$ be an alternating sequence in $\gameV_L$ of length $N$ --- such a~sequence exists by the assumption and Fact~\ref{ft:removing-letters}.
By Pigeon\=/hole Principle, there exist two types $v\neq v'\in V_L$ such that the word ${(v,v')}^n$ can be obtained from $(v_1, \ldots, v_{N})$ by removing letters. \new{By the definition of $V_L$ (see Definition~\ref{def:cont-equiv}) we know that there exists a~context environment $E$ such that $E[v]\neq E[v']$. Thus, by Item~\ref{it:composes-E-and-v} of Lemma~\ref{lem:compose-winh-with-multicontext} we know that the word ${\big(E[v],E[v']\big)}^n$ is a~member of $\gameH_L$. Since this word is alternating and has length $2n$, the claim holds.}
\end{proof}

\section{Effective characterisation of \texorpdfstring{$\boolcomb{1}$}{BC(Sigma01)}}
\label{sec:main-thm}

In this section we state the crucial result of the paper, providing an effective characterisation of the class of regular tree languages in $\boolcomb{1}$.

\begin{thm}%
\label{thm:mainforBC1}
For a~regular language $L$ of infinite trees, the following conditions are equivalent.
\begin{enumerate}
\item $L$ is a~Boolean combination of open sets, i.e. $L \in \boolcomb{1}$.
\item Constrainer wins the game $\gameHs(L,n)$ for all but finitely many $n$.
\item The set $\gameH_L$ has bounded alternation.
\item The following identities are satisfied in the algebra $(H_L, V_L)$:
\begin{align}
\label{eq:cond-bool}
u^\monoid u w^\monoid = u^\monoid v w^\monoid &= u^\monoid w w^\monoid  &\text{ if $(u,v,w) \in \gameV_L$ or $(w,v,u) \in \gameV_L$}\\
\label{eq:cond-limit}
{(u_2 w_2^\monoid v)}^\monoid u_1 w_1^\infty &= {(u_2 w_2^\monoid v)}^\infty &\text{ if $v\in V_L$ and $(u_1,u_2), (w_1,w_2) \in \gameV_L$}
\end{align}
\end{enumerate}
\end{thm}

\noindent
We have already proved the implications $1 \Leftrightarrow 2 \Leftrightarrow 3$ respectively in Corollary~\ref{cor:char-by-winning-the-game} and Lemma~\ref{lem:char-by-alternation}.
It remains to prove $3 \Leftrightarrow 4$.
The direction $3 \Rightarrow 4$ is not hard and we prove it in the next section, whereas the direction $4 \Rightarrow 3$ forms the technical core of the paper.

\begin{cor}
The problem whether a~regular language $L$ belongs to $\boolcomb{1}$ is \EXPTIME\=/complete, when the language $L$ is given as a~parity non\=/deterministic automaton $\Aa$ recognising~$L$.
\end{cor}

\begin{proof}
Given a~representation of $L$, one can compute the algebra $(H_L,V_L)$ in \EXPTIME, see Fact~\ref{pro:alg-finite-for-reg}. Then, verifying the equations from condition~4 of Theorem~\ref{thm:mainforBC1} can be done in time polynomial in the size of the algebra (which is exponential in the number of states of the given automaton for $L$).

Hardness follows immediately from \EXPTIME hardness of the universality problem for non\=/deterministic automata over finite trees~\cite{seidl_hardness}, see~\cite[Theorem~4.1, page~8]{walukiewicz_low_levels} for a~generic reduction of problems about complexity of infinite tree languages to universality of finite tree languages.
\end{proof}

\subsection{Overview of the proof}

\new{
The remaining part of the proof of Theorem~\ref{thm:mainforBC1} consists of two implications: $3\Rightarrow 4$ (Subsection~\ref{ssec:three-to-four}) and $4\Rightarrow 3$ (Subsection~\ref{ssec:four-to-three}).
}

\new{
The proof of the implication $3\Rightarrow 4$ is rather straightforward: we assume that either~\eqref{eq:cond-bool} or~\eqref{eq:cond-limit} is violated and prove that the set $\gameH_L$ has unbounded alternation. This is achieved by providing concrete strategies for Alternator in the game $\gameH_L$.
}

\new{
The proof of the implication $4\Rightarrow 3$ is more demanding. First, we introduce a~concept of \emph{strategy trees} that represent specific strategies of Alternator in $\gameH$ (Subsections~\ref{sssec:strat-trees} to~\ref{sssec:locally-opt}). We study properties of the set $\Sigma_L$ of those strategies and notice two distinct cases of their behaviour: Case~(C1) leading to a~violation of~\eqref{eq:cond-bool}; and Case~(C2) leading to a~violation of~\eqref{eq:cond-limit}. See Subsection~\ref{sssec:case-distinction} for the case distinction.
}

\new{
The first argument, provided in Subsection~\ref{ssec:limit-bounded}, utilises the assumption of Case~(C1) to construct \emph{strategy matrices} --- finite combinatorial objects witnessing the assumptions of Case~(C1). By applying Erd\"os--Szekeres Theorem and Ramsey Theorem for hypergraphs, we gradually simplify the structure of these matrices. Ultimately, we construct a~strategy matrix that literally encodes a~violation of~\eqref{eq:cond-bool} (see Lemma~\ref{lem:very-regular-matrix}). 
}

\new{
The second argument, given in Subsection~\ref{ssec:limit-unbounded}, works under the assumption of Case~(C2). This proof involves another object for representing certain strategies in $G_L$: a~\emph{strategy graph} $G_L$. The edges of that directed graph encode special types of strategies of Alternator in the game $\gameV$, see Lemma~\ref{lem:def-of-edges}. A~relatively easy observation (Lemma~\ref{lem:rec-to-violation}) shows that if the graph is \emph{recursive} (i.e.~contains certain loops) then it witnesses a~violation of~\eqref{eq:cond-limit}. Subsection~\ref{sssec:cons-path} constructs inductively such a~loop in $G_L$ based on the assumption of Case~(C2).
}

\subsection{The implication \texorpdfstring{$3 \Rightarrow 4$}{3 => 4}}%
\label{ssec:three-to-four}

\new{
In this section we prove the implication $3 \Rightarrow 4$ of Theorem~\ref{thm:mainforBC1} in the contra positive, as stated below.}

\begin{prop}%
\label{pro:three-to-four}
\new{
If one of the identities~\eqref{eq:cond-bool} and~\eqref{eq:cond-limit} of Theorem~\ref{thm:mainforBC1} is violated then the set $\gameH_L$ has unbounded alternation.}
\end{prop}

\subsubsection{The case when~\eqref{eq:cond-bool} is violated.}

The assumption that~\eqref{eq:cond-bool} is violated says that there are $u,v,w \in V_L$ such that
\[(u,v,w) \in \gameV_L \qquad \text{or}\qquad (w,v,u) \in \gameV_L,\]
but the three context types $u^\monoid u w^\monoid$,  $u^\monoid v w^\monoid$, and $u^\monoid w w^\monoid$ are not all equal. If the three context types are not equal then the second one must be different from either the first one or the third one. We only do the proof for the case when $(u,v,w) \in \gameV_L$ and when $ u^\monoid u w^\monoid \neq
u^\monoid v w^\monoid$; the other cases are entirely dual. For $n \geq 0$ and $i \in \{1,\ldots,n\}$, define
\[\vec w_{(i,n)} \eqdef \big(\overbrace{u,u,\ldots,u}^{2(n{-}i){+}1},\ v,\ \overbrace{w,w,\ldots,w}^{2(i{-}1)}\big) \quad \in {(V_L)}^{2n}.\]
This word is obtained from $(u,v,w)$ by duplicating some letters, and therefore it  belongs to $\gameV_L$. For a~given $n$, consider the words
\[\vec w_{(1,n)},  \ldots, \vec w_{(n,n)} \in \gameV_L.\]
These are $n$  words of length $2n$. Let us multiply all these words coordinate\=/wise, yielding a~word $\vec w_{n}$, also of length $2n$, which
is depicted in the following picture:
\begin{center}
\begin{tikzpicture}

\fill[gray!50] (-6+5*0.7-0.2, 0.35-0.5) rectangle ++(0.4, -4*0.5-0.7);

\fill[gray!50] (-6+6*0.7-0.2, 0.35-0.5) rectangle ++(0.4, -4*0.5-0.7);

\foreach \i in {1,...,5} {
	\node[toL] at (-5.7, -\i*0.5) {$\vec{w}_{(\i,n)}=$};
	\tikzEvalInt{\maxX}{2*(5-\i)+1}
	\foreach \x in {1,...,\maxX} {
		\node[toC] at (-6 + \x*0.7, -\i*0.5) {$u$};
	}
	\tikzEvalInt{\maxX}{\maxX+1}
	\node[toC] at (-6 + \maxX*0.7, -\i*0.5) {$v$};
	\tikzEvalInt{\maxX}{\maxX+1}
	\ifthenelse{\maxX < 11}{
	\foreach \x in {\maxX,...,10} {
		\node[toC] at (-6 + \x*0.7, -\i*0.5) {$w$};
	}}{}
}

\node[toL] at (-6+5*0.7, -0.5 + 0.9) {letter $2i{-}1$};
\node[toR] at (-6+6*0.7, -0.5 + 0.9) {letter $2i$};

\draw[edge,<-] (-6+5*0.7, -0.5 + 0.4) -- ++(-0.5, +0.3);
\draw[edge,<-] (-6+6*0.7, -0.5 + 0.4) -- ++(+0.5, +0.3);

\node[toL] at (-6+5*0.7, -5*0.5 - 1.0) {$u^{n-i+1}w^{i-1}$};
\node[toR] at (-6+6*0.7, -5*0.5 - 1.0) {$u^{n-i}vw^{i-1}$};

\draw[edge,<-] (-6+5*0.7, -5*0.5 - 0.4) -- ++(-0.5, -0.3);
\draw[edge,<-] (-6+6*0.7, -5*0.5 - 0.4) -- ++(+0.5, -0.3);
\end{tikzpicture}
\end{center}

As $\vec{w}_{n}$ is obtained by a~coordinate\=/wise multiplication of the rows of the above matrix, and all these rows belong to $\gameV_L$, Lemma~\ref{lem:compose-winh-with-multicontext} implies that also $\vec{w}_{n}\in\gameV_L$.

Recall that $\monoid$ is a~number dependant on $V_L$ such that for every $z\in V_L$ we know that $z^\monoid$ is an idempotent. Choose some $k$, and take  $n = k \cdot \monoid+1$, and $i \in \{\monoid{+}1,2\monoid{+}1, \ldots,(k-1) \cdot  \monoid{+}1\}$. Consider the letters $2i-1$ and $2i$ in the word $\vec w_{n}$, which are
\[u^{n-i+1}w^{i-1} = u^\monoid uw^{\monoid}  \qquad u^{n-i} v w^{i-1} = u^\monoid v w^\monoid.\]

By the assumption, these letters are different, and therefore the word $\vec w$ has alternation at least $k$. Because $k$ was chosen
arbitrarily, it follows that $\gameV_L$ has unbounded alternation. Lemma~\ref{lem:VLunbHLunb} implies that in that case also $\gameH_L$ has unbounded alternation.

\subsubsection{The case when~\eqref{eq:cond-limit} is violated.}

The assumption that~\eqref{eq:cond-limit} is violated says that $\gameV_L$ contains pairs $(u_1,u_2)$ and $(w_1,w_2)$ such that for some $v \in V_L$,
\begin{align*}
e^\infty \neq e u_1 w_1^\infty \qquad \mbox{ for }e \eqdef {(u_2w_2^\monoid v)}^\monoid.
\end{align*}

Let $h_1=e^\infty$ and $h_2=e u_1 w_1^\infty$, by the above assumption we know that $h_1\neq h_2$. It turns out that a~violation of Equation~\eqref{eq:cond-limit} has even stronger consequences than a~violation of Equation~\eqref{eq:cond-bool}, as expressed by the following claim. It speaks about the infinite variant of the game of types, which is defined analogously to the finite one, see Subsection~\ref{ssec:game-on-types}.

\begin{claim}%
\label{cl:alt-wins-infinite}
If Equation~\eqref{eq:cond-limit} is violated then Alternator wins $\gameHinfty(h_1,h_2,h_1,\ldots)$.
\end{claim}

The above claim implies in particular that for every $n$ the sequence
\[{\big(h_1, h_2\big)}^n\]
belongs to $\gameH_L$, and thus $\gameH_L$ has unbounded alternation.

\begin{proof}
Let $C_{u_1}$, \new{$C_{w_1}$} be contexts of types $u_1$, \new{$w_1$} that witness that $(u_1,u_2), (w_1,w_2)\in \gameV_L$ i.e.~they are contexts played by Alternator in the first rounds of the respective games. Let $C_{u_2}$, $C_{w_2}$, and $C_v$ be any three contexts of types $u_2$, $w_2$, and $v$ respectively. Let $C_e\eqdef {\big(C_{u_2} C_{w_2}^\monoid C_v\big)}^\monoid$. The type of the context $C_e$ is $e$.

\new{
Our aim is to provide an explicit strategy of Alternator in the game $\gameHinfty(h_1,h_2,h_1,\ldots)$. This will be done inductively, considering the consecutive rounds of the game. The outcome will be a~set of trees ${(t_i)}_{i<\w}$ that are played in the considered play of the game, with $\alpha_L(t_i)=h_1$ for odd $i$ and $\alpha_L(t_i)=h_2$ for even $i$.
}

The strategy of Alternator starts with the tree $t_1={\big(C_e\big)}^\infty$. We will demonstrate how it works for the first three rounds of the game, the rest is analogous.

Consider a~prefix $p_1$ of the tree $t_1$ that is fixed by Constrainer. It must be the case that $p_1$ is a~prefix of ${\big(C_e\big)}^{n_1}$ for some $n_1$. Thus, Alternator can now provide the tree
\[t_2={\big(C_e\big)}^{n_1}\cdot C_{u_1}\cdot C_{w_1}^\infty.\]
Now Constrainer chooses a~prefix $p_2$ of the above tree, in fact $p_2$ is a~prefix of ${\big(C_e\big)}^{n_1}\cdot C_{u_1}\cdot C_{w_1}^{\monoid\cdot n_2}$ for some $n_2$. Thus, by the assumptions on $C_{u_1}$ and $C_{u_2}$, Alternator is able to provide a~tree of the form
\[t_2 = {\big(C_e\big)}^{n_1}\cdot \big( D_{u_2}\cdot D_{w_2}^{\monoid\cdot n_2}\cdot C_v\big) \cdot {\big(C_e\big)}^{\infty},\]
where the contexts $D_{u_2}$ and $D_{w_2}$ depend on the prefix $p_2$ and have types $u_2$ and $w_2$ respectively. Now we proceed inductively as before, because any prefix $p_3$ of $t_2$ must be a~prefix of ${\big(C_e\big)}^{n_1}\cdot \big( D_{u_2}\cdot D_{w_2}^{\monoid\cdot n_2}\cdot C_v\big) \cdot {\big(C_e\big)}^{n_3}$.

Notice that by the choice of $e$, the type of the context $D_{u_2}\cdot D_{w_2}^{\monoid\cdot n_2}\cdot C_v$ is $e$. Therefore, $\alpha_L(t_i)=e^\infty=h_1$ for odd $i$ and $\alpha_L(t_i)=e u_1 w_1^\infty=h_2$ for even~$i$.
\end{proof}

\new{
Consider a~pair of trees $t_1$ and $t_2$ of types $h_1$ and $h_2$ respectively. Since $t_1\not\Teq_L t_2$, there exists a~multicontext $C$ such that\footnote{We swap $h_1$ and $h_2$ if needed.}
\begin{equation}
C[t_1]\in L\ \land\ C[t_2]\notin L,
\label{eq:distinction}
\end{equation}
see Definition~\ref{def:myhill-nerode}. Now, by Item~1 of Lemma~\ref{lem:preserve} we know that the types $\alpha_L\big(C[t_1]\big)$, $\alpha_L(C[t_2]\big)$ (and therefore the conditions in~\eqref{eq:distinction}) do not depend on the actual choice of the trees $t_1\in h_1$ and $t_2\in h_2$. Thus, by using the above strategy for Alternator under the multicontext~$C$, we obtain the following corollary.
}

\begin{cor}%
\label{cor:eq-limit-to-game}
If Equation~\eqref{eq:cond-limit} is violated then Alternator wins~$\gameHinfty(L)$.
\end{cor}

\subsection{The implication \texorpdfstring{$4 \Rightarrow 3$}{4=>3} --- case distinction}%
\label{ssec:four-to-three}

We now move to the implication $4 \Rightarrow 3$ of Theorem~\ref{thm:mainforBC1}, which is the most involved part of the article. To prove it, we need to introduce a~crucial concept witnessing a~large alternation of the set $\gameH_L$. The objects witnessing that will be called \emph{strategy trees} and \emph{locally optimal strategy trees}. Using these objects we will split the proof of that implication into two separate cases, \new{see Subsection~\ref{sssec:case-distinction}. We deal with these cases in Subsections~\ref{ssec:limit-bounded} and~\ref{ssec:limit-unbounded}.}.

\subsubsection{Strategy trees}%
\label{sssec:strat-trees}

First, a~\emph{\newdef} is a~tree such that the nodes are labelled with tree types, i.e.~it is a~tree over the alphabet $H_L$.

\begin{defi}
If $t$ is a~tree over the alphabet $A$, the \emph{\newdef induced by} $t$ is the~\newdef $\sigma$ where the label of a~node $u$ is the tree type of the subtree $t.u$, i.e.~$\sigma(u)=\alpha_L(t.u)$. In particular $\sigma(\epsilon)=\alpha_L(t)$.
\end{defi}

\begin{defi}
Let $\sigma$ be a~\newdef and let $t$ be a~tree over $A$. We say that $\sigma$ is \emph{locally consistent} with $t$ if for every node $u \in {\{\dL,\dR\}}^\ast$, whose label in $t$ is $a$, we have that $\sigma(u)$ is the type obtained by applying the letter $a$ to the pair of types $\sigma(u\dL)$ and $\sigma(u\dR)$, that is
\begin{equation}
\sigma(u)=a\big(\sigma(u\dL),\hole\big)\cdot \sigma(u\dR).
\label{eq:local-consistency}
\end{equation}
In other words, if we take a~tree $t_1$ inside $\sigma(u\dL)$ and a~tree $t_2$ inside $\sigma(u\dR)$ and we plug these two trees as the left and the right children of $a$, then the type of the result is $\sigma(u)$.
\end{defi}

\begin{rem}
The \newdef induced by $t$ is locally consistent with $t$.
\end{rem}

\begin{exa}
Consider the language
\[ L = \{ t \in \trees_A \mid \text{$t$ contains at least one $b$}\}\]
over the alphabet $A = \{a,b\}$. $H_L$ contains two tree types that (treated as sets of trees) are $L$ and $L\complemento$. Now consider a~\newdef $\sigma$ where every node is labelled with $L$. It is easy to check that $\sigma$ is locally consistent with every tree $t \in \trees_A$, even with the tree that contains only letters $a$.
\end{exa}

\begin{fact}%
\label{ft:consistent-closed}
The set of pairs
\[\big\{(t,\sigma)\in\trees_A\times\trees_{H_L}\mid \text{$\sigma$ is locally consistent with $t$}\big\}\]
is closed in the product topology.
\end{fact}

\begin{proof}
As Condition~\eqref{eq:local-consistency} speaks about finitely many values of $\sigma$ and $t$, it corresponds to a~clopen set of pairs. A \newdef $\sigma$ is locally consistent with a~tree $t$ if they obey the local consistency conditions from~\eqref{eq:local-consistency} in all the vertices. Thus, the above defined set of pairs is an intersection of a~family of clopen sets.
\end{proof}

\begin{lem}%
\label{lem:limitofstrtrees}
Let ${(t_n)}_{n\in \N}$ be a~sequence of trees that converges to $t_\ast$ and let ${(\sigma_n)}_{n\in \N}$ be a~sequence of \newdef{s} that converges to $\sigma_\ast$. If $\sigma_n$ is locally consistent with $t_n$ for every $n$ then $\sigma_\ast$ is locally consistent with $t_\ast$.
\end{lem}
\begin{proof}
Follows directly from Fact~\ref{ft:consistent-closed}.
\end{proof}

Now we are ready to define strategy trees, a~key concept of this paper.

\begin{defi}\label{def_strategytree}
A~\emph{strategy tree} is a~tuple $\sigma = (t, \sigma_1, \ldots, \sigma_n)$
where:
\begin{enumerate}
\item $t$ is a~tree over $A$, called the \emph{support} of $\sigma$.
\item $\sigma_1$ is the \newdef induced by $t$.
\item The \newdef{s} $\sigma_2, \ldots, \sigma_n$ are locally consistent with $t$.
\item For each node $u$ of $t$, the sequence $\big(\sigma_2(u), \ldots, \sigma_n(u)\big)$ belongs to $\gameH_L$.
\end{enumerate}
\end{defi}

\noindent
Notice that the fourth condition of Definition~\ref{def_strategytree} does not mention the \newdef $\sigma_1$. By the definition, $\sigma$ can be interpreted as a~single tree over the alphabet $A \times H_L^n$.

Intuitively speaking, a~strategy tree represents a~special kind of strategy for Alternator. In the first round, Alternator plays the support $t$ of $\sigma$. However,
Alternator also declares all the types that will appear in the nodes of $t$ as the game progresses. More specifically, he declares that for every
node $u$ of $t$ and round $k \in \{2,\ldots,n\}$, he has a
strategy so that for the tree played in the round $k$, the subtree in the node $u$ has type $\sigma_k(u)$.

\paragraph*{\bf Alternations.}

The number $n$ is called the \emph{duration} of a~strategy tree $\sigma=(t,\sigma_1,\ldots,\sigma_n)$. We define the \emph{root sequence} of a~strategy tree to be the sequence of root labels of $\sigma_1,\ldots,\sigma_n$. If the duration is $n$, the root sequence is in $H_L^n$. We define the \emph{root alternation} of a~strategy tree $\sigma$ to be the alternation of its root sequence. We define the \emph{limit alternation} of $\sigma$ to be the maximal number $\ell$ such that infinitely many subtrees of $\sigma$ have root alternation at least $\ell$. This means that if the limit alternation of $\sigma$ is $\ell$ then there exist infinitely many nodes $u$ such that their root sequence (i.e.~the root sequence of the strategy tree obtained by truncating $\sigma$ in the node $u$) has alternation at least $\ell$.

\paragraph*{\bf Context zones.}

A~\emph{context zone} is a~set $X$ of nodes for which there exists a~node $u$, called the \emph{root} of $X$, and a~node $y$, called the \emph{port} of $X$, such that $u\prec y$ and $X$ contains the nodes that are in the subtree of $u$, but not in the subtree of $y$:
\[X=\big\{x\in{\{\dL,\dR\}}^\ast\mid u\preceq x \land y\not\preceq x\big\}.\]
We say that context zones $X_1,\ldots,X_n$ are \emph{consecutive} if for each $i \in \{1,\ldots,n-1\}$, the port of $X_i$ is the root of $X_{i+1}$. The union of consecutive context zones is a~context zone.

Consider a~tree $t$, a~\newdef $\sigma$, and a~context zone $X$. $X$ can be seen as a~context inside $t$ taking into account the types of the~subtrees of $t$ as declared in $\sigma$. This is achieved by defining a~value $\val(t,\sigma,X)\in V_L$. The definition of $\val(t,\sigma,X)$ is inductive on the number of nodes $v$ in $X$ such that $v \preceq y$, where $y$ is the port of $X$. If there is only one such node then the port $y$ of $X$ is a~child of its root $u$. Let $a$ be the label of
$u$ in $t$ and $v\in X$ be the other child of $u$. We set $\val(t,\sigma,X)$
as $a(\hole,\sigma(v))$ if $v$ is the right child and $a(\sigma(v)
,\hole)$ if $v$ is the left child. Otherwise, let $u$ be the root of
$X$, $y$ its port and $z$ the child of $u$ such that $z \preceq y$. Let $X_1$ be the context zone with root $u$ and port $z$ and $X_2$
the context zone of root $z$ and port $y$. We define
\[\val(t,\sigma,X)\eqdef\val(t,\sigma,X_1) \cdot \val(t,\sigma,X_2).\]
Figure~\ref{fig:val-example} illustrates an example of a~tree and a~context zone with
\[\val(t,\sigma,X)=b(\hole,h_1)\cdot c(h_2,\hole)\in V_L.\]
In other words, $\val(t,\sigma,X)$ is the value of the context $C\eqdef t\restr X$ when we assume that the subtrees of $t$ aside of the branch leading to the port of $C$ have types as declared by $\sigma$.

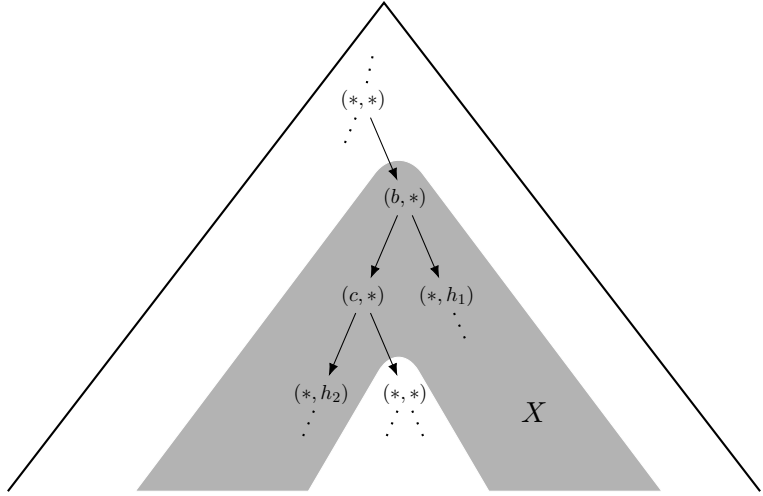
\begin{figure}
\begin{center}
\begin{tikzpicture}
\newcommand{\vVs}{1.3}

\tikzEvalFloat{\vTs}{\vVs*1.8}

\newcommand{\hTp}{0.5*2.0/1.8}

\newcommand{\hHt}{5}

\newcommand{\hHu}{3.5}
\tikzEvalFloat{\vMu}{(\hHt-\hHu)*\vVs}

\newcommand{\hHl}{1.5}
\tikzEvalFloat{\vMl}{(\hHt-\hHl)*\vVs}

\draw[thick] (-1*\hHt, -1*\hHt*\vVs) -- (0,0) -- (+1*\hHt, -1*\hHt*\vVs);

\coordinate (tu) at (+\hTp*0.35,0.1);
\coordinate (hu) at (+\hTp*0.35,0);

\coordinate (lu) at ($(-1*\hHu, -1*\hHu*\vVs-\vMu) + (hu)$);
\coordinate (uu) at ($(0,-\vMu) + (tu)$);
\coordinate (ru) at ($(+1*\hHu, -1*\hHu*\vVs-\vMu) + (hu)$);

\coordinate (ll) at ($(-1*\hHl*0.8, -1*\hHl*\vVs-\vMl) + (hu)$);
\coordinate (ul) at ($(0,-\vMl) + (tu)$);
\coordinate (rl) at ($(+1*\hHl*0.8, -1*\hHl*\vVs-\vMl) + (hu)$);

\path[fill=gray!60]
	(ll)[rounded corners=15pt, draw] --
	(ul)[rounded corners=0pt] --
	(rl)[draw=white!0] --
	(ru)[rounded corners=15pt, draw] --
	(uu)[rounded corners=0pt] --
	(lu)[draw=white!0] --
	cycle;

\tikzstyle{trlabb} = [scale=0.7]

\node[trlabb] (n1) at ($(0,-\vMu)+(-\hTp*0.5,+\hTp*\vTs*0.5)$) {$(\ast,\ast)$};
\node[trlabb] (n2) at ($(n1)+(+\hTp,-\hTp*\vTs)$) {$(b,\ast)$};
\node[trlabb] (n3r) at ($(n2)+(+\hTp,-\hTp*\vTs)$) {$(\ast,h_1)$};
\node[trlabb] (n3l) at ($(n2)+(-\hTp,-\hTp*\vTs)$) {$(c,\ast)$};
\node[trlabb] (n4r) at ($(n3l)+(+\hTp,-\hTp*\vTs)$) {$(\ast,\ast)$};
\node[trlabb] (n4l) at ($(n3l)+(-\hTp,-\hTp*\vTs)$) {$(\ast,h_2)$};

\path[tredge] (n1) -- (n2);
\path[tredge] (n2) -- (n3l);
\path[tredge] (n2) -- (n3r);
\path[tredge] (n3l) -- (n4l);
\path[tredge] (n3l) -- (n4r);

\coordinate (d0) at ($(0,-\vMu*0.5)+(-\hTp*0.5*0.5,+\hTp*\vTs*0.5*0.5)$);
\coordinate (d1) at ($(n1)+(-\hTp*0.5,-\hTp*\vTs*0.5)$);
\coordinate (d2r) at ($(n3r)+(+\hTp*0.5,-\hTp*\vTs*0.5)$);
\coordinate (d3l) at ($(n4l)+(-\hTp*0.5,-\hTp*\vTs*0.5)$);
\coordinate (d4r) at ($(n4r)+(+\hTp*0.5,-\hTp*\vTs*0.5)$);
\coordinate (d4l) at ($(n4r)+(-\hTp*0.5,-\hTp*\vTs*0.5)$);

\path[trdots] (n1) -- (d0);
\path[trdots] (n1) -- (d1);
\path[trdots] (n3r) -- (d2r);
\path[trdots] (n4l) -- (d3l);
\path[trdots] (n4r) -- (d4r);
\path[trdots] (n4r) -- (d4l);

\node[toC] at (2, -5.5) {$X$};
\end{tikzpicture}
\end{center}
\caption{An illustration to the value $\val(t,\sigma,X)$. The grey area is the set $X$, a~label $(a,h)$ of a~node $u$ represents the values of $t(u)=a$ and $\sigma(u)=h$ respectively. The symbol $\ast$ indicates values that are irrelevant for the value of the set $X$.}%
\label{fig:val-example}
\end{figure}

Now assume that $\sigma= (t, \sigma_1,\ldots,\sigma_n)$ is a~strategy tree. Let $X$
be a~context zone of $\sigma$ (we can see $\sigma$ as a~single tree). For a~round $i \in \{1,\ldots,n\}$, we define the type
\[\val(\sigma,X,i) \eqdef \val(t,\sigma_i,X)\in V_L.\]

\begin{fact}%
\label{fact:zone-strat}
Let $\sigma=(t,\sigma_1,\dots,\sigma_n)$ be a~strategy tree and $X$ be
a~context zone. Then
\[\big(\val(\sigma,X,1),\dots,\val(\sigma,X,n)\big) \in \gameV_L.\]
\end{fact}

\begin{proof}
This is by construction and Items~3, 6 and 7 of Lemma~\ref{lem:compose-winh-with-multicontext}.
\end{proof}

\subsubsection{Existence of strategy trees}

We will now prove that whenever Alternator has a~strategy in $\gameH$ then this fact is witnessed by a~strategy tree. This fact is expressed by the following two lemmas. For inductive reasons these lemmas are parametrised by a~finite multicontext $C$. The games $\gameH_C(h_1,\ldots,h_n)$ and $\gameV_C(h_1,\ldots,h_n)$ are defined as the games $\gameH_U(h_1,\ldots,h_n)$ and $\gameV_U(h_1,\ldots,h_n)$ respectively, where $U$ is the open set of all trees (resp.\ contexts) that can be obtained from $C$, i.e.~$C[\ast]$.

\begin{lem}%
\label{lem:compose-strategies}
Let $C$ be a~finite multicontext for strategy trees. Consider a~sequence of valuations $\eta_1,\ldots,\eta_n$ that map ports of $C$ to $H_L$. If
\[ (\eta_1(u),\ldots,\eta_n(u)) \in \gameH_L\]
for every port $u$, then Alternator wins the game
\[\gameH_C(C[\eta_1],\ldots,C[\eta_n]).\]
\end{lem}

\begin{proof}
For every port $u$ of $C$, Alternator has a~winning strategy for the game
\[\gameH(\eta_1(u),\ldots, \eta_n(u)).\]
We describe a~winning strategy for Alternator in $\gameH_C(C[\eta_1],\ldots,C[\eta_n])$. Alternator starts with a~tree obtained by plugging the initial tree from his winning strategy in the game $\gameH(\eta_1(u),\ldots, \eta_n(u))$ in every port $u$ of $C$. Then every prefix of this tree yields a~prefix for each subtree plugged into a~port of $C$ and Alternator can then use his strategies for the game $\gameH(\eta_1(u),\ldots, \eta_n(u))$ to answer.
\end{proof}

\begin{lem}
For a~sequence $(h_1, \ldots, h_n) \in H_L^\ast$ and a~finite multicontext $C$ the
following conditions are equivalent:
\begin{enumerate}
\item Alternator wins $\gameH_C(h_1, \ldots, h_n)$.
\item There is a~strategy tree  whose support extends $C$, with root sequence  $(h_1, \ldots, h_n)$.
\end{enumerate}
\end{lem}
\begin{proof}
We prove the lemma by induction on $n$. The base case when $n=0$ or $n=1$, is trivial. We do the induction
step. Let us begin with the easier implication from 2 to 1. Suppose that
\begin{align*}
	\sigma \equiv (t,\sigma_1,  \ldots, \sigma_n)
\end{align*}
is a~strategy tree as in item 2. Alternator's strategy is as follows. In the first round, he plays the tree $t$, which has type $h_1$. Suppose Constrainer chooses a~prefix $p$ of
$t$. \new{Let $D$ be the multicontext obtained from $p$ by labelling all the ${\preceq}$\=/minimal elements outside of $\dom(p)$ by~$\hole$.} Consider the valuations $\eta_2,\ldots,\eta_n$ that map all the ports of $D$ to $H_L$ with $\eta_i(u)= \sigma_i(u)$ for any port $u$ of $D$. For every $i \in \{2,\ldots,n\}$, the root label of $\sigma_i$ is the same as $D[\eta_i]$, because $\sigma_i$ is locally consistent with $t$ and $D$ is obtained from a~prefix of $t$.
By Lemma~\ref{lem:compose-strategies}, Alternator wins the game
\begin{align*}
	\gameH_D(D[\eta_2], \ldots, D[\eta_n]).
\end{align*}
This shows  that for every choice of $p$, Alternator has a~winning strategy in the remaining part of the game.

Now we move to the more difficult implication from 1 to 2.

Suppose that Alternator wins $\gameH_C(h_1, \ldots, h_n)$. Let $t$ be
the tree of type $h_1$ that Alternator plays in the first round. This tree
has prefix $C$. Also, for every finite prefix $D$ of $t$, Alternator
wins $\gameH_D(h_2, \ldots, h_n)$. By the induction assumption, for every finite
prefix $D$ of $t$, there is a~strategy tree
\begin{align*}
	\sigma_D = (t_D,\sigma_{D_2},\ldots,\sigma_{D_n})
\end{align*}
such that  $t_D$ has  prefix $D$, and the root sequence
of $\sigma_D$ is $(h_2, \ldots, h_n)$.

A~sequence of finite multicontexts ${(D_i)}_{i \in \omega}$ is said to converge to $t$ if all of the multicontexts are prefixes of $t$, and for every $j \in \omega$, only finitely many multicontexts have some port at depth at most $j$. By compactness, there is an infinite sequence of finite multicontexts ${(D_i)}_{i \in \omega}$ which converges to the tree $t$ and such that all of the sequences
\begin{align*}
	{(t_{D_i})}_{i \in \omega} \qquad  {(\sigma_{D_i 2})}_{i \in \omega} \qquad \ldots \qquad  {(\sigma_{D_i n})}_{i \in \omega}
\end{align*}
are convergent. Let the limits of these sequences be
\begin{align*}
 t_* \qquad \sigma_{*2} \qquad \ldots \qquad \sigma_{*n}.
\end{align*}
Because the sequence ${(D_i)}_{i \in \omega}$ converges to $t$, it follows that $t_*=t$.
For each $D$, the type trees
\begin{align*}
	(\sigma_{D2},\ldots,\sigma_{Dn})
\end{align*}
are locally consistent with $t_D$. Therefore, by Lemma~\ref{lem:limitofstrtrees}
it follows that the limits  $\sigma_{*2},\ldots,\sigma_{*n}$ are locally consistent with $t$.
Finally, define $\sigma_{*1}$ to be the unique type tree that is globally consistent with $t$.
We have just proved that
\[(\sigma_{*1},\ldots,\sigma_{*n})\]
is a~strategy tree. Because root values are preserved under limits, the root value of
this strategy tree is the desired $(h_1,\ldots,h_n)$.
\end{proof}

\subsubsection{Locally optimal strategy trees}%
\label{sssec:locally-opt}

To make the structure of a~strategy tree more rigid, we will introduce the~notion of a~\emph{locally optimal strategy tree}. In such a~strategy tree, whenever $\sigma_i(u)\neq\sigma_{i+1}(u)$, there is some concrete reason for that fact.

\begin{defi}%
\label{def_optimalstrategies}
Consider a~strategy tree $(t,\sigma_1,\ldots,\sigma_n)$ and let $v \in V_L\sqcup\{\one_L\}$. The strategy tree is called \emph{locally optimal for $v$} if for every $i \in \{2,\ldots,n\}$ and for every \newdef $\sigma'$, if $\sigma'$ is locally consistent with $t$ and
$v \cdot \sigma'(\epsilon) = v \cdot \sigma_i(\epsilon)$, then
\[\topdisc(\sigma_{i-1},\sigma_i) \leq \topdisc(\sigma_{i-1},\sigma'),\]
where $\topdisc$ is the discounted distance.
\end{defi}

If $\sigma$ is optimal for the context type $\one_L$ of the trivial context $\hole$ then we just say that $\sigma$ is \emph{locally optimal}. \new{Let $\Sigma_L$ denote the set of all locally optimal strategy trees for $L$.}

\begin{fact}%
\label{ft:locally-premultiply}
If $\sigma$ is locally optimal for $u\cdot v$ then $\sigma$ is locally optimal for $v$ as well.
\end{fact}

\begin{lem}%
\label{lem:locally-optimal}
Consider a~strategy tree $\sigma$ with root sequence $(h_1,\ldots,h_n)$. Let $v\in V_L\sqcup\{\one_L\}$ be a~context type. Then there exists a~locally optimal strategy tree $\sigma'$ for $v$ with root sequence $(h_1',\ldots,h_n')$ such that
\begin{equation}
v\cdot (h_1,\ldots,h_n)=v\cdot (h_1',\ldots,h_n').
\label{eq:preserve-types}
\end{equation}
\end{lem}

\begin{proof}
Take a~strategy tree $\sigma = (t, \sigma_1, \ldots, \sigma_n)$. Consider the set $\Sigma_2$ of type\=/labelled trees $\sigma_2'$ that are:
\begin{itemize}
\item locally consistent with $t$,
\item if $h$ (resp. $h'$) is the root value of $\sigma_2$ (resp. $\sigma_2')$ then $v\cdot h = v\cdot h'$.
\end{itemize}

\noindent
Fact~\ref{ft:consistent-closed} implies that $\Sigma_2$ is a~closed set. As a~closed subset of the compact space of all trees, $\Sigma_2$ is compact. It follows that some elements of $\Sigma_2$ minimise the discounted distance with respect to $\sigma_1$. We choose such an element as the new $\sigma_2$ and we iterate this mechanism to build a~locally optimal strategy tree $\bar{\sigma}$. This tree satisfies condition~\eqref{eq:preserve-types}.
\end{proof}

\new{
Notice that in the above construction, if $v\neq\one_L$ then the root sequence of the new strategy tree might be different than the original root sequence. However, if $v=\one_L$ then the root sequence is not modified.
}

\begin{cor}%
\label{cor:H-to-stra-trees}
If $(h_1,h_2,\ldots,h_n)\in\gameH_L$ then there exists a~locally optimal strategy tree $\sigma$ with root sequence $(h_1,\ldots,h_n)$.
\end{cor}

Using the above properties, without loss of generality we can restrict our attention to locally optimal strategy trees. One of the direct consequences of working with such strategy trees is expressed by the following lemma.

\begin{lem}%
\label{lem:loc-opt-to-prefix-change}
If $\sigma$ is a~locally optimal strategy tree, $u\preceq y$ two nodes of $\sigma$ then for all $i=1,\ldots,n$ we have
\[\sigma_i(y)\neq \sigma_{i+1}(y)\Rightarrow \sigma_i(u)\neq\sigma_{i+1}(u).\]
\end{lem}

\begin{proof}
Assume that $\sigma_{i}(y) \neq \sigma_{i+1}(y)$ and $\sigma_{i}(u) =
\sigma_{i+1}(u)$. We show that this contradicts local optimality.
Consider the \newdef $\sigma'_{i+1}$ defined as follows:
\begin{itemize}
\item for every $z$ such that $u\preceq z$, $\sigma'_{i+1}(z)=\sigma_i(z)$,
\item for all other nodes $z$, $\sigma'_{i+1}(z) = \sigma_{i+1}(z)$.
\end{itemize}

\noindent
By the construction, $\sigma'_{i+1}$ is locally consistent with $t$ (this
is by the definition for all nodes $z \neq u$, and because $\sigma_{i}(u)
= \sigma_{i+1}(u)$). Note that by the definition, for all nodes
$z$:
\[ \dist(\sigma_{i}(z),\sigma'_{i+1}(z)) \leq \dist(\sigma_{i}(z),\sigma_{i+1}(z))\]

\noindent
Moreover, $\dist(\sigma_{i}(y),\sigma'_{i+1}(y))=0$ and $\dist(\sigma_{i}(y)
,\sigma_{i+1}(y))=1$. Combining all this we obtain:

\begin{itemize}
\item $\topdisc(\sigma_{i},\sigma'_{i+1}) < \topdisc(\sigma_{i},\sigma_{i+1})$.
\item \new{$\sigma'_{i+1}(\epsilon) = \sigma_{i+1}(\epsilon)$, as sequences in ${(H_L)}^n$.}
\end{itemize}
This contradicts local optimality of $\sigma$.
\end{proof}

\begin{cor}%
\label{cor:loc-opt-to-alt-prefix}
If $\sigma$ is a~locally optimal strategy tree and $u\preceq y$ are two nodes of $\sigma$ then the root alternation of $\sigma.u$ is not smaller than the root alternation of $\sigma.y$.
\end{cor}

\subsubsection{Case distinction}%
\label{sssec:case-distinction}

We can now split the proof of the implication $4 \Rightarrow 3$ of Theorem~\ref{thm:mainforBC1} into two subcases, as discussed below. Then we will treat these cases separately, as expressed by Propositions~\ref{pro:if-bounded-limit} and~\ref{pro:if-unbounded-limit}.

Assume for the sake of contradiction that Condition~3 of Theorem~\ref{thm:mainforBC1} is violated, what means that $\gameH_L$ has unbounded alternation. Thus, by Corollary~\ref{cor:H-to-stra-trees} we know that the root alternation of the set $\Sigma_L$ of all locally optimal strategy trees for $L$ is unbounded. Now consider the following two dual subcases:

\begin{description}
\item[(C1)]\label{it:lim-bounded}
$\Sigma_L$ has unbounded root alternation and there exists a~subset $\Sigma'\subseteq\Sigma_L$ that has unbounded root alternation but bounded limit alternation.
\item[(C2)]\label{it:lim-unbounded}
$\Sigma_L$ has unbounded root alternation and every subset $\Sigma'\subseteq\Sigma_L$ with unbounded root alternation has also unbounded limit alternation.
\end{description}

\noindent
The following two propositions consider these cases separately:

\begin{prop}%
\label{pro:if-bounded-limit}
Assuming~(C1), Equation~\eqref{eq:cond-bool} is violated.
\end{prop}

\begin{prop}%
\label{pro:if-unbounded-limit}
Assuming~(C2), Equation~\eqref{eq:cond-limit} is violated.
\end{prop}

The proofs of these propositions are given in \new{Subsections~\ref{ssec:limit-bounded} and~\ref{ssec:limit-unbounded}}. Notice that since either Case~(C1) or Case~(C2) must hold, both these propositions together complete the proof of the implication $4 \Rightarrow 3$ of Theorem~\ref{thm:mainforBC1}.

Roughly speaking, Case~(C1) corresponds to the situation in which the high alternation of~$\gameH_L$ is achieved by modifications in the strategy trees that are spread inside their structure. On the opposite, Case~(C2) corresponds to the situation in which the high alternation occurs along some infinite branch of the considered strategy trees.

\subsection{Case (C1) --- limit alternation is bounded}%
\label{ssec:limit-bounded}

In this section we assume that $\Sigma'\subseteq\Sigma_L$ is a~set of locally optimal strategy trees such that the root alternation of $\Sigma'$ is unbounded but the limit alternation of $\Sigma'$ is bounded. Under that assumption we prove that Equation~\eqref{eq:cond-bool} is violated. We do this in two steps, using a~new object called \emph{strategy matrix} as an intermediary. Strategy matrices represent special strategies for Alternator in the game on tree types. In our first step, we show that if the root alternation of $\Sigma'$ is unbounded but the limit alternation of $\Sigma'$ is bounded then there exist special
strategy matrices of arbitrarily large size. Finally we show that the existence
of sufficiently large special strategy matrices violates Equation~\eqref{eq:cond-bool}.

\begin{defi}
A~\emph{strategy matrix} is a~rectangular matrix with entries from $V_L$ such that every row belongs to $\gameV_L$. The \emph{value} of a~column of a~strategy matrix is the value in $V_L$ obtained by multiplying the entries in that column in $V_L$ in the top\=/to\=/bottom order.
\end{defi}

\begin{defi} A~strategy matrix $M$ is called \emph{parity
alternating} if for some $n \in \omega$ it has $2n$ columns and $n$ rows, and one of the following conditions holds (see Figure~\ref{fig:strategy-matrix}):
\begin{enumerate}[label={(\alph*)}] 
	\item For every $i \in \{1,\ldots,n\}$,
    \begin{itemize}
        \item Columns $2i-1$ and $2i$ have the same entries in all rows except for row $i$.
        \item The values of columns $2i-1$ and $2i$ are different.
    \end{itemize}
    \item Condition (a) holds when the order of columns is reversed.
\end{enumerate}

\noindent
If the case $(a)$ holds then the matrix is called \emph{top\=/down} and if the case $(b)$
holds then it is called \emph{bottom\=/up}. The set of parity alternating matrices with
$n$ rows and $2n$ columns is denoted by $\palt{n}$.
\end{defi}

Figure~\ref{fig:strategy-matrix} depicts a~top\=/down parity alternating strategy matrix in $\palt{6}$. The picture presents columns $2i-1$ and $2i$ for $i=3$, the entries of these columns agree everywhere except the third row, where the values $x$ and $y$ are distinct. The differences in all the other pairs of columns are indicated by grey rectangles, their values are not written down for the sake of clarity. It is important that the definition of a~parity alternating strategy matrix requires the total values of the respective columns to differ, not only the entries in grey rectangles.

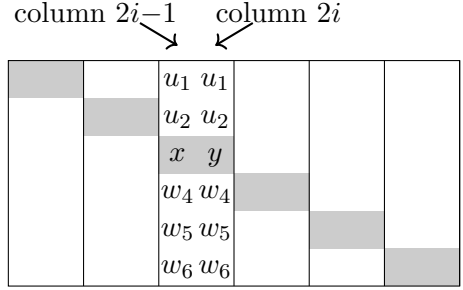
\begin{figure}
\begin{center}
\begin{tikzpicture}

\foreach \i in {1,...,6} {
	\fill[gray!40] (\i-1, -\i*0.5+0.5) rectangle ++(1,-0.5);
}

\draw (0,0) rectangle (6,-3);

\foreach \i in {1,...,6} {
	\draw (\i, 0) -- ++(0, -3);
}

\foreach \i in {1,...,2} {
	\node[toC] at (2.25, -\i*0.5+0.25) {$u_\i$};
	\node[toC] at (2.75, -\i*0.5+0.25) {$u_\i$};
}

\node[toC] at (2.25, -3*0.5+0.25) {$x$};
\node[toC] at (2.75, -3*0.5+0.25) {$y$};

\foreach \i in {4,...,6} {
	\node[toC] at (2.25, -\i*0.5+0.25) {$w_\i$};
	\node[toC] at (2.75, -\i*0.5+0.25) {$w_\i$};
}

\node[toL] at (2.25, 0.6) {column $2i{-}1$};
\node[toR] at (2.75, 0.6) {column $2i$};

\draw[edge,<-] (2.25, 0.2) -- ++(-0.5, +0.3);
\draw[edge,<-] (2.75, 0.2) -- ++(+0.5, +0.3);
\end{tikzpicture}
\end{center}
\caption{A~matrix in $\palt{6}$.}%
\label{fig:strategy-matrix}
\end{figure}

\subsubsection{Constructing strategy matrices}%
\label{ssec:constructing-matrices}

As we have already said, we want to prove that under our assumptions we can build arbitrary large strategy matrices: the goal that we aim now, is to prove that if $\Sigma'$ has unbounded root alternation and bounded limit alternation, then $\palt{n}$ is non\=/empty for every $n \in \omega$.

Fix some arbitrary $n \in \omega$. We want to construct a~strategy matrix $M \in \palt n$. Let $\ell$ be the maximal limit alternation among the strategy trees in $\Sigma'$. Choose a~strategy tree $\sigma \in \Sigma'$ with root alternation at least $\ell \cdot 2^{n^2}$ and let $m$ be the duration of $\sigma$, i.e.~$\sigma=(t, \sigma_1, \ldots, \sigma_m)$.

During the proof we will use a~known combinatorial fact that we recall now.

\begin{thm}[Erd\"{o}s--Szekeres Theorem, see~\cite{erdosszekeres}] 
For given $r, s \in \N$, any sequence of length at least $(r {-} 1)(s {-} 1) {+} 1$ contains a~monotonically increasing subsequence of length $r$ or a~monotonically decreasing subsequence of length $s$.
\end{thm}

\begin{lem}%
\label{lem:path}
There are consecutive contexts zones $X_1,\ldots,X_n$ in $\sigma$ and a~sequence of rounds  $i_1,\ldots,i_{n}
\in \{1,\ldots,m-1\}$ such that the sequence of rounds  is either strictly increasing or strictly decreasing, and
\begin{align*}
\val(\sigma,X_j,i_j) \neq \val(\sigma,X_j,i_j{+}1) \qquad \text{for every $j \in \{1,\ldots,n\}$}
\end{align*}
\end{lem}

\begin{proof}
For a~round  $i \in \{1,\ldots,{m{-}1}\}$, define:
\[\change_i(\sigma) =  \{ x \in {\{\dL,\dR\}}^\ast \mid \sigma_i(x) \neq \sigma_{i+1}(x)\}.\]
Now we prove three different claims. Once these claims are proved, we easily finish the proof of the lemma.

\paragraph*{\bf Claim~1} Let $X$ be  a~context zone, and let $\pi$ be an infinite path that passes through the root of $X$, but not through the port. Then
\[\pi \subseteq \change_i(\sigma) \quad \Longrightarrow \quad \val(\sigma,X,i) \neq \val(\sigma,X,i{+}1)\]
holds for every round $i \in \{1,\ldots,m-1\}$.

\begin{proof}[Proof of Claim~1]
This is by local optimality of $\sigma$. Assume that $\pi \subseteq
\change_i(\sigma)$ and $\val(\sigma,X,i) = \val(\sigma,X,i{+}1)$ for some
round $i$. We construct a~\newdef $\sigma_{i+1}'$ that is closer to $\sigma_i$ than
$\sigma_{i+1}$ regarding the discounted distance $\topdisc$ and with the same
root value as $\sigma_{i+1}$, contradicting local optimality.

Consider the \newdef $\sigma'_{i+1}$ defined as follows:

\begin{itemize}
\item For nodes $x \in X$ $\sigma'_{i+1}(x)=\sigma_i(x)$. Hence, for $x \in X$
\begin{align*}
\dist(\sigma_i(x),\sigma'_{i+1}(x)) = 0 \leq \dist(\sigma_i(x),\sigma_{i+1}(x))
\end{align*}

\item For other nodes $y$ we define $\sigma'_{i+1}(y)=\sigma_{i+1}(y)$. Hence
\[\dist(\sigma_i(y),\sigma'_{i+1}(y)) = \dist(\sigma_i(y),\sigma_{i+1}(y)).\]
\end{itemize}
Because $\pi \subseteq \change_i(\sigma)$, there exists at least one node
$x \in X$ such that s
\[\dist(\sigma_i(x),\sigma_{i+1}(x))=1.\]
It follows that:
\begin{align*}
\topdisc(\sigma_{i},\sigma'_{i+1}) < \topdisc(\sigma_{i},\sigma_{i+1})
\end{align*}
Moreover since $\val(\sigma,X,i) = \val(\sigma,X,i{+}1)$, we have
$\sigma'_{i+1}(\epsilon)= \sigma_{i+1}(\epsilon)$. This contradicts local
optimality of $\sigma$.
\end{proof}

\paragraph{\bf Claim~2}
There is a~set $\Pi$ of at least $2^{n^2}$ distinct infinite paths, and a~function $\fun{\when}{\Pi}{\{1,\ldots,m-1\}}$  such that for every $\pi\in\Pi$
\[\pi \subseteq\change_{\when(\pi)}(\sigma).\]

\begin{proof}[Proof of Claim~2]
By consistency of a~strategy tree, every node in $\change_j(\sigma)$ has at least one child in $\change_j(\sigma)$. Therefore, each of the non\=/empty sets $\change_j(\sigma)$ contains at least one infinite path. By the assumption on the root alternation there are at least $\ell \cdot 2^{n^2}$ rounds $j$ where $\change_j(\sigma)$ is non\=/empty. By the assumption on limit alternation, an infinite path can be contained in sets $\change_j(\sigma)$ for at most $\ell$ different values of $j$.
\end{proof}

\paragraph*{\bf Claim~3} Let $\Pi$ be a~set of $2^k$ distinct infinite paths in a~binary tree. There exist consecutive context zones $X_1,\ldots,X_k$ and paths
$\pi_1,\ldots, \pi_k \in \Pi$ such that for every $i \in \{1,\ldots,k\}$, the path $\pi_i$ passes through the ports of the context zones $X_1,\ldots,X_{i-1}$, but not through the port of $X_i$.

\begin{proof}[Proof of Claim~3]
The proof is by induction on $k$. The induction base of $k=1$ is obvious.

Now assume that the thesis holds for $k$ and consider a~set $\Pi$ of $2^{k+1}$ paths.
Let $u$ be the deepest node in the tree that belongs to all paths of $\Pi$
($u$ exists since the root belongs to all paths of $\Pi$). Let $\Pi_\dL$
(respectively, $\Pi_\dR$) be those paths in $\Pi$ that pass through the left
child of $u$ (respectively, the right child of $u$). One of the sets $\Pi_\dL$
or $\Pi_\dR$ must have at least half of the paths, i.e.~at least $2^k$ paths.
By symmetry, assume that $\Pi_\dL$ has at least $2^k$ paths and let $y$ be the
left child of $u$. We apply the induction hypothesis to $\Pi_\dL$ and obtain
paths $\pi_2,\ldots,\pi_{k+1} \in \Pi_\dL$ and consecutive context zones $X_2,
\ldots,X_{k+1}$ such that for every $i \in \{2,\ldots,k+1\}$, path $\pi_i$ passes through the ports of contexts $X_2,\ldots,X_{i-1}$, but not through the port of $X_i$.

We slightly modify $X_2$ by setting $y$ as its root. Note that this does
not affect the properties of the paths $\pi_2,\ldots,\pi_{k+1} \in \Pi_\dL$. Now, we define $X_1$ as the context zone with the root $u$ and the root of $X_2$ as the port. Let $\pi_1$ be some arbitrary path in $\Pi_\dR$. By the definition, $X_1,\dots,X_{k+1}$ are consecutive context zones. Moreover, the paths $\pi_2,\ldots, \pi_{k+1}$ are paths of $\Pi_\dL$ and therefore pass through $y$, i.e.~the port of $X_\dL$. Finally, by the definition $\pi_1$ passes through the right child of $u$ and therefore not through the port of $X_1$.
\end{proof}

Now we can easily complete the proof of Lemma~\ref{lem:path}. Let $\Pi$ and the function $\when$ be as in Claim~2. Apply Claim~3 to $\Pi$, yielding a~sequence of paths,  $\tau_1,\ldots,\tau_{n^2}$,
and a~sequence of context zones, $Y_1,\ldots,Y_{n^2}$, such that for every $i \in \{1,\ldots,n^2\}$, the path $\tau_i$ passes through the ports of context zones $Y_1,\ldots,Y_{i-1}$, but not through the port of $Y_i$. Consider the sequence
\begin{align*}
	\when(\tau_1),\ldots,\when(\tau_{n^2}) \in \{1,\ldots,m-1\}.
\end{align*}
By Erd\"{o}s--Szekeres Theorem, we can find a~sequence of indexes 
\begin{align*}
	j_1 < \cdots < j_n \in \{1,\ldots,n^2\}
\end{align*}
such that the  sequence
\begin{align*}
	\when(\tau_{j_1}),\ldots,\when(\tau_{j_n})
\end{align*}
is either increasing or decreasing. For $i \in \{1,\ldots,n\}$, define $\pi_i$ to  be $\tau_{j_i}$ and $X_i$ to be the union of the context zones
\[ Y_{j_i} \cup \cdots \cup Y_{j_{i+1}-1}.\]
By the construction, we know that the path $\pi_i$ passes through the ports of the context zones $X_1,\ldots,X_{i-1}$, but not through the port of the context zone $X_i$. By Claim~1 we know that
\begin{align*}
\val\big(\sigma,X_i,\when(\pi_i)\big) \neq \val\big(\sigma,X_i,\when(\pi_i){+}1\big).
\end{align*}
Therefore, \new{the proof of Lemma~\ref{lem:path}} is complete if we define $i_1,\ldots,i_n$ to be
\[\when(\pi_1),\ldots,\when(\pi_n). \qedhere\]
\end{proof}

\begin{defi} Let $M$ be a~matrix and consider a~row $j$ and a~column $i$ of $M$ which is not the first column. We define a~new column, denoted by
\begin{align*}
	\mathrm{almostcopy}^{\neq j}_{i}(M),
\end{align*}
as follows: $\mathrm{almostcopy}^{\neq j}_{i}(M)$ is equal to column $i$ of $M$ in all rows except for row $j$, where it is equal to column $i-1$ of $M$.
\end{defi}

\begin{defi}
Let $\sigma$ be a~strategy tree of duration $m$ and let $X_1,\dots,X_n$ be consecutive context zones. We define a~matrix with $n$ rows and $m$
columns as follows (for $i=1,\ldots,m$ and $j=1,\ldots,n$):
\[\smat\big(\sigma,X_1,\ldots,X_n\big)\,[i,j]\eqdef\val(\sigma,X_j,i).\]
\end{defi}

Note that it follows from Fact~\ref{fact:zone-strat} that every row of $\smat(\sigma,X_1,\ldots,X_n)$ belongs to $\gameV_L$ and therefore it is a~strategy matrix.

\begin{lem}%
\label{lem:almost-copy}
Let $N$ be a~strategy matrix defined by
\[N = \smat(\sigma,X_1,\ldots,X_n),\]
for some locally optimal strategy tree $\sigma$ and consecutive context zones $X_1,\ldots,X_n$.
Let $j$ be a~row and $i$ a~column, which is not the first column. If columns $i$ and $i-1$
in $N$ have different entries in row $j$ then the value of the column
\[\mathrm{almostcopy}^{\neq j}_{i}(N)\]
is different from the value of column $i$ in $N$.
\end{lem}

\begin{proof}
This is a~consequence of local optimality of $\sigma$. The proof is the same as Claim~2 in the proof of Lemma~\ref{lem:path}.
\end{proof}

Now we are ready to prove the first intermediary step: constructing a~matrix in $\palt n$.

\begin{prop}
Under our assumptions about the strategy tree $\sigma$ we can construct a~matrix $M \in \palt n$.
\end{prop}

\begin{proof}
Let $X_1,\ldots,X_n$ and $i_1,\ldots,i_n$ be as in Lemma~\ref{lem:path}. Consider the strategy matrix
$N = \smat(\sigma,X_1,\ldots,X_n)$.

By Lemma~\ref{lem:path} we know that for every $j \in \{1,\ldots,n\}$ the entries in row $j$ are different in columns $i_j$ and $i_j+1$.

Suppose first that the sequence $i_1,\ldots,i_n$ is strictly increasing.
Define a~new matrix $M$, which has $n$ rows and $2n$ columns as follows.
\begin{itemize}
\item For  $j \in \{1,\ldots,n\}$, column $2j-1$ of $M$ is $\mathrm{almostcopy}^{\neq j}_{i_j+1}(N)$.
\item For  $j \in \{1,\ldots,n\}$, column $2j$ of $M$ is column $i_j+1$ of $N$.
\end{itemize}

\noindent
Figure~\ref{fig:construction-of-M} depicts the process of constructing the matrix $M$. Gray regions in the matrix $N$ indicate pairs of distinct values in a~row $j$ of the consecutive columns $i_j$ and $i_j+1$. We assume that the sequence $i_j$ is strictly increasing (the other case leads to a~bottom\=/up parity alternating strategy matrix).

\begin{figure}
\begin{center}
\begin{tikzpicture}

\draw (-2, +3.6) edge[obrace] node{$m$} ++(9,0);

\draw (-2.1, +1.0) edge[lbrace] node{$n$} ++(0,2.5);

\node[toC] at (6.5, 3.0) {$N$};

\node[toC] at (0.75, 4.4) {column $i_j$};

\draw[edge,<-] (0.75, 3.8) -- ++(0, +0.3);

\foreach \j/\i in {1/1.5, 2/3.5, 3/5.5, 4/6.0, 5/7.5} {
	\fill[gray!40] (-2+\i-1, +3.5-\j*0.5+0.5) rectangle ++(1,-0.5);
}

\draw (1.0, 3.5) rectangle (1.5, 1.0);

\tikzstyle{almostcopy}=[pattern=north east lines, pattern color=black!80]

\path[almostcopy] (1.0, 3.5) rectangle (1.5, 3.0);
\path[almostcopy] (0.5, 3.0) rectangle (1.0, 2.5);
\path[almostcopy] (1.0, 2.5) rectangle (1.5, 1.0);

\draw (-2,+3.5) rectangle (7,+1);

\foreach \i in {1,...,5} {
	\fill[gray!40] (\i-1, -\i*0.5+0.5) rectangle ++(1,-0.5);
}

\draw (0,0) rectangle (5,-2.5);

\foreach \i in {1,...,5} {
	\draw (\i, 0) -- ++(0, -2.5);
}

\node[toC] at (4.5, -0.5) {$M$};

\path[almostcopy] (1.0, 0.0) rectangle (1.5, -2.5);

\draw (-0.1, -2.5) edge[lbrace] node{$n$} ++(0,2.5);

\node[toL] at (1.25, -3.4) {column $2j{-}1$};
\node[toR] at (1.75, -3.4) {column $2j$};

\draw[edge,<-] (1.25, -2.8) -- ++(-0.5, -0.3);
\draw[edge,<-] (1.75, -2.8) -- ++(+0.5, -0.3);

\draw[edge,->] (1.25, 0.8) -- ++(0.0, -0.6);
\draw[edge,->] (1.25, 0.8) -- ++(0.5, -0.6);
\end{tikzpicture}
\end{center}
\caption{Construction of the matrix $M$ from $N$. Column $2j-1$ of $M$ equals $\mathrm{almostcopy}$ of column $i_j+1$ of $N$ (i.e.~the dashed part), while column $2j$ just equals column $i_j+1$ of $N$.}%
\label{fig:construction-of-M}
\end{figure}
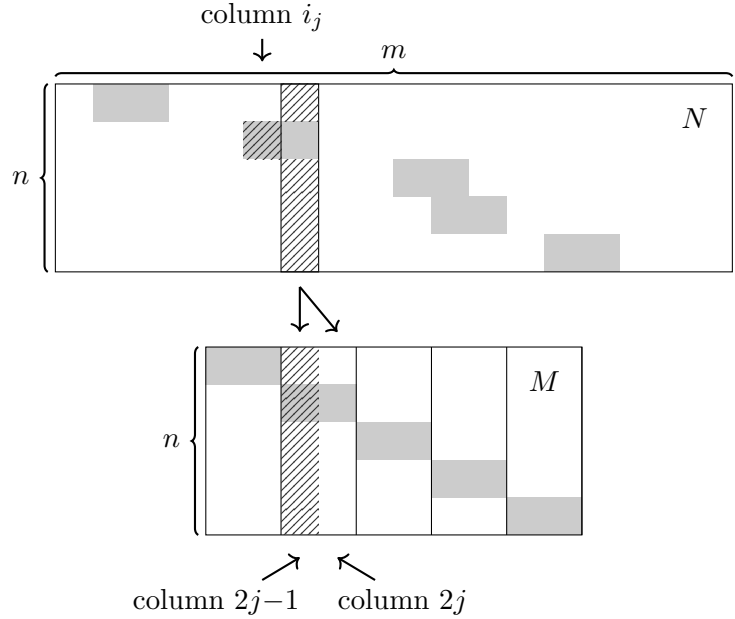

Now we show that $M \in  \palt n$. The dimensions of the matrix $M$ are correct: it has $n$ rows and $2n$ columns. We now show that $M$ is a~strategy matrix, which means that each row belongs to $\gameV_L$.
For $j \in \{1,\ldots,n\}$,  let us see how row $j$ of $M$
\[M[j,1], \ldots, M[j,2n]\]
depends on row $j$ of $N$
\[N[j,1],\ldots, N[j,m].\]

By reading the definition of $M$, we see that the dependency is
\begin{itemize}
\item When $k \neq j$, then $M[j,2k-1]=M[j,2k]=N[j,i_k+1]$.
\item When $k=j$, then  $M[j,2k-1]=N[j,i_k]$ and $M[j,2k]=N[j,i_k+1]$.
\end{itemize}
It follows that row $j$ of $M$ is obtained from row $j$ of $N$ by eliminating some letters and duplicating some other letters. Since  $\gameV_L$ is closed under eliminating and duplicating letters, and since $N$ was a~strategy matrix, it follows that also $M$ is a~strategy matrix.

By Lemma~\ref{lem:almost-copy}, for every $j \in \{1,\ldots,n\}$, the values of columns $2j-1$ and $2j$ in $M$ are different. By the construction, columns $2j-1$ and $2j$ in $M$ have the same entries, except for row $j$. So $M$ is parity alternating.

When the sequence $i_1,\ldots,i_n$ is strictly decreasing, the matrix $M$ is defined like for a~strictly increasing sequence, except that the columns of $M$ are filled in not from left to right, but from right to left. Formally speaking:
\begin{itemize}
\item For  $j \in \{1,\ldots,n\}$, column $2(n-j+1)-1$ of $M$ is $\mathrm{almostcopy}^{\neq j}_{i_j+1}(N)$.
\item For  $j \in \{1,\ldots,n\}$, column $2(n-j+1)$ of $M$ is column $i_j+1$ of $N$.
\end{itemize}
The proof that $M$ belongs to $\palt n$ is the same as above.
\end{proof}

\subsubsection{From strategy matrices to violation of~\texorpdfstring{\eqref{eq:cond-bool}}{(bool)}}

So now we can do the second intermediary step. Let us assume that $\palt n$ is non\=/empty for any $n$. By the symmetry we can assume that for every $n$ there exists a~top\=/down parity alternating matrix of size $n$ (the other case is handled symmetrically): under this assumption we prove that Equation~\eqref{eq:cond-bool} is violated. To do that, we need to start from a~parity alternating matrix of a~large size, and by consecutive simplifications, construct a~small matrix from which the violation can be easily extracted.

Take a~strategy matrix $M$.  Our aim is define a~matrix obtained form $M$ by \emph{merging} two consecutive rows $i$ and $i+1$ of $M$. Let
\[(v_{1,1},\ldots,v_{1,n}), \ldots, (v_{k,1},\ldots,v_{k,n}) \in V_L^n.\]
be the rows in $M$. The merge operation removes rows
\[(v_{i,1},\ldots,v_{i,n}) \quad \text{and} \quad (v_{(i+1),1},\ldots,v_{(i+1),n})\]
and replaces them by the row
\[(v_{i,1} \cdot v_{(i+1),1},\ldots,v_{i,n} \cdot v_{(i+1),n}),\]
which is the product of the two removed rows in the monoid $V_L^n$. Clearly, the result of the merging operation is also a~strategy matrix with the same values of all the columns.

Now we define four \emph{safe rules}, i.e.~rules that preserve parity alternating strategy matrices.

\begin{lem}%
\label{lem:safe-rules}
For every $M \in \palt n$ that is top\=/down and $i \in \{1,\ldots,n\}$, applying the following rules yields a~parity alternating strategy matrix in $\palt {n-1}$:
\begin{itemize}
\item remove columns $2i-1$ and $2i$ and then merge row $i$ with $i{+}1$;
\item remove columns $2i-1$ and $2i$ and then merge row $i{-}1$ with $i$;
\item remove the first two columns and remove the first row;
\item remove the last two columns and remove the last row.
\end{itemize}
The same holds for a~bottom\=/up parity alternating matrices but then we count the columns in the reversed order.
\end{lem}

\begin{proof}
Immediate by the definition of the rules.
\end{proof}

Notice that $ux\neq uy$ implies that $x\neq y$, but $uvx\neq uvy$ does not imply that $ux\neq uy$, therefore we cannot safely remove columns $2i-1$ and $2i$ and remove row $i$ except for $i=1$ or $i=n$.

Consider a~parity alternating strategy matrix $M\in\palt{n}$ and two indices $i\in\{1,\ldots,n\}$ and $j\in\{1,\ldots,2n\}$. We say that the entry $M[i,j]$ is \emph{above diagonal} (resp. \emph{below diagonal}) if:
\begin{itemize}
\item if $M$ is top\=/down then $2i$ must be smaller than $j$ (resp. $2i$ must be grater than $j+1$),
\item if $M$ is bottom\=/up parity alternating then $2(n-i)$ must be greater than $j-1$ (resp. $2(n-i)$ must be smaller than $j-2$).
\end{itemize}

\noindent
The following picture depicts the regions above and below the diagonal of a~top\=/down parity alternating matrix in $\palt{6}$.

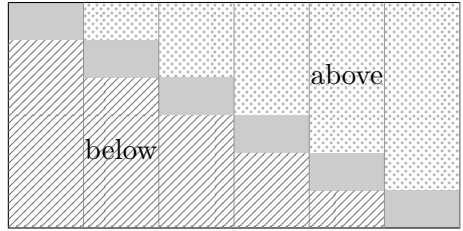
\begin{center}
\begin{tikzpicture}

\foreach \i in {1,...,6} {
	\fill[gray!40] (\i-1, -\i*0.5+0.5) rectangle ++(1,-0.5);
}

\draw (0,0) rectangle (6,-3);

\foreach \i in {1,...,6} {
	\draw[black!40] (\i, 0) -- ++(0, -3);
}

\path[pattern=crosshatch dots, pattern color=black!30] (1,0) -- (6,0) -- (6,-2.5)
  \foreach \i in {1,...,5} {-- (6-\i, -3.0+\i*0.5) -- ++(0, 0.5)}
  -- cycle;

\path[pattern=north east lines, pattern color=black!50] (0,-0.5)
  \foreach \i in {1,...,5} {-- (\i, -\i*0.5) -- ++(0, -0.5)}
  -- (0, -3) -- cycle;

\node[toC] at (4.5, -1) {above};

\node[toC] at (6-4.5, -3+1) {below};
\end{tikzpicture}
\end{center}

\begin{defi}
A~parity alternating strategy matrix $M$ is upper (resp.\ lower) \emph{idempotent} if there exists an idempotent $e \in V_L$ such that every entry above (resp.\ below) the diagonal of $M$ equals $e$. $M$ is called just \emph{idempotent} if it is both upper and lower idempotent (possibly with two distinct idempotents $e,e'\in V_L$).
\end{defi}

\begin{fact}
The safe rules preserve the fact that a~given parity alternating matrix is upper (resp.\ lower) idempotent.
\end{fact}

\begin{lem}%
\label{lem:matrix-lemma}
For each  $m \in \N$ there is some $n \in \N$ such that for any matrix
in $\palt n$ there is a~sequence of safe rules that yields a~matrix $N
\in \palt m$ that is upper (resp.\ lower) idempotent.
\end{lem}

\begin{proof}
By the symmetry we will only deal with the upper idempotent case.
The proof uses Ramsey Theorem for hypergraphs with edges of size $3$.
This theorem says that for every  $m \in \N$ there exists a~number
$f(m)$ such that for any complete hypergraph with edges coloured
over $V_L$, there exists a~complete sub\=/hypergraph of size $m$ in which
all edges share the same colour. We choose $n = f(m+1)$.

Fix a~matrix $M \in \palt n$. Again by the symmetry we assume that $M$ is a~top\=/down parity alternating matrix. Consider the hypergraph where the nodes are $\{1,\ldots,n\}$ and an edge
$\{ i < j < k\}$ is coloured by the value obtained by multiplying, in the monoid $V_L$, the entries that appear in rows $i+1,\ldots,j$ of column $2k$. By the choice of $n$, we can apply Ramsey Theorem to this colouring and get a~subset of size $m+1$:
\[I = \{i_0<i_1<\cdots<i_m\}\subseteq \{1,\ldots,n\}\]
such that all the hyperedges on $I$ have the same colour, say $e \in V_L$.
This colour must be an idempotent, i.e.~it must satisfy $e = ee$. This situation
is illustrated below, with a~matrix in $\palt{10}$. The multiplication in $V_L$ of the entries in the regions marked by red rectangles always gives the idempotent $e$. The regions come in pairs, for columns $2i_\ell$ and $2i_\ell-1$ but as the matrix is parity alternating, the entries in these columns agree (except row $i_\ell$).

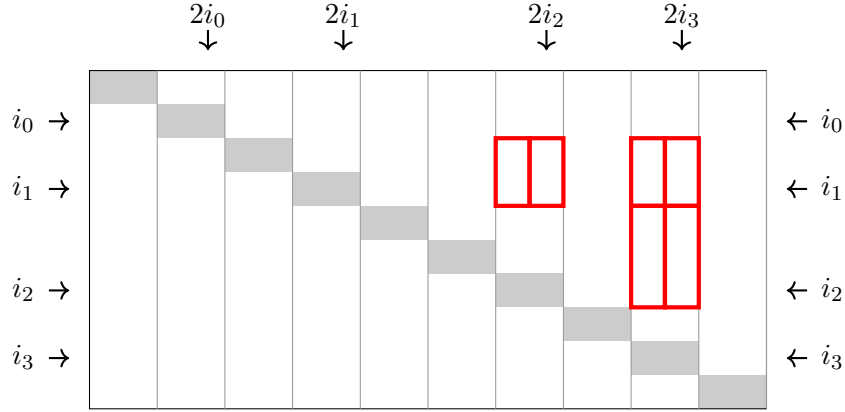
\begin{center}
\begin{tikzpicture}[scale=0.9]
\tikzstyle{region}=[draw=red, ultra thick]

\foreach \i in {1,...,10} {
	\fill[gray!40] (\i-1, -\i*0.5+0.5) rectangle ++(1,-0.5);
}

\draw (0,0) rectangle (10,-5);

\foreach \i in {1,...,10} {
	\draw[black!40] (\i, 0) -- ++(0, -5);
}

\foreach \x/\i in {0/2,1/4,2/7,3/9} {
	\draw[edge] (-0.6, -\i*0.5+0.25) -- ++(0.3, 0);
	\node[toL] at (-0.8, -\i*0.5+0.25) {$i_\x$};

	\draw[edge] (10.6, -\i*0.5+0.25) -- ++(-0.3, 0);
	\node[toR] at (10.8, -\i*0.5+0.25) {$i_\x$};

	\draw[edge] (\i-0.25, +0.6) -- ++(0, -0.3);
	\node[toC] at (\i-0.25, +0.8) {$2i_\x$};
}

\path[region] (8.0,-1) rectangle ++(0.5,-1.0);
\path[region] (8.5,-1) rectangle ++(0.5,-1.0);

\path[region] (8.0,-2) rectangle ++(0.5,-1.5);
\path[region] (8.5,-2) rectangle ++(0.5,-1.5);

\path[region] (6.0,-1) rectangle ++(0.5,-1.0);
\path[region] (6.5,-1) rectangle ++(0.5,-1.0);
\end{tikzpicture}
\end{center}

Now we apply a~sequence of safe rules. First, we remove the first $i_0$ rows and first $2i_0$ columns. Then we remove the last $n-i_m$ rows and last $2(n-i_m)$ columns. We assume that the indices of rows and columns are shifted accordingly to the operations we perform; i.e.~now~$i_\ell\eqdef i_\ell-i_0$ for $\ell=0,\ldots,m$. After these operations, the matrix gets the following shape.

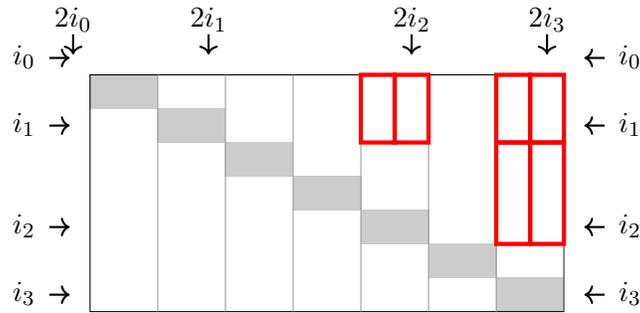
\begin{center}
\begin{tikzpicture}[scale=0.9]
\tikzstyle{region}=[draw=red, ultra thick]

\foreach \i in {1,...,7} {
	\fill[gray!40] (\i-1, -\i*0.5+0.5) rectangle ++(1,-0.5);
}

\draw (0,0) rectangle (7,-3.5);

\foreach \i in {1,...,7} {
	\draw[black!40] (\i, 0) -- ++(0, -3.5);
}

\foreach \x/\i in {0/0,1/2,2/5,3/7} {
	\draw[edge] (-0.6, -\i*0.5+0.25) -- ++(0.3, 0);
	\node[toL] at (-0.8, -\i*0.5+0.25) {$i_\x$};

	\draw[edge] (7.6, -\i*0.5+0.25) -- ++(-0.3, 0);
	\node[toR] at (7.8, -\i*0.5+0.25) {$i_\x$};

	\draw[edge] (\i-0.25, +0.6) -- ++(0, -0.3);
	\node[toC] at (\i-0.25, +0.8) {$2i_\x$};
}

\path[region] (6.0,-0) rectangle ++(0.5,-1.0);
\path[region] (6.5,-0) rectangle ++(0.5,-1.0);

\path[region] (6.0,-1) rectangle ++(0.5,-1.5);
\path[region] (6.5,-1) rectangle ++(0.5,-1.5);

\path[region] (4.0,-0) rectangle ++(0.5,-1.0);
\path[region] (4.5,-0) rectangle ++(0.5,-1.0);
\end{tikzpicture}
\end{center}

Now, for $\ell=1,2,\ldots,m$ we remove columns $2i_{\ell-1}+1,\ldots, 2i_\ell-2$ and merge rows $i_{\ell-1}+1,\ldots,i_\ell$ into one row (with value $e$ in columns $2i_{\ell+1}$, \ldots,$2i_m$). This way, we get the following matrix, which is upper idempotent.

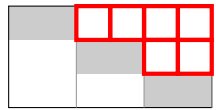
\[
\begin{tikzpicture}[scale=0.9]
\tikzstyle{region}=[draw=red, ultra thick]

\foreach \i in {1,...,3} {
	\fill[gray!40] (\i-1, -\i*0.5+0.5) rectangle ++(1,-0.5);
}

\draw (0,0) rectangle (3,-1.5);

\foreach \i in {1,...,3} {
	\draw[black!40] (\i, 0) -- ++(0, -1.5);
}

\path[region] (2.0,-0.0) rectangle ++(0.5,-0.5);
\path[region] (2.5,-0.0) rectangle ++(0.5,-0.5);

\path[region] (2.0,-0.5) rectangle ++(0.5,-0.5);
\path[region] (2.5,-0.5) rectangle ++(0.5,-0.5);

\path[region] (1.0,-0.0) rectangle ++(0.5,-0.5);
\path[region] (1.5,-0.0) rectangle ++(0.5,-0.5);
\end{tikzpicture}
\qedhere
\]
\end{proof}

\begin{cor}%
\label{cor:matrix-lemma}
For each  $m \in \N$ there is some $n \in \N$ such that for any matrix
in $\palt n$ there is a~sequence of safe rules that yields a~matrix $N
\in \palt m$ that is idempotent.
\end{cor}

\begin{proof}
Since the safe rules preserve the fact that a~matrix is upper (resp.\ lower) idempotent, we can apply Lemma~\ref{lem:matrix-lemma} twice, once to get an upper idempotent matrix and then to make it also lower idempotent.
\end{proof}

\begin{fact}%
\label{ft:idem-remove}
If $M$ is an idempotent parity alternating matrix, the operation that removes columns $2i$ and $2i-1$ and removes row $i$ preserves the fact that the matrix is idempotent parity alternating.
\end{fact}

\begin{proof}
If $i=1$ or $i=n$ then the above operation is safe. For $1<i<n$ we use the fact that the matrix is idempotent, so the values of columns $2i-1$ and $2i$ depend only on the unique entry in that column which is neither above nor below the diagonal. The following picture depicts an idempotent parity alternating matrix in $\palt{6}$ where the upper idempotent is $u$ and the lower idempotent is $w$. Removing row $4$ and columns $7$ and $8$ of that matrix preserves the fact that the matrix is idempotent parity alternating.

\[
\begin{tikzpicture}

\foreach \i in {1,...,6} {
	\fill[gray!40] (\i-1, -\i*0.5+0.5) rectangle ++(1,-0.5);
	\node[toC] at (\i-0.75, -\i*0.5+0.25) {$x_\i$};
	\node[toC] at (\i-0.25, -\i*0.5+0.25) {$y_\i$};
}

\draw (0,0) rectangle (6,-3);

\foreach \x in {2,...,6} {
	\tikzEvalInt{\maxX}{\x-1}
	\foreach \y in {1,...,\maxX} {
		\node[toC] at (\x-0.75, -\y*0.5+0.25) {$u$};
		\node[toC] at (\x-0.25, -\y*0.5+0.25) {$u$};

		\node[toC] at (6-\x+0.75, -3+\y*0.5-0.25) {$w$};
		\node[toC] at (6-\x+0.25, -3+\y*0.5-0.25) {$w$};
	}
}

\path[pattern=north east lines, pattern color=black!50] (0,-1.5) rectangle ++(6,-0.5);
\path[pattern=north east lines, pattern color=black!50] (3,+0.0) rectangle ++(1,-3.0);
\end{tikzpicture}
\qedhere
\]
\end{proof}

\begin{lem}%
\label{lem:very-regular-matrix}
If $\palt{n}$ is non\=/empty for arbitrarily big $n$ then there are $u,w,x,y \in V_L$ such that $u$ and $w$ are idempotents and $\palt 4$ contains the matrix
\begin{center}
\begin{tikzpicture}

\foreach \i in {1,...,4} {
	\fill[gray!40] (\i-1, -\i*0.5+0.5) rectangle ++(1,-0.5);
	\node[toC] at (\i-0.75, -\i*0.5+0.25) {$x$};
	\node[toC] at (\i-0.25, -\i*0.5+0.25) {$y$};
}

\draw (0,0) rectangle (4,-2);

\foreach \x in {2,...,4} {
	\tikzEvalInt{\maxX}{\x-1}
	\foreach \y in {1,...,\maxX} {
		\node[toC] at (\x-0.75, -\y*0.5+0.25) {$u$};
		\node[toC] at (\x-0.25, -\y*0.5+0.25) {$u$};

		\node[toC] at (4-\x+0.75, -2+\y*0.5-0.25) {$w$};
		\node[toC] at (4-\x+0.25, -2+\y*0.5-0.25) {$w$};
	}
}
\end{tikzpicture}
\end{center}
\end{lem}

\begin{proof}
By the hypothesis, we can apply Corollary~\ref{cor:matrix-lemma} for $m=4*|V_L|^2$. Let $N$ be the resulting
matrix. Each row of that matrix is of the form $w\cdots w\ \cdot x_i \cdot y_i \cdot\ u\cdots u$. Thus, there exists a~pair $(x,y)$ that appears as $(x_i,y_i)$ in at least $4$ rows. Using Fact~\ref{ft:idem-remove} we can remove all the remaining rows (and the corresponding pairs of columns) from $N$ and the result has the above form.
\end{proof}

We can conclude the proof of Proposition~\ref{pro:if-bounded-limit} by showing that under our assumptions Equation~\eqref{eq:cond-bool} is violated.

Let $M$ be the matrix described in Lemma~\ref{lem:very-regular-matrix}. Let $\{u_1,\ldots,u_{8}\}$ be the values
of all the columns in $M$. The matrix is in $\palt{4}$ so $u_{3} \neq u_{4}$.
Because $u$ and $w$ are idempotent,
\[u_3=uxww=uxw=uuxw=u_5\]
For the same reason, $u_4=u_6$. This is depicted in the picture below
\begin{center}
\begin{tikzpicture}

\foreach \i in {1,...,4} {
	\fill[gray!40] (\i-1, -\i*0.5+0.5) rectangle ++(1,-0.5);
	\node[toC] at (\i-0.75, -\i*0.5+0.25) {$x$};
	\node[toC] at (\i-0.25, -\i*0.5+0.25) {$y$};
}

\draw (0,0) rectangle (4,-2);

\foreach \x in {2,...,4} {
	\tikzEvalInt{\maxX}{\x-1}
	\foreach \y in {1,...,\maxX} {
		\node[toC] at (\x-0.75, -\y*0.5+0.25) {$u$};
		\node[toC] at (\x-0.25, -\y*0.5+0.25) {$u$};

		\node[toC] at (4-\x+0.75, -2+\y*0.5-0.25) {$w$};
		\node[toC] at (4-\x+0.25, -2+\y*0.5-0.25) {$w$};
	}
}

\foreach \x/\i in {3/3,4/4,5/3,6/4} {
	\draw[edge] (\x*0.5-0.25, -2.45) -- ++(0, +0.3);
	\node[toC] at (\x*0.5-0.25, -2.75) {$u_\i$};
}
\end{tikzpicture}
\end{center}
Since $u_3$ and $u_4$ are different, at least one of them is different than $uw$. Without loss of generality suppose that $u_3 \neq uw$.
Because each row of the matrix belongs to $\gameV_L$ and $\gameV_L$ is closed under removing letters, it follows that
\[(w,x,u) \in \gameV_L.\]

This means that we have a~violation of Equation~\eqref{eq:cond-bool}, which requires that
\[uw = u^\monoid \cdot u \cdot w^\monoid = u^\monoid \cdot x \cdot w^\monoid = u_3.\]

This concludes the proof of Proposition~\ref{pro:if-bounded-limit}.

\subsection{Case (C2) --- limit alternation is unbounded}%
\label{ssec:limit-unbounded}

In this section we assume that $\Sigma_L$ has unbounded root alternation and every subset $\Sigma'\subseteq\Sigma_L$ with unbounded root alternation has also unbounded limit alternation. Under that assumption we need to prove that Equation~\eqref{eq:cond-limit} is violated.

We start by constructing \new{a~finite directed graph $G_L$ (called the~\emph{strategy graph} of $L$)} that represents very specific strategies of \new{Alternator} in the game on types. We then show that the above assumptions imply that the strategy graph needs to be \emph{recursive}\footnote{\new{This notion, used previously in~\cite{bojanczyk_boolean} and~\cite{michalewski_delta}, has nothing to do with recursion theory, it speaks about the possibility to traverse certain edges in the graph and to go back to the beginning.}} --- roughly speaking it contains complicated connected components. Then we finish the proof of Proposition~\ref{pro:if-unbounded-limit} by showing that recursivity of the strategy graph gives rise to a~violation of Equation~\eqref{eq:cond-limit}.

\subsubsection{Strategy graph}%
\label{ssec:strat-graph}

Given a~language $L$, the \emph{strategy graph of $L$} (denoted $G_L$) is defined as follows. The set of nodes of $G_L$ is $\big(V_L\sqcup\{\one_L\}\big)\times H_L$. There is an edge from a~node $(v,h)$ to a~node $(v',h')$ if there exist:
\[(u_1,u_2), (w_1,w_2)\in\gameV_L,\ z\in V_L\sqcup\{\one_L\}\]
such that
\[h=vu_1w_1^\infty\quad\text{ and }\quad v'=vu_2w_2^\monoid z.\]
Notice that the above definition does not invoke $h'$ (the value of $h'$ matters for a~successive edge from $(v',h')$).

\new{
The following lemma provides a~less explicit definition of edges in $G_L$, via the notion of \emph{path\=/switching}.
}

\begin{defi}%
\label{def:path-switch}
\new{Consider a~tree type $h\in H_L$ and a~pair of context types $v,v'\in V_L$. We say that the pair $(v,h)$ is \emph{path\=/switching} into $v'$ if there exists a~tree $t\in\trees_A$ such that $v\cdot \alpha_L(t)=h$ and there exists an~infinite path $\pi$ of $t$ such that if $D$ is a~finite prefix of $t$ then~$D$ can be completed into a~context $C'$ with the port located on $\pi$ such that $v\cdot \alpha_L(C')\cdot z=v'$ for some $z\in V_L\sqcup\{\one_L\}$.}
\end{defi}

\begin{lem}%
\label{lem:def-of-edges}
There exists an edge from $(v,h)$ to $(v',h')$ in $G_L$ if and only if \new{$(v,h)$ is path\=/switching into $v'$.}
\end{lem}

\begin{proof}
First assume that there exists an edge from $(v,h)$ to $(v',h')$ in $G_L$ with $(u_1,u_2)$, $(w_1,w_2)$, $z$ witnessing that. Let $C_u$, $C_w$, $C_z$ be contexts of types $u_1$, $w_1$, $z$ respectively. Take $t\eqdef C_u\cdot C_w^\infty$. In that case $v\cdot\alpha_L(t)=h$ and let $\pi$ be the infinite path that contains all the ports of the contexts $C_u\cdot C_w^k$ for $k=0,1,\ldots$

Consider a~finite prefix $D$ of $t$. Let $D'$ be an extension of $D$ that is also finite, contains exactly one port, and $D'$ is a~prefix of $C_u\cdot C_w^{(\monoid\cdot k)}$ for some $k$.

By the definition, $D'$ can be written as
\[D=D_0\cdot D_1\cdot\ldots\cdot D_{\monoid\cdot k},\] 
where $D_0,\ldots,D_{\monoid\cdot k}$ are finite prefixes of contexts (i.e.~each of them has exactly one port); and moreover $D_0$ is a~prefix of $C_u$ and all $D_1,\ldots,D_{\monoid\cdot k}$ are prefixes of $C_w$. By the assumption that $(u_1,u_2),(w_1,w_2)\in\gameV_L$ we know that we can extend all $D_0,D_1,\ldots,D_{\monoid\cdot k}$ into contexts $C_0,C_1,\ldots,C_{\monoid\cdot k}$ such that $\alpha_L(C_0)=u_2$ and $\alpha_L(C_i)=w_2$ for $i=1,\ldots,{\monoid\cdot k}$. Take $C'=C_0\cdot C_1\cdot\ldots\cdot C_{\monoid\cdot k}$. Clearly the port of $C'$ is located on $\pi$. Notice that 
\[\alpha_L(C')=u_2\cdot w_2^{\monoid\cdot k}=u_2\cdot w_2^\monoid.\]
Therefore, $v\cdot \alpha_L(C')\cdot z = v'$ by the assumption that there is an edge from $(v,h)$ to $(v',h')$.

Now we assume that $(v,h)$ \new{is path\=/switching into $v'$} and we prove that there is an~edge between $(v,h)$ and $(v',h')$ in $G_L$. This will be achieved by the Ramsey theorem. Consider a~triple of nodes $x\prec y\prec \delta$ on $\pi$. Let $d=|\delta|$ be the depth of $\delta$ and let $D_d$ be the $(d{+}1)$\=/prefix of $t$ (i.e.~$D_d$ is $t$ restricted to the subset ${\{\dL,\dR\}}^{\leq d}$ of its domain). Let $C'_d$ be given from the assumption in Definition~\ref{def:path-switch} for $D=D_d$.

\new{
Let $u_1$, $u_2$, $w_1$, $w_2$ be the $\alpha_L$\=/types of the following context zones:
}
\begin{itemize}
\item $u_1$ (resp. $u_2$) the type of the context zone rooted in $\epsilon$ with the port in $x$ of $t$ (resp.\ of $C'_d$);
\item $w_1$ (resp. $w_2$) the type of the context zone rooted in $x$ with the port in $y$ of $t$ (resp.\ of $C'_d$).
\end{itemize}
Let $z$ be the type given by the assumption for $C'_d$.

\new{
Define $f(x,y,\delta)$ as the quintuple $(u_1,u_2,w_1,w_2,z)\in {\big(V_L\sqcup\{\one_L\}\big)}^5$. Apply the Ramsey theorem for triples to $f$ to get an infinite set $P$ of nodes on $\pi$. We know that for all the triples from~$P$ the function $f$ is constantly equal a~fixed quintuple $(u_1,u_2,w_1,w_2,z)$. It is easy to see that $\alpha_L(t)=u_1\cdot w_1^\infty$. Thus, $h=vu_1w_1^\infty$. Similarly, the Ramsey theorem guarantees that $w_2\cdot w_2=w_2$ and therefore $v'=vu_2w_2^\monoid z$. Thus, to know that there is an~edge from $(v,h)$ to $(v'h,')$ in $G_L$ it is enough to prove that $(u_1,u_2)$ and $(w_1,w_2)$ are both elements of~$\gameV_L$.
}

\new{
We will show that Alternator has a~winning strategy in the game $\gameV(u_1,u_2)$ (the case of $\gameV(w_1,w_2)$ is analogous). Fix $x\prec y \in P$ and let $C_1$ be the context zone rooted in $\epsilon$ with the port in $x$ of $t$. Assume that Alternator plays $C_1$ as the first context in the game $\gameV(u_1,u_2)$. Consider a~finite prefix $D$ that is played by Constrainer. Since $D$ is finite and $P$ is infinite, there exists $\delta\in P$ such that $d\eqdef |\delta|$ is greater than the depth of all the nodes in $D_1$. Let Alternator reply with $C'_d$. By the assumption on $d$, we know that $C'_d$ indeed extends $D$. By the choice of $P$ we know that $\alpha_L(C'_d)=u_2$. Thus, we have proved that Alternator wins $\gameV(u_1,u_2)$.
}
\end{proof}

\begin{lem}%
\label{lem:strat-graph-transitive}
The strategy graph $G_L$ is transitive.
\end{lem}

\begin{proof}
Consider an edge between $(v,h)$ and $(v',h')$ witnessed by $(u_1,u_2)$, $(w_1,w_2)$, and $z$; and an edge between $(v',h')$ and $(v'',h'')$ witnessed by $(u_1',u_2')$, $(w_1',w_2')$, $z'$. Then $h=vu_1w_1^\infty$ and $v''=vu_2w_2^\monoid\big(zu_2'{(w_2')}^\monoid z'\big)$. It means that there is an edge between $(v,h)$ and $(v'',h'')$ with $(u_1,u_2),(w_1,w_2)\in\gameV_L$, and $z''\eqdef zu_2'{(w_2')}^\monoid z'$.
\end{proof}

\new{Recall that a~strongly connected component of a~directed graph $G$ is a~set of nodes~$X$ such that for each $x\neq x'\in X$ there exists a~path in $G$ leading from $x$ to $x'$. If $G$ is transitive then this condition boils down to saying that there is a~single edge from $x$ to $x'$.}

\begin{defi}[\emph{Recursive strategy graphs}]
We say that a~strategy graph $G_L$ is \emph{recursive} if there exists a~strongly connected component of $G_L$ that contains two nodes $(v,h)$, $(v',h')$ with $h\neq h'$.
\end{defi}

\begin{lem}%
\label{lem:rec-to-violation}
If the strategy graph $G_L$ is recursive then Equation~\eqref{eq:cond-limit} is violated.
\end{lem}

\begin{proof}
Assume contrarily that $G_L$ is recursive but Equation~\eqref{eq:cond-limit} holds. Consider a~pair of nodes $(v,h)$ and $(v',h')$ with $h\neq h'$ that witness recursivity of $G_L$. By Lemma~\ref{lem:strat-graph-transitive} there must exist edges in $G_L$ from $(v,h)$ to $(v',h')$ and back. Let $(u_1,u_2)$, $(w_1,w_2)$, $z$; and $(u_1',u_2')$, $(w_1',w_2')$, $z'$ witness the existence of these edges respectively.

Apply Equation~\eqref{eq:cond-limit} to obtain that:
\begin{align*}
{\big(u_2{(w_2)}^\monoid\ z u_2' {(w_2')}^\monoid z'\big)}^\monoid\ u_1{(w_1)}^\infty&={\big(u_2{(w_2)}^\monoid\ z u_2' {(w_2')}^\monoid z'\big)}^\infty\\
{\big(u_2'{(w_2')}^\monoid\ z' u_2 {(w_2)}^\monoid z\big)}^\monoid\ u_1'{(w_1')}^\infty&={\big(u_2'{(w_2')}^\monoid\ z' u_2 {(w_2)}^\monoid z\big)}^\infty
\end{align*}
Let $W = u_2{(w_2)}^\monoid z$ and $W'=u_2' {(w_2')}^\monoid z'$. Then by the assumptions on the edges between $(v,h)$ and $(v',h')$ we get that $v W= v'$ and $v' W' = v$. Moreover, the above equations get the form
\begin{align*}
{\big(WW'\big)}^\monoid u_1{(w_1)}^\infty&={\big(WW'\big)}^\infty\\
{\big(W'W\big)}^\monoid u_2{(w_2)}^\infty&={\big(W'W\big)}^\infty
\end{align*}
And therefore, by using the values of $h$, $h'$ and multiplying these equations by $v$ and $v'$ respectively we get:
\begin{align*}
h=v\cdot u_1{(w_1)}^\infty=v\cdot{\big(W W'\big)}^\monoid\ u_1{(w_1)}^\infty&=v\cdot{\big(W W'\big)}^\infty\\
h'=v'\cdot u_1'{(w_1')}^\infty=v'\cdot{\big(W'W\big)}^\monoid\ u_1'{(w_1')}^\infty&=v'\cdot{\big(W'W\big)}^\infty
\end{align*}
Now since $h=v\cdot {\big(WW'\big)}^\infty=v W\cdot {\big(W'W\big)}^\infty=v'\cdot {\big(W'W\big)}^\infty=h'$ we obtain the contradiction as we assumed that $h\neq h'$.
\end{proof}

\subsubsection{Constructing a~path in \texorpdfstring{$G_L$}{GL}}%
\label{sssec:cons-path}

We will now use the assumption that Case~(C2) holds to construct an infinite path in $G_L$ such that every two consecutive nodes on that path $(v,h)$, $(v',h')$ satisfy $h\neq h'$. Since $G_L$ is finite, such an infinite path witnesses that some strongly connected component of $G_L$ contains such two nodes and therefore $G_L$ is recursive. By Lemma~\ref{lem:rec-to-violation} it means a~violation of Equation~\eqref{eq:cond-limit} and the proof of Proposition~\ref{pro:if-unbounded-limit} is finished.

The construction of the path will be inductive, preserving an invariant that the last node constructed on the path is \emph{alternating}: a~node $(v,h)$ in $G_L$ is \emph{alternating} if $v\cdot \gameH_L$ contains words that begin with $h$ and have arbitrarily high alternation.

\begin{lem}%
\label{lem:one-alternating}
$G_L$ contains at least one alternating node.
\end{lem}

\begin{proof}
This is because $\gameH_L$ has unbounded alternation. Therefore,
by Pigeon-hole Principle there exists some $h \in H_L$ such that
$\gameH_L$ contains words that begin with $h$ that have
arbitrarily high alternation. By the definition, this means that the node
$(\one_L,h)$ is alternating in $G_L$.
\end{proof}

The inductive step of the construction will be based on the following lemma.

\begin{lem}%
\label{lem:step-alternating}
If $(v,h)$ is an alternating node of $G_L$ and Case~(C2) holds then there exists an edge in $G_L$ from $(v,h)$ to $(v',h')$ such that $h\neq h'$ and $(v',h')$ is also alternating.
\end{lem}

The rest of the section is devoted to a~proof of Lemma~\ref{lem:step-alternating}. Let $\Sigma_{(v,h)}$ be the set of all strategy trees $\sigma$ (with root sequences of the form $(h_1,\ldots,h_n)$) such that:
\begin{itemize}
\item $\sigma$ is locally optimal for $v$ (see Definition~\ref{def_optimalstrategies}),
\item $v\cdot h_1=h$,
\item the sequence $(vh_1,\ldots, vh_n)$ is alternating (i.e.~every two consecutive values in the sequence are distinct, see Definition~\ref{def:alternation}).
\end{itemize}

\begin{fact}%
\label{ft:unbounded-root-for-v-h}
The set $\Sigma_{(v,h)}$ is a~subset of $\Sigma_L$ that has unbounded both root \new{and limit alternation}.
\end{fact}

\begin{proof}
Fact~\ref{ft:locally-premultiply} implies that all strategy trees in $\Sigma_{(v,h)}$ belong also to $\Sigma_L$. By the fact that $(v,h)$ is alternating and by Lemma~\ref{lem:locally-optimal} we know that $\Sigma_{(v,h)}$ has unbounded root alternation. Therefore, Case~(C2) guarantees that $\Sigma_{(v,h)}$ has also unbounded limit alternation.
\end{proof}

Now observe that for each $g,g'\in H_L$ either ${(g,g')}^k\in\gameH_L$ for all $k\in\w$, or there exists a~unique $k_{g,g'}\in\w$ such that ${(g,g')}^{k_{g,g'}}\in\gameH_L$ but ${(g,g')}^{1{+}k_{g,g'}}\notin\gameH_L$. Since the set $H_L$ is finite, there exists a~number $K$ that is greater than all the defined numbers $k_{g,g'}$. Thus, the following fact holds.

\begin{fact}%
\label{ft:great-K}
If for some pair $g,g'\in H_L$ we have ${(g,g')}^K\in\gameH_L$ then ${(g,g')}^k\in\gameH_L$ for all $k\in\w$.
\end{fact}

Let $\ell\eqdef |H_L|^2\cdot K$ and let $\sigma=(t,\sigma_1,\ldots,\sigma_n)$ be a~strategy tree in $\Sigma_{(v,h)}$ that has limit alternation greater than $\ell$.

We will now construct a~path $\pi$ in $t$ on which the high limit alternation of $\sigma$ is located. Let $Z$ be the set of nodes $z$ of $t$ such that the root alternation of $\sigma.z$ is greater than $\ell$. By Corollary~\ref{cor:loc-opt-to-alt-prefix} we know that $Z$ is prefix closed and by the fact that the limit alternation of $\sigma$ is greater than $\ell$ we know that $Z$ is infinite. Therefore, by K\"onig lemma we know that $Z$ contains an infinite path~$\pi$.

\begin{lem}%
\label{lem:pigeonhole-on-pi}
There exists an infinite set \new{$P\subseteq \pi$} and two sequences $(h_1,\ldots,h_n)\in H_L^\ast$ and $(v_1,\ldots,v_n)\in V_L^\ast$ such that for all nodes \new{$x\in P$}:
\begin{itemize}
\item $\big(\sigma_1(x),\ldots,\sigma_n(x)\big)=(h_1,\ldots,h_n)$,
\item $\big(\val(\sigma,X,1),\ldots,\val(\sigma,X,n)\big)=(v_1,\ldots,v_n)$, where $X$ is the context zone with the root in $\epsilon$ and port in $x$.
\end{itemize}
Moreover, $(h_1,\ldots,h_n)\in\gameH_L$ and $(v_1,\ldots,v_n)\in\gameV_L$.
\end{lem}

\begin{proof}
The choice of \new{$P$} and the sequences $(h_1,\ldots,h_n)\in H_L^\ast$, $(v_1,\ldots,v_n)\in V_L^\ast$ is just by Pigeon\=/hole Principle. By the definition of a~strategy tree we know that $(h_1,\ldots,h_n)\in\gameH_L$. By Fact~\ref{fact:zone-strat} we know that also $(v_1,\ldots,v_n)\in\gameV_L$.
\end{proof}

Since the alternation of $(h_1,\ldots,h_n)$ is greater than $\ell=|H_L|^2\cdot K$, we know that there exists a~pair $g\neq g'\in H_L$ and a~set of indices $I\subseteq\{1,\ldots,n-1\}$ such that $|I|\geq K$ and for all $i\in I$ we have $h_i=g$ and $h_{i+1}=g'$. Fix as $i_0$ the minimal element of $I$. By closure of $\gameH_L$ under subwords, we know that ${(g,g')}^K\in\gameH_L$. By Fact~\ref{ft:great-K} it implies that
\begin{align}
&{(g,g')}^k\in\gameH_L&\text{for all $k\in\w$.}
\label{eq:unbounded-g-alt}
\end{align}

Let $u=v_{i_0}$ and $u'=v_{{i_0}{+}1}$ where $i_0$ is the minimal element of $I$. We will finish the proof by showing the following three lemmas.

\begin{lem}%
\label{lem:both-distinct}
The values $vu'g$ and $vu'g'$ are distinct.
\end{lem}

\begin{proof}
Recall that $g=h_{i_0}$ and $g'=h_{i_0+1}$. Assume that the value $vu' g$ is equal to $vu'g'$. Consider a~strategy tree $\sigma'$ that equals $\sigma$ except for the subtree $\sigma'_{i_0+1}.x$ where $x$ is the ${\preceq}$\=/minimal element of $X$.  Over that subtree, let $\sigma'_{i_0+1}$ be equal to $\sigma_{i_0}$. Let $(h_1',\ldots,h_n')$ be the root sequence of $\sigma'$. We know that $v\cdot(h_1',\ldots,h_n')=(h_1,\ldots,h_n)$ and $\dist(\sigma_{i_0}, \sigma'_{i_0+1})$ is strictly smaller than $\dist(\sigma_{i_0}, \sigma_{i_0+1})$, what contradicts local minimality of $\sigma$ for $v$.
\end{proof}

\begin{lem}%
\label{lem:both-alternating}
The nodes $(vu',vu'g)$ and $(vu',vu'g')$ are both alternating in~$G_L$.
\end{lem}

\begin{proof}
By~\eqref{eq:unbounded-g-alt} and Fact~\ref{lem:compose-winh-with-multicontext} we know that ${(vu'g,vu'g')}^k\in\gameH_L$ for all $k\in\w$. Moreover, Lemma~\ref{lem:both-distinct} says that $vu'g\neq vu'g'$.
\end{proof}

Now it remains to prove the following lemma.

\begin{lem}%
\label{lem:construct-an-edge}
There exist edges in $G_L$ from $(v,h)$ to both $(vu',vu'g)$ and $(vu',vu'g')$.
\end{lem}

\begin{proof}
\new{Using Lemma~\ref{lem:def-of-edges} it is enough to show that $(v,h)$ is path\=/switching into $vu'$.} Let $\sigma'$ be the strategy tree obtained from $\sigma$ by removing the first $i_0$ rounds: $\sigma'=(t,\sigma_{i_0+1},\sigma_{i_0+2},\ldots,\sigma_n)$ ($i_0$ is the minimal element of~$I$).

Consider the tree $t$ from the strategy tree $\sigma$ and the path~$\pi$. By the definition of $\sigma$ we know that $v\cdot\alpha_L(t)=h$. Let $D$ be a~finite prefix of $t$. The strategy tree $\sigma'$ witnesses that $D$ can be extended into a~context $C'$ with a~port located on $\pi$ such that that $\alpha_L(C')=v_{i_0+1}=u'$. Therefore $v\cdot\alpha_L(C')\cdot 1_{V_L}=vu'$. This concludes the proof that $(v,h)$ \new{is path\=/switching into~$vu'$}.
\end{proof}

Thus, Lemma~\ref{lem:both-distinct} implies that at least one of the nodes $(vu',vu'g)$ and $(vu',vu'g')$ has the second coordinate distinct than $h$, Lemma~\ref{lem:both-alternating} implies that this node is alternating, and Lemma~\ref{lem:construct-an-edge} implies that $G_L$ contains an~edge from $(v,h)$ to that node. Thus, the proof of Lemma~\ref{lem:step-alternating} is concluded.

\section{Effective characterisation of \texorpdfstring{$\bdelta{2}$}{Delta2}}%
\label{sec:delta-case}

This final section is devoted to the effective characterisation of the Borel class $\bdelta{2}$ (that corresponds to the union of the first $\omega_1$ Wadge degrees). Decidability of the class $\bdelta{2}$ can be actually obtained as a~direct corollary of~\cite{cavallari_gdelta}: in this work the authors proved that it is decidable if a~regular tree language is in the Borel class $\bpi{2}$; since regular languages are closed under complement and $\bdelta{2} =\bpi{2} \cap \bsigma{2}$, we automatically obtain that it is decidable if a~regular tree language is in $\bdelta{2}$. However, here we show that the algebraic approach developed in the previous sections covers the case of the Borel class $\bdelta{2}$ as well. What we will prove is the following:

\begin{thm}%
\label{thm:delta-char}
A~regular tree language $L$ belongs to the class $\bdelta{2}$ if and only if its syntactic algebra satisfies Equation~\eqref{eq:cond-limit} from Theorem~\ref{thm:mainforBC1}.
\end{thm}

This theorem is the main result of~\cite{michalewski_delta}; unfortunately the proof provided there is wrong:
\begin{itemize}
\item Proposition~4 in~\cite{michalewski_delta} cites Lemma~G.2 from~\cite{bojanczyk_boolean} in a~wrong way. The logical claim in Proposition~F.2 is of the form: if there is a~set $\Sigma$ of unbounded root alternation then there is a~set $\Sigma'$ of unbounded limit alternation\footnote{Unbounded limit alternation implies unbounded root alternation.}. Lemma~G.2 says that if Proposition~F.2 is violated then the strategy graph is recursive. Logically, it takes the form: if there is a~set $\Sigma$ of unbounded root alternation but no set $\Sigma'$ has unbounded limit alternation then the graph is recursive. The way Proposition~4 summarises that statement is: if there exists a~set $\Sigma$ of unbounded limit alternation then the strategy graph is recursive. This statement does not follow from~\cite{bojanczyk_boolean}.
\item The proof of Theorem~1 in~\cite{michalewski_delta} in the implication $(2)\Rightarrow (1)$, shows how to construct, given an infinite strategy tree $s^\infty$, a~family of strategy trees of unbounded limit alternation. The first step of the proof is to construct a~family $\Gg$ of strategy trees of finite duration but unbounded root alternation. Then, an invalid application of a~compactness argument (to find the set $X$) shows that $\Gg$ has in fact unbounded limit alternation. Such a~set $X$ doesn't need to exist, it can be the case that the limit alternation of $s$ is $2$ but for each pair $k<k'\leq j$ there are infinitely many nodes $n$ in $s$ such that $\sigma_k(n)\neq\sigma_{k'}(n)$ --- it is enough that the nodes for distinct values $k,k'$ lie on distinct infinite branches.

Additionally, if it was the case that every set of strategy trees of unbounded root alternation has unbounded limit alternation; Proposition~F.1 would hold always, nevertheless of the assumption of Identity~(2).
\end{itemize}

\noindent
The proof we provide here follows the logical structure of~\cite{michalewski_delta}, with the following differences:
\begin{itemize}
\item Instead of the wrong reference used in Proposition~4 of~\cite{michalewski_delta} we show recursivity of $G_L$ using a~new Lemma~\ref{lem:def-of-edges} that characterises existence of edges in $G_L$.
\item Instead of the statement about existence of the set $X$ from the proof of Theorem~1 in~\cite{michalewski_delta} we provide here a~direct construction (Lemma~\ref{lem:infinitely-essential}) showing that a~winning strategy of Alternator in $\gameHinfty(L)$ must take a~specific structure that fits the characterisation from Lemma~\ref{lem:def-of-edges}
\end{itemize}

As observed in~\cite{michalewski_delta}, the techniques used to characterise $\boolcomb{1}$ provide a~big part of a~proof of Theorem~\ref{thm:delta-char}. First, Proposition~\ref{pro:inf-game-char} says that $L$ is $\bdelta{2}$ if and only if Constrainer wins $\gameHinfty(L)$. Corollary~\ref{cor:eq-limit-to-game} says that if Equation~\eqref{eq:cond-limit} is violated then Alternator wins $\gameHinfty(L)$. Thus, the ``only if'' part of Theorem~\ref{thm:delta-char} follows.

On the other hand, Lemma~\ref{lem:rec-to-violation} says that if the strategy graph $G_L$ is recursive then Equation~\eqref{eq:cond-limit} is violated. It means that the only remaining statement to prove Theorem~\ref{thm:delta-char} is expressed by the following proposition.

\begin{prop}%
\label{pro:not-delta-to-rec}
If Alternator wins $\gameHinfty(L)$ then the strategy graph $G_L$ is recursive.
\end{prop}

The rest of this section is devoted to a~proof of the above proposition. During the proof we will inductively construct an infinite path ${\big(v_n,h_n\big)}_{n\in\w}$ in the strategy graph such that the sequence $h_n$ alternates between $L$ and $L\complemento$. The invariant of our construction is expressed by the following definition, using the notions of quotients from Subsection~\ref{ssec:quotients}.

\begin{defi}
Consider a~node $(v,h)$ of the strategy graph $G_L$. We say that $(v,h)$ is \emph{prolongable} if there exists a~winning strategy $\sigma$ of Alternator in $\gameHinfty\big(v^{-1}(L)\big)$ or in $\gameHinfty\big(v^{-1}(L)\complemento\big)$ such that if $t$ is the tree played by $\sigma$ in the first round then $v\cdot \alpha_L(t)=h$.
\end{defi}

\new{
We begin with an introduction of an~object called \emph{essential node} and a~study of its properties.
}

\subsection{Essential nodes}%
\label{ssec:essential}

\new{
Let $(v,h)$ be a~prolongable node, as witnessed by a~strategy $\sigma$ and a~tree $t$. By the symmetry we can assume that $\sigma$ is winning in $\gameHinfty\big(v^{-1}(L)\complemento\big)$.
}

\begin{defi}%
\label{def:essential}
\new{Consider a~node $u\in{\{\dL,\dR\}}^d$ for $d>0$. Let $p$ be the $d$\=/prefix of $t$ (i.e.~$p\eqdef t\restr{\{\dL,\dR\}}^{< d}$). We say that $u$ is \emph{essential} if there exists a~context $C$ of a~type $w\in V_L$, such that the port of $C$ is in $u$, $C$ extends $p$, and we have ${\big(vw\big)}^{-1}(L)\notin\bdelta{2}$.
}

The last condition implies that Alternator has a~winning strategy in\footnote{We have assumed that $\sigma$ is winning in $\gameHinfty\big(v^{-1}(L)\complemento\big)$ and therefore we have no complement here; in the dual case $\sigma$ is winning in $\gameHinfty\big(v^{-1}(L)\big)$ and here we consider a~winning strategy in $\gameHinfty\big({\big(vw\big)}^{-1}(L)\complemento\big)$.} $\gameHinfty\big({\big(vw\big)}^{-1}(L)\big)$. Let $t'$ be a~tree played by one of such winning strategies of Alternator and let $g=\alpha_L(t')$. Using the above notions we say that $u$ is \emph{$(w,g)$\=/essential}.
\end{defi}

\begin{fact}%
\label{ft:essential-prefix}
If $C$ is a~context then the function $t\mapsto C[t]$ is continuous. Therefore, if $w^{-1}(K)\notin\bdelta{2}$ then also $K\notin\bdelta{2}$. In particular, using~\eqref{eq:quot-comp} we obtain that for $v,w\in V_L\sqcup\{\one_L\}$ if ${(vw)}^{-1}(L)\notin\bdelta{2}$ then also $v^{-1}(L)\notin\bdelta{2}$. It means that if $\epsilon\prec u\preceq u'$ and $u'$ is essential then also $u$ is essential.
\end{fact}

\new{
The crucial part in the proof of Proposition~\ref{pro:not-delta-to-rec} is expressed by the following lemma.
}

\begin{lem}%
\label{lem:infinitely-essential}
If $(v,h)$ be a~prolongable node, witnessed by a~strategy $\sigma$ and a~tree $t$, then for every $d> 0$ there exists an~essential node $u$ in $t$ such that $|u|=d$.
\end{lem}

The rest of this subsection is devoted to a~proof of this lemma. Consider a~number $d>0$ and let $p$ be the $d$\=/prefix of $t$. We will find an essential node $u\in{\{\dL,\dR\}}^d$. Let $D$ be the multicontext obtained from $p$ by making all the nodes in ${\{\dL,\dR\}}^d$ ports; let $u_1,\ldots,u_N$ be the list of these ports ($N=2^d$); and $U\eqdef D[\ast]$ be the basic open set defined by $p$. Finally, put
\[M\eqdef D^{-1}\big(v^{-1}(L)\big)\subseteq \trees_A^N.\]

\begin{rem}%
\label{rem:M-finitely-gen}
The fact whether $(t_1,\ldots,t_N)\in M$ depends only on the tuple of types $\big(\alpha_L(t_1),\ldots,\alpha_L(t_N)\big)$.
\end{rem}

Since $(v,h)$ is prolongable and $U$ is a~valid response of Constrainer to Alternator playing $t$ in $\gameHinfty\big(v^{-1}(L)\complemento\big)$, we know that Alternator has a~winning strategy in $\gameHinfty_U\big(v^{-1}(L)\big)$. This game is equivalent to the game $\gameHinfty\big(M\big)$ played in the topological space $\trees_A^N$ --- the $N$\=/fold product of the space $\trees_A$, with the induced product topology. Therefore, by Proposition~\ref{pro:inf-game-char} we know that $M\notin\bdelta{2}$.

\begin{defi}
A~\emph{section} of $M\subseteq\trees_A^N$ is any set of the form
\[\big\{t\in\trees_A\mid (t_1,\ldots,t_{i-1},t,t_{i+1},\ldots,t_N)\in M\big\}.\]
for $i\in\{1,\ldots,N\}$.
\end{defi}

Due to Remark~\ref{rem:M-finitely-gen}, taking $h_j=\alpha_L(t_j)$ for $j=1,\ldots,N$, $j\neq i$, we can equivalently define the above section as
\begin{align}
&\big(h_1,\ldots,h_{i-1},\hole,h_{i+1},\ldots,h_N\big)\eqdef \nonumber\\
&\qquad\big\{t\in\trees_A\mid h_1\times \cdots \times h _{i-1}\times \{t\}\times h_{i+1}\times \cdots\times h_N\subseteq M\big\}.
\label{eq:M-section}
\end{align}

Notice that the notation from~\eqref{eq:M-section} naturally extends to formulae with more than one~$\hole$. Such sets are called \emph{multisections}. The number of holes in a~multisection is called its \emph{dimension}, i.e.~a~multisection of dimension $n$ is a~subset $\trees_A^n$. There is one multisection of the maximal dimension $N$: $\big(\hole,\ldots,\hole\big)=M$.

\begin{lem}%
\label{lem:multisections}
Each multisection can be obtained from sections by taking finite unions, finite intersections, and products.
\end{lem}

\begin{proof}
The proof will be inductive on the dimension of the given multisection. For the sake of simplicity of notation we even consider multisections of dimension zero, i.e.~expressions of the form $(h_1,\ldots,h_N)$ with no holes. Formally such an expression is either the empty set, or a~set containing the empty tuple $()$ (depending on whether $h_1\times\cdots\times h_N\subseteq M$ or not).

A~multisection of dimension $1$ is just a~section, so the claim follows. Consider a~multisection of the form (this form is generic up to rearranging the coordinates):
\[N=\big(\hole,\hole,\ldots,\hole, h_{i+1},h_{i+2},\ldots,h_N\big).\]
We will prove that $N$ can be represented as in Lemma~\ref{lem:multisections}.

Consider a~tuple $h_1,\ldots,h_{i-1}$. If $\big(h_1,\ldots,h_{i-1},\hole,h_{i+1},\ldots,h_N\big)$ is empty then let us put $S(h_1,\ldots,h_{i-1})=\emptyset$ and otherwise let $S(h_1,\ldots,h_{i-1})$ be
\[\bigcap_{h_i\subseteq\big(h_1,\ldots,h_{i-1},\hole,h_{i+1},\ldots,h_N\big)} \big(\hole,\ldots,\hole,h_i,\ldots,h_N\big)\times \big(h_1,\ldots,h_{i-1},\hole,h_{i+1},\ldots,h_N\big).\]

Now consider the set $K$ defined as
\[\bigcup_{h_1,\ldots,h_{i-1}} S(h_1,\ldots,h_{i-1}).\]
By the inductive assumption the set $K$ can be obtained from sections by finite unions, finite intersections, and products. It remains to prove that $K=N$.

First consider a~tuple of types $(g_1,g_2,\ldots,g_i)$ such that
\begin{equation}
\label{eq:contained-in-N}
g_1\times\cdots\times g_i\subseteq N.
\end{equation}
We will show that this product is contained in $S(g_1,\ldots,g_{i-1})$. First, by~\eqref{eq:contained-in-N} we know that $g_i\subseteq \big(g_1,\ldots,g_{i-1},\hole,h_{i+1},\ldots,h_N\big)$ and therefore the later set is not empty. Consider any value $h_i\subseteq \big(g_1,\ldots,g_{i-1},\hole,h_{i+1},\ldots,h_N\big)$. By the assumption about the considered values $h_i$ we know that $g_1\times\cdots\times g_{i-1}\subseteq \big(\hole,\ldots,\hole,h_i,\ldots,h_N\big)$. Therefore, $g_1\times\cdots\times g_i$ is contained in the product $\big(\hole,\ldots,\hole,h_i,\ldots,h_N\big)\times \big(h_1,\ldots,h_{i-1},\hole,h_{i+1},\ldots,h_N\big)$.

Now consider the other implication: take a~tuple of types $(g_1,g_2,\ldots,g_i)$ such that
\begin{equation}
\label{eq:contained-in-Nprime}
g_1\times\cdots\times g_i\subseteq S(h_1,\ldots,h_{i-1}),
\end{equation}
for some choice of types $h_1,\ldots,h_{i-1}$. It means that there must exist a~type $h_i$ contained in $\big(h_1,\ldots,h_{i-1},\hole,h_{i+1},\ldots,h_N\big)$, because the later set cannot be empty. Therefore, $g_i\subseteq \big(h_1,\ldots,h_{i-1},\hole,h_{i+1},\ldots,h_N\big)$ and it means that $g_i$ is among the values $h_i$ over which we take the intersection. Thus
\[g_1\times\cdots \times g_i\subseteq \big(\hole,\ldots,\hole,g_i,\ldots,h_N\big)\times \big(h_1,\ldots,h_{i-1},\hole,h_{i+1},\ldots,h_N\big).\]
In particular $g_1\times\cdots\times g_{i-1}\subseteq \big(\hole,\ldots,\hole,g_i,\ldots,h_N\big)$ what implies that
\[g_1\times\cdots\times g_i\subseteq N.\]

This concludes the proof that $N=K$ and thus the claim is proved.
\end{proof}

\begin{cor}
There exists a~section of $M$ that is not $\bdelta{2}$.
\end{cor}

\begin{proof}
\new{Assume that all the sections of $M$ are $\bdelta{2}$. Since $\bdelta{2}$ sets are closed under finite unions, finite intersections, and products, Lemma~\ref{lem:multisections} implies that the multisection $(\hole,\ldots,\hole)=M$ is also $\bdelta{2}$. A~contradiction.}
\end{proof}

Let $\big\{t\in\trees_A\mid (t_1,\ldots,t_{i-1},t,t_{i+1},\ldots,t_N)\in M\big\}$ be a~section of $M$ that is not $\bdelta{2}$. Put
\[C\eqdef D[t_1,\ldots,t_{i-1},\hole,t_{i+1},\ldots,t_N].\]
By the assumption about the chosen section, we have $C^{-1}\big(v^{-1}(L)\big)\notin\bdelta{2}$. Thus, by putting $w\eqdef\alpha_L(C)$ we know that ${\big(vw\big)}^{-1}(L)\notin\bdelta{2}$ and therefore the context $C$ witnesses that $u_i$ is essential. This concludes the proof of Lemma~\ref{lem:infinitely-essential}.

\subsection{Constructing a~path in \texorpdfstring{$G_L$}{GL}}%
\label{ssec:cons-path-in-gl}

Lets fix a~strategy $\sigma$ of Alternator in $\gameHinfty(L)$ and let $t_0$ be the tree played by $\sigma$ in the first round. Let $v_0=\one_L$ and $h_0=\alpha_L(t_0)$. Directly from the definition we know that $(v_0,h_0)$ is prolongable. Therefore, to prove Proposition~\ref{pro:not-delta-to-rec} it is enough to prove the following inductive lemma.

\begin{lem}%
\label{lem:prolongable-induction}
If $(v,h)$ is prolongable then there exists a~node $(v',h')$ in the strategy graph $G_L$ such that $h'\neq h$, $(v',h')$ is prolongable, and there is an edge from $(v,h)$ to $(v',h')$.
\end{lem}

Lets fix a~node $(v,h)$ that is prolongable, as witnessed by a~strategy $\sigma$ and a~tree $t$. By the symmetry we can assume that $\sigma$ is winning in $\gameHinfty\big(v^{-1}(L)\complemento\big)$.

Since there are infinitely many essential nodes in $t$, and a~prefix of an essential node is also essential (see Fact~\ref{ft:essential-prefix}), there exists a~path $\pi$ such that for all $\epsilon\prec u\prec \pi$ the node $u$ is essential. Notice that each of these nodes $u$ comes with at least one pair $(w_u,g_u)\in V_L\times H_L$ such that $u$ is $(w_u,g_u)$\=/essential, see Definition~\ref{def:essential}. Let $(w,g)$ be a~pair that equals $(w_u,g_u)$ for infinitely many $u$. Let $v'=vw$ and $h'=vw g$. By the choice of $h$ and $h'$ we know that $h\subseteq L\complemento$ and $h'\subseteq L$ and therefore $h'\neq h$. Also, directly from the definition of a~$(w,g)$\=/essential node we know that the node $(v',h')$ is prolongable.

Therefore, to conclude Lemma~\ref{lem:prolongable-induction} it is enough check that $(v,h)$ is path\=/switching into $v'$ (see Lemma~\ref{lem:def-of-edges}). Consider a~finite prefix $D$ of $t$. There must exist an essential node $u\prec \pi$ such that $u\in{\{\dL,\dR\}}^d$ for some $d>0$; $(w_u,g_u)=(w,g)$; and $D$ is a~prefix of the $d$\=/prefix of $t$. Since $u$ is essential, $D$ can be completed into a~context $C$ with the port located in $u$, such that $\alpha_L(C)=w$. Thus, $v\cdot\alpha_L(C)\cdot \one_L=v\cdot w=v'$. This fulfils the requirements of Definition~\ref{def:path-switch} and therefore $G_L$ contains an edge from $(v,h)$ to $(v',h')$.

\section{Conclusions}%
\label{sec:conclusions}

This paper utilises the syntactic algebra of a~given regular language $L$ to understand the topological complexity of $L$. The main result (Theorem~\ref{thm:mainforBC1}) says that the language is $\boolcomb{1}$ if and only if its syntactic algebra satisfies Identities~\eqref{eq:cond-bool} and~\eqref{eq:cond-limit}.

The first equation speaks about the \emph{branching} structure of alternations between the language and the complement. The combinatorial background of this equation is represented by the three claims from Subsection~\ref{ssec:constructing-matrices}, which show that one can alternate using distinct infinite branches of a~given tree.

The second equation is focused on alternations that continue along a~single branch, see Lemma~\ref{lem:def-of-edges} for the explicit form of such an alternation.

As observed by Facchini and Michalewski~\cite{michalewski_delta}, languages outside the class $\bdelta{2}$ should admit the second kind of alternation. Although the original paper contained certain mistakes, the claim is correct (Theorem~\ref{thm:delta-char}): a~regular language is $\bdelta{2}$ if and only if its syntactic algebra satisfies Identity~\eqref{eq:cond-limit}.

\subsection{Limitations}%
\label{ssec:limitations}

In general there is no proper algebraic framework for analysing all regular languages of infinite trees. The available algebras are either too weak~\cite{bojanczyk_ef_trees,bilkowski_unambiguity}, too strong~\cite{bojanczyk_tree_algebra}, or not closed under homomorphic images (i.e.~contain a~hidden existential quantifier)~\cite{blumensath_recognisability}. Therefore, effective characterisations based on algebraic approach are limited either to certain subclasses of languages as the input; or to simple classes that are being characterised.

In this paper the input is not restricted (i.e.~the characterisation works for all regular languages as the input); but the characterised classes are low in the Borel hierarchy. It is not surprising, as the structure of the considered algebras is quite weak, as expressed by the following remark.

\begin{rem}
For every non\=/trivial language $L_0\subseteq\trees_A$ the syntactic algebra of the language
\[\{t\in\trees_A\mid \forall u\ \exists v\succeq u.\ t\restr v \in L_0\}\]
is the same, with both $H$ and $V$ being two\=/element sets.
\end{rem}

Notice that if $L_0=\{t\in\trees_{\{a,b\}}\mid t(\epsilon)=a\}$ then the above language belongs to $\bpi{2}$, while for $L_0=\big\{t\in\trees_{\{a,b,c\}}\mid \text{if $t(\epsilon)=c$ then $t$ contains a~branch with infinitely many $a$}\big\}$ the above language is non\=/Borel ($\asigma{1}$\=/complete). This shows that the structure of the considered algebras is not strong enough to distinguish the complexity of languages right above the class~$\bdelta{2}$. To climb higher in the Borel hierarchy one should consider some class of algebras with richer structure; unfortunately there is no natural candidate for such a~class at the moment. The known characterisations of classes higher in the hierarchy\footnote{There are other characterisations where the input is restricted to languages with limited use of non\=/determinism, see~\cite{murlak_wadge,murlak_final_game}.}~\cite{walukiewicz_buchi,cavallari_gdelta} are based directly on games instead of using algebras.

\subsection{Comparison with original papers}

Comparing to the original work~\cite{bojanczyk_boolean}, this work provides more detailed and polished proofs, more pictures, and certain minor glitches corrected. Moreover, the newly added Lemma~\ref{lem:def-of-edges} provides an explicit characterisation of the edges in the strategy graph $G_L$.

As discussed in Section~\ref{sec:delta-case}, the logical structure of~\cite{michalewski_delta} is not sound. In this work we repair the gap by providing a~direct proof of existence of certain edges in the strategy graph $G_L$. For that, we introduce a~new concept of an \emph{essential node} (see Definition~\ref{def:essential}) and a~combinatorial lemma about sections of matrices, see Lemma~\ref{lem:multisections}.

\subsection{Further work}

There are still natural classes of languages within $\bdelta{2}$ that await characterisations.

\paragraph*{\bf Equations for finite levels.} In this work we remark that the finite levels of the difference hierarchy can be characterised using sentences of \mso. It remains open whether these finite levels correspond to equations (or their ordered variants) in the syntactic algebra of the language.

\paragraph*{\bf Bounds for infinite levels.} In~\cite{duparc_murlak_operations} it has been proved that if $L_1$ and $L_2$ are regular tree languages with Wadge rank $\alpha_1$ and $\alpha_2$ respectively, then we can build regular languages $L_1 \oplus L_2$ and $L_1 \otimes \omega$ with Wadge rank $\alpha_1 + \alpha_2$ and $\alpha_1 \times \omega$ respectively. Hence, every Wadge degree with Wadge rank below $\omega^\omega$ is inhabited by regular languages. In particular, there are examples of regular tree languages in $\bdelta{2}\setminus\boolcomb{1}$. However, as the Wadge hierarchy over $\bsigma{1}$ has length $\omega_1$ and there are only countably many regular languages, there must exist a~bound on the levels of the Wadge hierarchy occupied by regular languages. The ordinal $\omega^\omega$ is a~natural candidate for the upper bound, as all the examples of the regular languages inside $\bdelta{2}$ exhaust exactly the first $\omega^\omega$ levels of the Wadge hierarchy.


\paragraph*{\bf Characterisations of transfinite levels.} Not only we don't know how many transfinite levels of the difference hierarchy are occupied, but there is no effective characterisation of any transfinite level known. Thus, one can ask for instance how to verify if the Wadge degree $\eta$ of a~regular tree language verifies $\omega\leq \eta< \omega\cdot 2$.

\paragraph*{\bf Higher in the hierarchy.} The operations of the algebras used in this paper seem to be suited exactly to the classes up to $\bdelta{2}$. However, one might imagine enriching their algebraic structure just a~bit in such a~way to cover one more level of the Borel hierarchy. This motivates the following problem.

\begin{prob}
Is there a~natural algebraic structure extending the algebras given in this paper which allow an~equational characterisation of the Borel class $\bdelta{3}$.
\end{prob}

\subsection{Acknowledgements} The authors would like to express their gratitude to Henryk Michalewski for a~number of insightful discussions on the subject. Additionally, the authors than anonymous referees for their helpful comments and suggestions.

\bibliographystyle{alpha}
\bibliography{biblioBP}

\end{document}